\documentclass{article}

\usepackage{microtype}
\usepackage{booktabs} 
\usepackage[normalem]{ulem}

\usepackage{amssymb}
\usepackage{mathtools}
\usepackage{amsthm}
\usepackage{dsfont}
\usepackage{graphicx}
\usepackage{thm-restate}
\usepackage{amsmath}
\usepackage{amsfonts}
\usepackage{algorithm}
\usepackage{algorithmic}
\usepackage{hyperref}
\usepackage[capitalize,noabbrev]{cleveref}
\usepackage{thmtools, thm-restate}
\usepackage{multirow}
\usepackage{minibox}
\usepackage{enumitem}
\usepackage{authblk}

\usepackage[margin=1in]{geometry}
\usepackage{parskip}
\usepackage{setspace}
\newcommand{\R}{\mathbb{R}}
\newcommand{\Rd}{\mathbb{R}^d}
\newcommand{\hRd}{\mathbb{R}^{\hat d}}
\newcommand{\hd}{\hat d}
\newcommand{\opt}{\textsc{opt}}
\newcommand{\eps}{\varepsilon}

\newcommand{\bc}{\mathfrak{c}}

\newcommand{\parent}{\textsc{Parent}}
\newcommand{\mval}{M_{\val}}

\DeclareMathOperator{\poly}{poly}
\DeclareMathOperator{\level}{level}
\DeclareMathOperator{\polylog}{polylog}
\DeclareMathOperator{\dist}{dist}
\DeclareMathOperator{\cost}{cost}
\DeclareMathOperator{\val}{Value}

\DeclareMathOperator{\argmin}{argmin}
\DeclareMathOperator{\isolated}{isolated}

\DeclareMathOperator{\len}{len}

\newcommand{\children}{C}

\newcommand{\calA}{\mathcal{A}} 
\newcommand{\calC}{\mathcal{C}} 
\newcommand{\calD}{\mathcal{D}} 

\newcommand{\calB}{\mathcal{B}}
\newcommand{\calG}{\mathcal{G}}
\newcommand{\calM}{\mathcal{M}}
\newcommand{\calN}{\mathcal{N}}
\newcommand{\calS}{\mathcal{S}}

\newcommand{\calX}{\mathcal{X}}

\newcommand{\mettuP}{\textsc{RecursiveGreedy}}
\newcommand{\mettuPModified}{\textsc{RecursiveGreedyModified}}

\newcommand{\sumEst}{\textsc{Sum}}
\newcommand{\sumNorm}{\textsc{SumNorm}}

\theoremstyle{plain}
\newtheorem{theorem}{Theorem}[section]
\newtheorem{infTheorem}{Informal Theorem}[section]

\newtheorem{lemma}[theorem]{Lemma}

\newtheorem{claim}[theorem]{Claim}

\theoremstyle{definition}
\newtheorem{definition}[theorem]{Definition}

\newtheorem{property}[theorem]{Property}

\usepackage{mdframed}
\usepackage{wrapfig}
\newcommand{\erclogowrapped}[1]{%
\setlength\intextsep{0pt}%
\begin{wrapfigure}[3]{r}{#1*\real{1.1}}%
\includegraphics[width=#1]{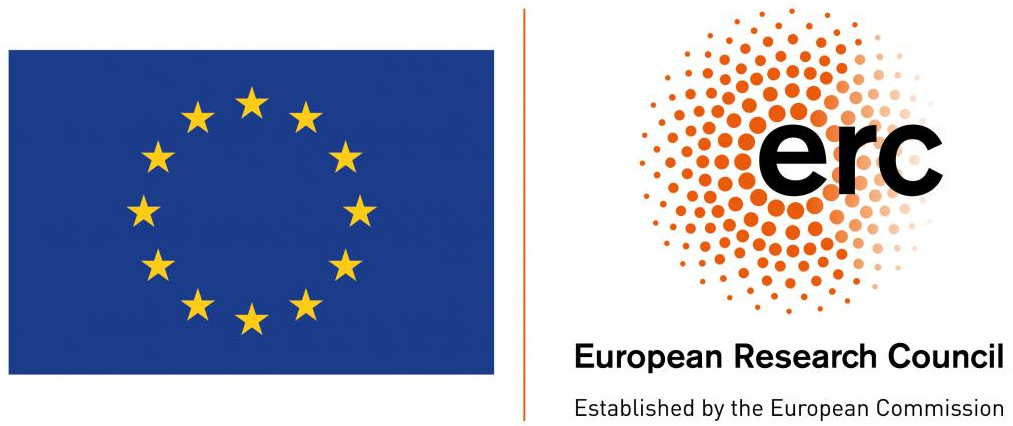}%
\end{wrapfigure}%
}

\newcommand{\lpar}{\left(}
\newcommand{\rpar}{\right)}
\newcommand{\lbra}{\left\{}
\newcommand{\rbra}{\right\}}
\newcommand{\lnor}{\left\|}
\newcommand{\rnor}{\right\|}
\newcommand{\lbrak}{\left[}
\newcommand{\rbrak}{\right]}

\newcounter{sideremark}

\date{}
\title{Making Old Things New: A Unified Algorithm for Differentially Private Clustering}
\author[1]{Max Dupré la Tour\footnote{Equal contribution.}}
\author[2]{Monika Henzinger$^*$}
\author[3]{David Saulpic$^*$}

\affil[1]{McGill University, Montreal, Canada}
\affil[2]{Institute for Science and Technology Austria (ISTA), Klosterneuburg, Austria}
\affil[3]{CNRS \& IRIF, Université Paris Cité, Paris, France}

\begin{document}
\maketitle

\begin{abstract}
    As a staple of data analysis and unsupervised learning, the problem of private clustering has been widely studied under various privacy models. Centralized differential privacy is the first of them, and the problem has also been studied for the local and the shuffle variation. In each case, the goal is to design an algorithm that computes privately a clustering, with the smallest possible error. 
    The study of each variation gave rise to new algorithms: the landscape of private clustering algorithms is therefore quite intricate.
    In this paper, we show that a 20-year-old algorithm can be slightly modified to work for any of these models. This provides a unified picture: while matching almost all previously known results, it allows us to improve some of them and extend it to a new privacy model, the continual observation setting, where the input is changing over time and the algorithm must output a new solution at each time step. 
\end{abstract}


\section{Introduction}
The massive, continuous and automatic collection of personal data by public as well as private organisations has raised privacy concerns, both legally  \cite{gdpr}, and in terms of citizens' demands \cite{nyt, nyt2}.
To address those concerns, formal privacy standards for algorithms were defined and developed, with the most prominent one being Differential Privacy~\cite{dworkDef}. This standard allows for a formal definition of privacy, enabling the development of algorithms with provable privacy guarantees.
Differentially private algorithms are now widely deployed. For instance, the U.S. Census Bureau uses them to release information from a private Census \cite{abowd2018us}, Apple employs them to collect data from its phone users~\cite{apple1, apple2}, and Google has developed an extensive library of private algorithms ready to be used~\cite{googleLib}.

The standard definition of \textit{centralized} Differential Privacy ensures that these algorithms ``behave roughly the same way'' on two databases differing by only a single element. The motivation behind this is to prevent inferring from the result whether a specific element is present in the database, thereby protecting against membership-inference attacks.
Stronger notions of privacy exist: most notably, the \textit{local} model, where the adversary observes not only the result but also all communications between a server and data owners. The communication must not reveal the presence of a specific element in the database. This guarantee is much stronger but comes at a price: achieving it significantly degrades the accuracy of the algorithm's answers. Therefore, some intermediate models have been defined, offering stronger privacy guarantees than the centralized model while maintaining better accuracy than in the local model.

Those guarantees are only valid for a \emph{static} database. However, real-life data often evolves over time, as seen in Apple's case, where personal data is collected and transferred daily. Therefore, a definition of privacy that accounts for such changes is necessary. This is formalized in the \emph{continual observation (or continual release) model} \cite{dwork2010differential}.

In this article, we consider one of the most common data-analysis and unsupervised learning primitive, namely $k$-means clustering.
This problem has been extensively studied under various notions of privacy: essentially each privacy model gives rise to a new algorithm, with a new and often delicate analysis. We describe this complex landscape in detail in \Cref{sec:previous}. 
However, all these privacy definitions share common ground: it seems possible that, instead of having specialized algorithms for each of them, one could identify the key properties of the private $k$-means problem and exploit them to design a unified algorithm. This is precisely the question we consider in this paper:

\begin{center}
        \minibox[c, frame]{ Is there a single clustering algorithm that could perform well in all privacy models?}
\end{center}

We answer this question positively for all privacy models in which clustering is known to be possible. 
This has a significant benefit: when studying a new privacy model (there are already more than 225 variations of differential privacy! \cite{DesfontainesP20}), there is a go-to algorithm to try that is likely to succeed. 
Indeed, we show that this algorithm can easily be made private under continual observation, which is the first such result for that model.

\subsection{Our Results}
We show that a 20-year-old greedy algorithm from Mettu and Plaxton \cite{MettuP00} can be easily made differentially private (DP). 
This algorithm provides a non-private approximation to the more general $(k,z)$-clustering problem. In this problem, the input data consists of a set of points $P$ in $\R^d$, and the goal is to find a set $S$ of $k$ points (the \emph{centers}) in order to minimize the cost, defined as 
$$\cost(P, S) = \sum_{p\in P} \min_{s \in S} \dist(p, s)^z.$$
We focus especially on the case where $z=2$, which is the popular $k$-means problem (while $z=1$ is $k$-median), and also provide results for general $z$. We denote the optimal cost for the $(k,z)$-clustering problem as $\opt_{k,z}$.

We show that a slight variation of the algorithm from Mettu and Plaxton is private, provided that one can privately solve a generalized version of the \textit{max summation} problem. Given a fixed, non-private set of balls in $\R^d$ and a private set of points $P \subset \R^d$, each point in $P$ contributes a value to each ball that contains it. The goal is to output a ball with approximately the maximum value; an algorithm for this problem has error $\theta$ if the absolute value of the difference between the actual maximum value and the one returned is at most $\theta$. This is a simplified version of the problem we need to solve, referred to as the \emph{generalized summation problem}, which we formally define in Section~\ref{sec:approx}.

Before stating our result, we note that the quality of the private $k$-means solution $S$ has to be measured by two parameters:  $S$ has \textit{multiplicative approximation} $M$ and \textit{additive error} $A$ when $\cost(S) \leq M \cdot \opt_{k,z} + A$. 
Since even the non-private problem is NP-hard to approximate within 1.06 \cite{Cohen-AddadSL22}, we must have $M > 1$ if we insist on a polynomial-time algorithm.
The privacy constraints enforce $A > 0$ as well; when the input is in the $d$-dimensional ball $B_d(0,\Lambda)$ in $\R^d$, \cite{chaturvediCentral} showed that $A$ has to be at least $k\sqrt{d} \cdot \Lambda^2$ for any $(\varepsilon, \delta)$ differentially-private mechanism. In light of this lower bound, we will assume $\Lambda=1$ in the following.

Our meta-theorem shows how to reduce the computation of a $k$-means solution to a \textit{repeated} application of an algorithm solving the max summation problem. 
To provide some intuition, a (perhaps too much) simplified version of the algorithm from Mettu and Plaxton repeats the following process $k$ times: select a ball with approximately maximum value, and remove all balls intersecting with the selected one. Thus, we can use an algorithm for repeated max summation as a black-box. Our main result relates the error of the max summation algorithm to the error of the clustering algorithm:

\begin{infTheorem}[see \Cref{thm:mpWithError} \Cref{thm:mainApprox}, and \Cref{lem:centralizedError}]\label{thm:main}
Let $\beta>0$.
    If one can solve privately the repeated max summation problem such that, with probability at least $2/3$, the error is $\theta$, then one can solve DP $k$-means such that, with probability $1-\beta$, either of the following guarantees is achievable:
    \begin{itemize}[noitemsep, topsep=0pt]
        \item multiplicative approximation $O(1)$ and additive error $\approx k \polylog(n/\beta) \cdot \theta$,
        \item or multiplicative approximation $w^*(1+\alpha)$ and additive error $\approx \sqrt{d} \poly(k, \log(n/\beta)) \cdot \theta$, where $w^*$ is the best non-private approximation ratio.\footnote{For this result, we actually need something slightly stronger than max summation, see \Cref{def:bucketSum}}
    \end{itemize} 
    For the more general $(k,z)$-clustering, in the second case the multiplicative approximation is $w^*(2^z+\alpha)$ and additive error $\approx \sqrt{d} \poly(k/\beta, \log(n)) \cdot \theta$.
\end{infTheorem}
To illustrate the above informal theorem, in  centralized $(\eps, \delta$)-DP we can use the exponential mechanism to solve the repeated max summation problem. This is formalized in \Cref{lem:centralizedError}, with $\theta = \sqrt{d} \polylog(n/\delta)/\eps$. 

We highlight a few features of our $k$-means results: in the first case, the additive error is optimal, as it matches the lower bound of \cite{chaturvediCentral}. In the second case, the multiplicative approximation is close to optimal, in the sense that it is almost as good as any polynomial-time non-private algorithm. Furthermore, even starting from an algorithm with constant probability of success, we show that the probability can be boosted arbitrarily high.

Our result is actually even stronger: it computes a solution not only to $k$-means, but to all $k'$-means for $k' \leq k$, with the same multiplicative and additive error guarantee as above.
This allows the use of the so-called elbow method to select the 'correct' value for $k$ without any further loss of privacy. We refer to \Cref{ap:elbow}.

We apply this meta-theorem to several different privacy settings, and state the bounds obtained in \Cref{table:results}. We present in this table the bounds for $(\eps, \delta)$-privacy (see \Cref{sec:prelim}); we address the particular case $\delta=0$ in \Cref{app:epsprivacy}.

To summarize our contribution, we match the previous bounds in almost all settings, and make improvements in several cases. For Local and Shuffle DP in one round, we improve exponentially the dependency in the probability for $k$-means and extend the results to $(k,z)$-clustering, which partially answers an open question from \cite{ChangG0M21}. In the MPC model, we improve the dependency in $k$ to get an optimal bound. Finally, we present the first result in the Continual Observation setting. We summarize our bounds in \Cref{table:results}, and discuss in greater detail the previous algorithms -- and why they do not work in full generality -- in \Cref{sec:previous}.

\begin{table*}
\caption{Comparison with the previous state-of-the-art for $(\eps, \delta)$ privacy. The success probability is $1-\beta$. $\alpha \in (0, 1/4]$ and $c > 0$ are precision parameters. For simplicity, dependency in $\log(1/\beta), 1/\eps, \log(1/\delta), \polylog(nd)$ and $\log \log T$ (for continual observation) are hidden, and the diameter is assumed to be $\Lambda = 1$. The notation $O_\alpha(1)$ is to insist that the constant hidden depends on $\alpha$ -- here it is $\log(1/\alpha)/\alpha^2$.}
\vskip 0.05in
    \label{table:results}
\begin{tabular}{l|c|c|r}
    Model & Approximation & Error & \\
    \hline
    \multirow{2}{*}{Centralized DP} & $w^*(1+\alpha)$ & $k^{O_\alpha(1)} + k\sqrt{d}$ & \cite{ghaziTight}, Cor.~\ref{cor:mainResApprox} \\
    & $O(1)$ & $k \sqrt{d}$ & \cite{chaturvediCentral}, Lem.~\ref{lem:centralizedError}\\
    \hline
    \multirow{4}{*}{Local DP} & $w^*(1+\alpha)$ & $\sqrt{n} \cdot \lpar (k/\beta)^{O_\alpha(1)} + k\sqrt{d}\rpar $ & $k$-means only, 1 round, \cite{ChangG0M21}\\
    & $w^*(1+\alpha)$ & $\sqrt{n} \cdot \lpar k^{O_\alpha(1)} + k\sqrt{d}\rpar$ & $k$-means only, 1 round, Cor.~\ref{cor:mainResApprox}\\
    & $w^*(2^z+\alpha)$ & $\sqrt{n} \cdot \lpar (k/\beta)^{O_\alpha(1)} + k\sqrt{d}\rpar$ & $(k,z)$-clustering, 1 round, Cor.~\ref{cor:mainResApprox}\\
    & $O(1/c)$ & $\sqrt{nd} \cdot k^{1+O_c(1)}$ &  $k$-means only, \cite{ChaturvediJN22}\\   
    \hline
     \multirow{3}{*}{Shuffle DP,  1 round} & $w^*(1+\alpha)$ & $\lpar (k/\beta)^{O_\alpha(1)} + k\sqrt{d}\rpar$ & $k$-means only, \cite{ChangG0M21}\\
    & $w^*(1+\alpha)$ & $ k^{O_\alpha(1)} + k\sqrt{d}$ & $k$-means only, Cor.~\ref{cor:mainResApprox}\\
    & $w^*(2^z+\alpha)$ & $ (k/\beta)^{O_\alpha(1)} + k\sqrt{d}$ & $(k,z)$-clustering, Cor.~\ref{cor:mainResApprox}\\
    \hline
   \multirow{3}{*}{MPC} & $w^*(1+\alpha)$ & $k^{O_\alpha(1)} + k\sqrt{d}$ & \cite{Cohen-AddadEMNZ22}, Thm.\ref{thm:mpc}\\
   & $O(1)$ & $k^{2.5} + k^{1.01}\sqrt{d}$ & \cite{Cohen-AddadEMNZ22}\\
    & $O(1)$ & $k\sqrt{d}$ & \Cref{thm:mpc}\\
    \hline
    Continual  & $w^*(1+\alpha)$ & $ (k^{O_\alpha(1)} + k\sqrt{d})\log^{1.5}(T)$ & $k$-means only, Cor.~\ref{cor:mainResApprox}\\ 
    ~~observation& $w^*(2^z+\alpha)$ & $ \lpar (k/\beta)^{O_\alpha(1)} + k\sqrt{d}\rpar \log^{1.5}(T)$ & $(k,z)$-clustering, Cor.~\ref{cor:mainResApprox}

\end{tabular}
\end{table*}

\subsection{Brief Overview}
To show the first point of \Cref{thm:main}, we rely on the algorithm from Mettu and Plaxton. We reinterpret this algorithm, introducing some key changes: first, to make it private, and second, to enable implementation based on an algorithm for the generalized summation problem. To provide some intuition, this algorithm iteratively chooses cluster centers, intuitively by selecting smaller and smaller regions that are far away from any center previously selected, based on the region's ``value" (a proxy for the contribution to the cost). 
For any $k$, the first $k$ centers form a constant-factor approximation to $(k,z)$-clustering.
It turns out that we can repeatedly use a generalized summation algorithm to compute those ``values": we show that, if we have a private generalized summation algorithm with an error of $\theta$, then we can solve $(k,z)$-clustering with additive error $k \theta$ (see \Cref{sec:MPalgo}).

We start this paper by formalizing some general building blocks for private clustering in \Cref{sec:coating}, namely techniques that can be used to simplify the input and the problem. We show how to perform all of them based only on estimating the size of the clusters and other related quantities.
This includes for $k$-means (a) a dimension-reduction technique; (b) a technique that improves the approximation ratio from $O(1)$ to almost $w^*$, the best non-private approximation ratio; (c) a new way of boosting the success probability for $k$-means. (a) and (b) are well-known for $k$-means, we extend them for the general $(k,z)$-clustering problem.

We combine these techniques in \Cref{sec:approx} to obtain the near-optimal approximation factor that we presented in \Cref{table:results}. 
For this, we show that the max summation problem can be solved privately using histograms to estimate the value of each ball, and then to select the maximum estimated value. 
This leads directly to a novel private algorithm for the centralized model that is also much simpler than prior algorithms.
However, applying the other building blocks (a)-(c) requires estimating the size of each cluster. 
This would be doable with histograms \emph{if the clusters were known a priori and remained fixed throughout the algorithm}; however,  the clusters depend on the input data and are not known a priori, and, thus, computing their size cannot be reduced to a simple histogram query. 
To solve this issue, we introduce a structural result on the shape of clusters: 
we show that each cluster is the disjoint union of a small number of \emph{pre-determined} sets. Therefore, to estimate the size of each cluster, it is enough to apply a general summation algorithm on the pre-determined sets, and combine the results on those.

Finally, we present another option: instead of histograms, one can use the exponential mechanism in \Cref{sec:error} to show that the max summation problem can be solved (in some privacy settings) with a very tiny $\theta$, resulting in a near-optimal additive error.

\subsection{Privacy Models}\label{sec:prelim}
As it is common in the differential privacy literature (see e.g. \cite{dwork2014algorithmic}), we will assume our input is given as a multiset, as formalized in \Cref{app:defPriv}. 

\paragraph{Central Differential Privacy:}
We will use the formalism of \cite{dwork2014algorithmic}. A \emph{dataset} is a multiset $P$ of points of a universe $X$. We say that two datasets $P, P'$ are \emph{neighboring} when they differ by a single point, namely $\sum_{x\in X}|P(x)-P'(x)| = 1$.
We say that a mechanism $\mathcal{M}$ is $(\varepsilon, \delta)$-\emph{differentially private} if for any two neighboring  datasets $P, P'$ and any set $S$, we have:
\[\mathbb{P}(\mathcal{M}(P) \in S) \leq \exp(\varepsilon)\cdot \mathbb{P}(\mathcal{M}(P') \in S) + \delta.\] 
We say that an algorithm is $\varepsilon$-\emph{differentially private} if it is $(\varepsilon, 0)$-differentially private. To emphasize the difference to other models, we will refer to this privacy model as \emph{central differential privacy}. 

In this paper, we will also study the following other privacy models, deferring the formal definitions to \Cref{app:defPriv}.
\noindent \paragraph{Local Model}~\cite{kasiviswanathan2011can}: In the local model, there is no trusted central server with access to the entire raw database. Instead, we have $n$ individuals, each with one data point. Two inputs are  adjacent if the data of a single user changes. The \textit{transcript} of an algorithm is the sequence of messages exchanged between clients and the server: an algorithm is \emph{$(\eps, \delta)$-local differentially-private (LDP)} if the transcript is $(\eps, \delta)$-DP. 
In this paper, we focus on the local privacy model with a single round of communication from clients to the server, also known as the non-interactive model.

\paragraph{Shuffle Model}~\cite{bittau2017prochlo}:
Similarly to the local model, we have $n$ individuals, each with one data point. However, in the shuffle model, a trusted intermediary comes into play between the individuals and the server: the \emph{shuffler}. The shuffler gathers the messages from the individuals and shuffles them randomly before sending them to the server, preventing the server from attributing a specific message to a particular individual.\footnote{The original motivation behind this is that the random shuffling can be done via secure cryptographic protocols.} Only the transcript of interactions between the server and the shuffler has to be DP. 

\paragraph{Continual Observation Model}~\cite{dwork2010differential}: In the continual observation model, the input is not static but evolves over time. The algorithm is given a stream of updates (insertion or deletion) to its dataset, one per time step, and outputs a solution for the input so far at each time step. Two streams are \emph{(event-level) neighboring} if they differ by a single update. The algorithm is $(\eps, \delta)$-DP under continual observation if the algorithm mapping a stream to a sequence of outputs is $(\eps, \delta)$-DP.

\paragraph{Massively Parallel Computing Model (MPC)}: The MPC model is a model for distributed, scalable computation -- not necessarily private. The input is initially split among several machines, each of them having local memory sub-polynomial in the total database size ($n^\kappa$, for some fixed $\kappa \in (0, 1)$). The machines can send and receive messages from other machines, but the message length cannot exceed the machine's local memory. As opposed to the Local Model, the messages exchanged don't have to be private: the algorithm is $(\eps, \delta)$-private if its output is $(\eps, \delta)$-DP.

\section{General Building Blocks for Private Clustering} \label{sec:coating}
In this section, we present several techniques used in the literature as preprocessing or postprocessing steps to simplify the task of computing a private clustering. We also extend some of these techniques to allow for greater generality. We will not fixate on a specific privacy model in order to present the results in a modular way. The lemmas in this section will apply to any privacy model, assuming that we are given a partition of the space into $k$ subsets $S_1,\dots,S_k$ corresponding to a clustering, and that we can estimate for all $i$, $|S_i \cap P|$, $\sum_{p \in P \cap S_i} p$, and $\sum_{p \in P \cap S_i} |p|_2$. We will prove the existence of private algorithms to compute such a partition and the corresponding estimation in the next sections.

\begin{property}\label{assumption:coating}
We say that three sequences $(n_i)_{1 \leq i \leq k}$, $(\sumEst_i)_{1 \leq i \leq k}$, $(\sumNorm_i)_{1 \leq i \leq k}$ verify the \Cref{assumption:coating} for a partition of the space $\R^d = S_1\cup \dots \cup S_k$ with error parameter $e \geq 0$ if:
\begin{itemize}
\item $|n_i - |P \cap S_i|| \leq e$.
\item $\|\sumEst_i - \sum_{p \in P \cap S_i} p\|_2 \leq e$.
\item $|\sumNorm_i - \sum_{p \in P \cap S_i} \|p\|_2| \leq e$.
\end{itemize}
\end{property}

\subsection{Reducing the dimension}\label{sec:dimred}

Dimension-reduction techniques based on the Johnson-Lindenstrauss lemma allow the projection of the dataset onto $\hd = O\lpar z^4 \cdot \log (k/\beta) \alpha^{-2}\rpar$ dimensions, such that with probability $1-\beta/2$ the clustering cost is preserved up to a $(1\pm \alpha)$ factor. We call $\pi(P)$ the projected and rescaled dataset. We provide a more detailed description in \Cref{ap:dimred} and concentrate here on the main challenge, which is to "lift up" the solution: given a clustering of the projected dataset, how can we compute centers in the original space?

Our main "lifting" technique is modular, as it merely requires approximating the size of each cluster and the sum of the points inside each cluster. Following the analysis of \cite{ghaziTight}, by applying a standard concentration bound, it holds with probability at least $1-\beta/2$ that all the points of the projected and rescaled dataset $\pi(P)$ lie within the ball $B_{\hd}(0, \sqrt{2\log(n/\beta)})$. Using a private algorithm $\calA$, we compute a solution to $(k,z)$-clustering $\calC$ of $\pi(P)$, with multiplicative approximation $M$ and additive error $A$. The set of centers $\calC$ induces a \emph{private} partition $\hat{S_1}, \dots, \hat{S_k}$ of $\mathbb{R}^{\hd}$, defined by the Voronoi diagram of $\calC$, and a \emph{private} partition $S_1, \dots, S_k$ of the original space $\mathbb{R}^{d}$, defined as the preimage $\pi^{-1}(S_1, \dots, S_k)$.

The natural way of defining centers in $\R^d$ to lift the partition would be to take the $k$ optimal centers $\mu_z$ for the $(1,z)$-clustering of each $S_i \cap P$. We say that the \emph{cost induced} by the partition $S_1,\dots,S_k$ is $\sum_{i=1}^k \cost(P\cap S_i, \mu_z(P\cap S_i))$. To lift privately a partition, we will use an approximation of each average $\mu_2(S_i \cap P) = \frac{\sum_{p \in P_i} p}{|P_i|}$ by the quantity $\frac{\sumEst_i}{n_i}$ where $(n_i)_{1\leq i \leq k}, (\sumEst_i)_{1\leq i \leq k}$ are two sequences computed \emph{privately} and verifying \Cref{assumption:coating} for the partition $S_1, \dots , S_k$. This is formalized in the next lemma.

\begin{restatable}{lemma}{liftingViaHist}\label{lem:liftingViaHist}
Let $\hd = O(z^4 \cdot \log (k/\beta) \alpha^{-2})$, and let $\pi(P)$ be the projected and rescaled dataset in $\R^{\hd}$. Suppose that we are given a partition $\hat{S_1}, \dots, \hat{S_k}$ of $\R^{\hd}$ that induces a $(M,A)$-approximation of the $(k,z)$-clustering problem on $\pi(P)$.

Let $S_1,\dots,S_k$ be the partition of $\R^d$ that we obtain by taking the preimage $\pi^{-1}$ of $\hat{S_1}, \dots, \hat{S_k}$. Assume that we have two sequences $(n_i)$ and $(\sumEst_i)$ verifying \Cref{assumption:coating} for this partition, and consider the set of centers $S = \lbra \frac{\sumEst_1}{n_1}, \dots, \frac{\sumEst_k}{  n_k}\rbra$.
        The following holds with probability $1-\beta$:

    In the case of $k$-means $(z=2)$, we have $\cost(P, S) \leq (1+\alpha)M \cdot \opt_{k,2}(P) + \polylog n  \cdot A + O(ke)$. 
    For general $(k,z)$-clustering, we have instead 
    $\cost(P, S) \leq 2^z(1+\alpha)M \cdot \opt_{k,z}(P) +  \polylog n \cdot A + O(k e)$.
\end{restatable}

The previous result for $k$-means is a direct application of the dimension reduction results, but we introduce here the generalization to $(k,z)$-clustering. This relies on the following new lemma:
\begin{restatable}{lemma}{meanMed}\label{lem:meanMed}
    Let $P$ be a multiset of points in $\R^d$ with optimal center $\mu_z$ for $(1,z)$-clustering  and optimal $(1,2)$-clustering solution $\mu = \mu_2$. Then, 
    \[\sum_{p\in P} \|p-\mu\|^z \leq 2^z \sum_{p \in P} \|p-\mu_z\|^z.\]
\end{restatable}

\subsection{Boosting the multiplicative approximation}\label{sec:boostApprox}
Let $w^*$ be the best approximation ratio achievable non-privately for $(k,z)$-clustering.
An observation of \cite{dpHD} allows the conversion of any private algorithm with constant approximation into an algorithm with 
approximation almost $w^*$ while increasing the additive error: any $O(1)$-approximation of $(k', z)$-clustering, with $k' = \alpha^{-O(d)} k \log (n/\alpha)$, is an $\alpha$-approximation for $(k, z)$-clustering. 
If we can compute privately such a solution, one can convert it into a true solution for $k$-means with approximation $w^*(1+\alpha)$, while preserving the additive error (see \Cref{lem:bicriteria}).
This yields the following result:

\begin{lemma}[Theorem 4 in \cite{dpHD}]\label{lem:boostApprox}
Suppose we are given a private algorithm $\calA$ for $(k,z)$-clustering that has multiplicative approximation $M$ and additive error $A(k, d)$. Suppose we can privately compute a sequence $(n_i)$ verifying \Cref{assumption:coating} for the Voronoi diagram of the centers output by $\calA$. And let $w^*$ be the best approximation ratio achievable non-privately for $(k,z)$-clustering. 
    Then, for any $1/4 > \alpha > 0$, there is a private algorithm that computes a solution for $(k,z)$-clustering with cost at most  

    $w^*(1+\alpha) \cdot \opt_{k,z} +  O(A(k', d)+ k'e)$, where $k' = (\alpha/M)^{-O(d)} k \log (nM/\alpha)$. 
\end{lemma}

This can be combined with dimension reduction from the previous section, i.e.,~we apply this lemma in dimension $\log(k)/\alpha^2$, resulting in $k' = k^{O_\alpha(1)} \log (n/\alpha)$. 

\subsection{Boosting the success probability}
Assume we are given a private algorithm $\calA$ that computes with probability $2/3$ (or any constant $>1/2$) a $(M, A)$-approximation to $(k,z)$-clustering. To increase the success probability to $1-\beta$, the standard technique is to run $\log(1/\beta)$ copies of $\calA$ in parallel, and select (privately) the one with the best output. 
This can be easily implemented using the exponential mechanism, which requires computing the cost of each solution. While this is possible in many settings, in some settings (e.g., Local DP or continual observation), it is not obvious how to do so without ``losing'' too much privacy. 

In the particular case of $k$-means, we show how this can be done using mere histogram queries. This relies on the following new lemma: the $k$-means cost of a cluster can be expressed as a function of the points of a cluster, regardless of the location of its center.
 
\begin{restatable}{lemma}{costVariance}\label{lem:costVariance} 
    For any multiset $E$,
    \vspace{-1.5em}
    \begin{gather*}
        \cost(E, \mu(E)) =\sum_{p\in E} \|p\|_2^2 - \frac{\lnor \sum_{p \in E} p\rnor_2^2}{|E|}.
    \end{gather*}
\end{restatable}
    \vspace{-1.5em}
Together with the algorithms of the second point of \Cref{assumption:coating}, this yields the following corollary:

\begin{restatable}{corollary}{boostProba}\label{cor:boostProba}
Let $S_1,\dots,S_k$ be a partition of $B_d(0,\Lambda)$, and suppose we can privately compute the three sequences $(n_i)_{1 \leq i \leq k}$, $(\sumEst_i)_{1 \leq i \leq k}$, $(\sumNorm_i)_{1 \leq i \leq k}$ verifying the \Cref{assumption:coating}. Then one can estimate the $k$-means cost induced by the partition up to an additive error $O\lpar k e \rpar$.

    Therefore, given an $(\eps, \delta)$-DP algorithm $\calA$ with multiplicative approximation $M(\eps, \delta)$ and additive error $A(\eps, \delta)$ that succeeds with probability $2/3$, there exists a private algorithm that succeeds with probability $1-\beta$ with multiplicative approximation $M(\eps / \log(1/\beta), \delta/ \log(1/\beta))$ and additive error 
    $A(\eps/ \log(1/\beta), \delta/ \log(1/\beta)) + O(ke)$.
\end{restatable}

\section{The Greedy Algorithm}\label{sec:MPalgo}
In this section, we present the original non-private algorithm by Mettu and Plaxton~\cite{MettuP00}. This algorithm takes as input a multiset $P$ of points of $X$, and outputs a sequence $\lpar c_1, \dots, c_n \rpar$ in such a way that, for all $k \leq n$, the set $C_k = \{c_1, \dots, c_k\}$ approximates the optimal $(k,z)$-clustering cost.
\footnote{This particular variant, initially referred to as ``online" by Mettu and Plaxton, has undergone a name change in more recent literature to ``incremental." This change was made to avoid any confusion with other, more common meanings of ``online."}
\begin{definition}
Given a ball $B = B(x,r):= \{y\in X, \dist(x,y) \leq r\}$, the \emph{value} of $B$ is $\val(B):= \sum_{p\in B\cap P} (r-\dist(x,p))^z$.

A \emph{child} of a ball $B(x,r)$ is any ball $B(y,r/2)$, where $y\in P$ and $\dist(x,y) \leq 10 r$.

For any point $x\in X$ and a set of centers $C$, let $\isolated(x,C)$ denote the ball $B(x,\dist(x,C)/100)$ if $C$ is not empty; and $B(x,\max_{y\in P}d(x,y))$ if $C = \emptyset$. Intuitively, this corresponds to very large ball centered at $x$ that is far away from any center of $C$.\footnote{In those definitions, we chose the scalar constants $2, 10, 100$ for convenience: the whole analysis can be parameterized more carefully in order to optimize the approximation ratio. We opted for simplicity.}
\end{definition}
The algorithm is a simple greedy procedure, that starts with $C = \emptyset$ and repeats $n$ times the following steps:  start with the ball $\isolated(x, C)$ with maximum value over all $x \in P$ (with ties broken arbitrarily), and as long as this ball has more than one child (i.e. as long that there are at least two distinct points of the input $P$ "close" to the ball) replace it with the child with maximum value. Let $x$ be the center of the last chosen ball: add $x$ to $C$, and repeat. We give the pseudo-code of the procedure in the appendix (see \cref{alg:mp1}). We call $C_k$ the solution $C$ after $k$ repetitions of the loop. 

\begin{theorem}{\cite{MettuP00}}\label{thm:mp1}
    For any fixed $z$ and for all $k$, the cost of $C_k$ is a $O(1)$-approximation of the optimal $(k,z)$-clustering cost.  
\end{theorem}
This algorithm is not private for two reasons: the balls in the sequences $\sigma_i$ are centered at points of the input, and the value is computed non-privately. To ensure privacy, we introduce the following key modifications to \cref{alg:mp1}, resulting in Algorithm~\ref{alg:mp}. 

\begin{restatable}{algorithm}{MetPla}
\caption{$\mettuPModified(P,\theta)$}
\label{alg:mp}
\begin{algorithmic}[1]
\STATE{$\calA = \calB$, ($\calA$ is the set of \emph{available} balls)}
\STATE{Let $C_0 = \emptyset$}
\FOR{$i$ from $0$ to $n-1$}
\STATE{Let $\sigma_i$ denote the singleton sequence $(B)$ where $B\in \calA$ has a maximum value up to $\theta$ among the available balls.}
\WHILE{The last element of $\sigma_i$ has level less than $\lceil \log n \rceil$}
\STATE{Select a child with maximum value up to $\theta$ of the last element of $\sigma_i$, and append it to $\sigma_i$.}
\ENDWHILE
\STATE{$c_{i+1}$ is the center of the last ball of $\sigma_i$}
\STATE{$C_{i+1} = C_i \cup \{c_{i+1}\}$}
\STATE{Remove from $\calA$ the balls forbidden by $c_{i+1}$}
\ENDFOR
\end{algorithmic}
\end{restatable}

\paragraph{Modification 1}: We fix a set of balls $\calB$ independent from the data: $\calB$ consists essentially of balls  centered at the points of a fine-grained discretization of the space based on \emph{nets} instead of input data. An $\eta$-net of $X$ is a subset $\calN \subset X$ satisfying the two following properties: \emph{Packing:} for all distinct $x,x'\in \calN$ we have $\dist(x,x')>\eta$, and  \emph{Covering:} for all $x\in X$ we have $\dist(x,\calN)\leq \eta$ (see \Cref{app:net}).

For any $i\in \lbra1,\dots,\lceil \log n\rceil\rbra$, we let $\calN_i$ be a $2^{-i}/2$-net of $B(0,1)$. An element of $\calN_i$ is a \emph{net point} of \emph{level} $i$. We use the notation $\level(x)$ to denote the level of a net point $x$.\footnote{In order to uniquely define the level of a net point, we make the assumption that all $\calN_i$ are disjoint.}
We further define  $\calB_i$ to be the set of balls of radius $2^{-i}$ centered in the points of $\calN_i$, and $\calB = \calB_1 \cup \dots \cup \calB_{\lceil \log n \rceil}$. 

We change the while condition to stop when we reach a ball of the last level $\lceil \log n \rceil$: at that level, the region considered is so small that any point in it is a good center.

To replace the (non-private) concept of isolated balls $\isolated(x,C)$, we introduce the following new definitions: We say that a ball $B(x,2^{-i}) \in \calB_i$ is \emph{forbidden} by a center $c \in B(0, 1)$ if $\dist(x, c) \leq  100 \cdot 2^{-i}$. A ball is \emph{available} when it is not forbidden. Note that the radius under consideration is now only dependent on the level $i$ and no longer on the distance of $x$ to the centers.
This is a new point of view on the algorithm of \cite{MettuP00}, that we believe sheds a new, simplified light on the algorithm.
\paragraph{Modification 2}: The algorithm does not have access to the actual values of the balls. Instead of selecting the ball with the maximum value, 
the modified algorithm chooses any ball whose value is sufficiently close to the maximum value, with an additive error of $\theta$. Depending on the privacy model, this selection will be performed using either the standard exponential mechanism (\Cref{lem:exp_mechanism}), or using private noisy values instead of the actual values.

\begin{definition}
Let $\theta>0$.
    For any set of balls $S$, we say that a ball $B \in S$ has a maximum value in $S$ \emph{up to} $\theta$, if $\val(B) \geq \max_{A \in S} \val(A) - \theta$.
\end{definition}

\begin{restatable}{theorem}{mp}\label{thm:mpWithError}
     For any fixed $z>0$, for all integer $k>0$ and for all $\theta \geq 1$, the centers $C_k$ produced by \Cref{alg:mp} have cost at most $O(1) \cdot \opt_{k,z} + O(k\theta)$.
\end{restatable}

Thus \Cref{alg:mp} and \Cref{thm:mpWithError} reduces private $(k,z)$-clustering to the problem of privately maintaining the ball with maximum value. We show how to solve the latter problem in the next section.

\section{Near Optimal Multiplicative Approximation from Histograms}
\label{sec:approx}
The combination of results from \Cref{sec:MPalgo} and \Cref{sec:coating} shows that, to compute a solution to $(k,z)$-clustering privately (in any privacy model), one merely needs to compute privately the values of each ball (for \Cref{thm:mpWithError}), the size of each cluster, and $\sum_{p\in P_i} p$ (for \Cref{assumption:coating}). 

All this can be done through \textit{generalized bucketized vector summation} (which we abbreviate to \emph{generalized summation})  which is a generalization of histograms: given sets, standard histograms  estimate the size of each set, while generalized summation  allows for the summation over more complex functions of the sets' elements: 

\begin{restatable}{definition}{bucketSum}[Generalized bucketized vector summation]\label{def:bucketSum}
    Let $S_1, ..., S_m \subset X$ be fixed subsets of the universe $X$, and for $i= 1, \dots, m$, let $f_i : X \rightarrow B_D(0,1)$ be a function, where $B_D(0,1)$ is the $D$-dimensional unit ball.
    Given a set $P \subseteq X$,
    a generalized bucketized vector summation algorithm computes an $m$-dimensional vector $v$, whose $i$-th entry $v_i \in \mathbb{R}^D$ for $i=1,...,m$,  estimates $\sum_{p \in P\cap S_i} f_i(p)$. 
    It has additive error $\eta$ if 
    $\eta \ge \max_{i \in [1,m]} \lnor v_i - \sum_{p \in P\cap S_i} f_i(p)\rnor_2.$ 

    One important parameter for the performance of such an algorithm is the maximum frequency $b$ such that each item $x \in X$ appears in at most $b$ sets: $b = \max_{x} \left| \lbra i: x \in S_i \rbra \right|$.  
\end{restatable}

For example, generalized summation for $f_i(p) = 1$ estimates the number of input points $|S_i \cap P|$ in each $S_i$ ; in that case, the summation problem can be solved by a histogram algorithm, where each column/item corresponds to a set $S_i$. 
We show in \Cref{app:GenHist} how to turn an algorithm for private histogram into one for private generalized summation.
As an illustration, we prove in \Cref{lem:histoContinual} that there exists a private algorithm for generalized summation under continual observation, with additive error roughly $\eta = b \sqrt{D} \log(m DT)^{3/2} \log (b \log T)/ \eps$. \cite{ChangG0M21} shows an equivalent result for Local and Shuffle DP.

\paragraph{Using generalized summation for clustering.}To privately estimate the value of each ball, we apply a generalized summation algorithm with the sets being the balls, and for the $j$-th ball $B(x_j, r_j)$, $f_j(e) = r_j - \dist(x_j, e)$ when $e \in B(x_j, r_j)$, $f_j(e) = 0$ otherwise. 
If the private generalized summation has additive error $\eta$, then \Cref{thm:mpWithError} implies that our algorithm is differentially private with multiplicative approximation $O(1)$ and additive error $k \eta$. 
Thus, we only need to show the existence of a generalized summation algorithm in the different privacy models, and bound the additive error $\eta$.

We additionally \emph{boost the multiplicative approximation and success probability} using the techniques from \Cref{sec:coating}.
This requires estimating the size of each cluster (see \Cref{assumption:coating}). While this is very similar to a generalized summation problem, there is a key difficulty: \textit{the clusters are not fixed in advance}, as opposed to the sets $S_1, ..., S_m$! Therefore, it is not possible to directly apply a generalized summation algorithm for the cluster size estimation.

To resolve this issue, we introduce in the next subsection a decomposition of the space that allows the expression of each cluster as a union of \emph{pre-determined} sets. Our approach is then to apply private generalized summation to these sets, to estimate the size of each cluster (and all other quantities required for the results in \Cref{sec:coating}).

\subsection{Structured clusters}\label{sec:structuredCluster}
We show in this section how to transform the clusters into sets with enough structure, to be able to estimate their size using the generalized summation algorithm.

\begin{restatable}{lemma}{structClusters}[See formal statement in \Cref{lem:structClustersFormal}]
    \label{lem:structClusters}
Fix $\alpha >0$. There exists a family of sets $\calG = \lbra G_1,..., G_m\rbra$, independent of $P$, efficiently computable, with $m = n^{O(d)}$, and with the following properties. 
(1) Any point from $B_d(0,1)$ is part of at most $O(\log n)$ sets $G_i$. 
(2) Any possible cluster can be transformed to consist of the union of at most $k\cdot \alpha^{-O(d)} \log(n)$ sets $G_i$. (3)
This transformation preserves the cost up to an additive error of $\alpha\cost(P, \calC) + \frac{9}{\alpha}$. 
\end{restatable}

Using this lemma, we can describe more formally our algorithm, which is made of two main parts: first, an algorithm that computes an $(O(1), A(k,d))$-approximation -- with an exponential dependency in $d$ in the additive error; second, an algorithm that uses the techniques from \cref{sec:coating} with a generalized histogram and \cref{lem:structClusters} to boost the multiplicative approximation and reduce the dependency in the dimension.

\begin{algorithm}
\caption{Private-Clustering-low-Dimension($P, k, \eps, \delta$)}
\label{alg:lowDim}
\begin{algorithmic}[1]
\STATE \textbf{Input}: A dataset $P$, a number of clusters $k$
\STATE Compute the values of balls with the private generalized summation algorithm. 
\STATE Run the algorithm from \Cref{sec:MPalgo} to compute a partition $(S_1,...,S_k)$
\STATE \textbf{Output:} $(S_1, ..., S_k)$
\end{algorithmic}
\end{algorithm}
\begin{algorithm}
\caption{Private-Clustering($P, k, \eps, \delta$)}
\label{alg:main}
\begin{algorithmic}[1]
\STATE \textbf{Input}: A dataset $P$, a number of clusters $k$
\STATE Project $P$ onto a lower-dimensional space $\R^{\hd}$, with $\hd = O(z^4 \log(k)\alpha^{-2})$ as in Lem.~\ref{lem:liftingViaHist}.
\STATE Use Lem.~\ref{lem:boostApprox} to compute a partition $(S'_1, ..., S'_k)$: the clustering algorithm is first \Cref{alg:lowDim} (with parameter $k'$ as defined in Lem.~\ref{lem:boostApprox}), then transform the partition with Lem.~\ref{lem:structClusters}, and estimate the $n_i$ with the private generalized summation algorithm.
\STATE Transform the partition $(S'_1, ..., S'_k)$ into a structured one with Lem.~\ref{lem:structClusters}, and estimate $n_i, \sumEst_i, \sumNorm_i$ for this partition using the private generalized summation algorithm.
\STATE Use \Cref{lem:liftingViaHist} to transform the partition into centers in $\R^d$, and Lem.~\ref{cor:boostProba} to compute the cost of the clustering.
\STATE \textbf{Output:} the centers, with the cost of the clustering.
\end{algorithmic}
\end{algorithm}
This algorithm provides a good approximation with probability $2/3$. The probability can then be easily boosted using \cref{cor:boostProba}.

Using this structural lemma, we get our first theorem: if for a given privacy model there exists a private algorithm for generalized summation with small additive error, then we can apply \Cref{thm:mpWithError}, \Cref{lem:structClusters} and the results of \Cref{sec:coating} to get an efficient algorithm for $(k,z)$-clustering with arbitrarily high success probability.
\begin{restatable}{theorem}{mainApprox}[See full statement in \Cref{lem:mainApproxFormal}]\label{thm:mainApprox}
    Fix a privacy model where there exists a private generalized bucketized vector summation such that with probability at least $2/3$ the additive error is $A(m,b, D)$, where $m$ is the number of sets, $b$ is the maximal number of sets that any given item is contained in, and $D$ is the dimension of the image of the $f_i$. 
    Assume that the error $A$ is non-decreasing in all parameters.

    Then, for any $\alpha, \beta>0$ there is a private algorithm for $k$-means  clustering of points in $\Rd$ that, with  probability at least $1-\beta$, computes a solution with multiplicative approximation $w^*(1+\alpha)$ and additive error 
$A(m, k^{O_\alpha(1)} \log(n), d) \cdot m \polylog(n) \log(1/\beta)$, with $m = n^{O_\alpha(\log(k))}$. 

    There is also a private algorithm for $(k,z)$-clustering of points in $\Rd$ that,  with  probability at least $2/3$, computes a solution with multiplicative approximation $w^*(2^z+\alpha)$ and additive error 
    $A(m, k^{O_\alpha(1)} \cdot \log(n), d) \cdot m \polylog(n)$. 
\end{restatable}

Together with the private generalized summation algorithms that we present in the appendix, this leads to the results presented in \Cref{table:results}:

\begin{restatable}{corollary}{mainResApprox}\label{cor:mainResApprox}
    There are algorithms for $k$-means with multiplicative approximation $w^*(1+\alpha)$ that achieve with probability at least $1-\beta$:
    \begin{itemize}[noitemsep, topsep=0pt]
        \item additive error $\sqrt{nD} \cdot  k^{O_\alpha(1)} / \eps \cdot \polylog(n) \cdot \log(1/\beta)$ in the local $(\eps, \delta)$-DP model with one round of communication,
        \item additive error $\sqrt{D} \cdot k^{O_\alpha(1)} / \eps \cdot  \polylog(nD/\delta) \cdot \log(1/\beta)$ in the shuffle $(\eps, \delta)$-DP model with one round of communication,
        \item additive error $\sqrt{D} \cdot k^{O_\alpha(1)} / \eps \cdot \polylog(n) \cdot \log^{1.5}(DT) \log \log(T) \sqrt{\log(1/\delta)} \cdot \log(1/\beta)$ in $(\eps, \delta)$-DP under continual observation.
    \end{itemize}
    
    The same additive guarantee hold for $(k,z)$-clustering, with multiplicative approximation $w^*(2^z+\alpha)$ and probability $2/3$.
\end{restatable}

\section{Low Additive Error via Exponential Mechanism}\label{sec:error}

A different view on the algorithm of \Cref{alg:mp} is the following: "the algorithm iteratively selects the ball with the largest value. To initialize the $i$-th sequence $\sigma_i$, the algorithm selects a ball out of the set of available balls (line 4 of \Cref{alg:mp}), and to extend the sequence it chooses a child of the current ball (line 6).
Instead of using a histogram to compute privately the value of the balls, one can use the exponential mechanism on the non-private values to make each decision. This approach works directly for the Centralized DP setting -- though it requires a delicate privacy proof -- and its adaptation to the MPC setting requires new ideas.

\subsection{Centralized DP}
Our first result is an application to centralized DP. For this, we directly apply the algorithm of \Cref{sec:MPalgo}, where all balls are chosen with the exponential mechanism: out of a set $S$ of balls, select a ball $B$  randomly with probability proportional to $\exp(\eps' \val(B))$. 
Properties of the exponential distribution ensure that the value of each ball selected will be, with high probability, off from the optimal choice by an additive error $O(\log(|S|n)/\eps')$.
In that case, \Cref{thm:mpWithError} ensures that the algorithm computes a solution with constant multiplicative approximation and additive error $\frac{k\log(k)\log(n/\delta) + k\sqrt{d}}{\eps'}$. We formalize this in \Cref{lem:centralizedError}, and show in \Cref{lem:central-exp-priv} that  when run with $\eps' = \frac{\eps}{4\log(n/\delta)}$, the algorithm  is $(\eps, \delta)$-DP.

\subsection{MPC}
Implementing \Cref{alg:mp} in a parallel setting is not obvious because of the iterative nature of the greedy choices: selecting the $i$-th center influences dramatically the subsequent choices, and it is not clear how to parallelize those choices in fewer than $k$ parallel rounds.

Thus, we present a slightly different algorithm building heavily on the analysis of \Cref{alg:mp}. 
We introduce two key modifications. First, we show that instead of iteratively selecting smaller and smaller balls each level can be treated \emph{independently}. Our observation is that it is enough to run \Cref{alg:mp} for all levels of the decomposition, without the inner loop of line 5--7.  
The second modification allows us to run a \textit{relaxed} version of this greedy algorithm, tailored to a MPC setting: we extract the key elements of the proof of \Cref{thm:mpWithError} to allow some slack in the implementation.

To formalize this, we refine the definition of availability: a ball $B(x, 2^{-i})$ is {\em available at scale $D$ for a set of centers $C$} if for all $ c \in C, \dist(x, c) \geq D 2^{-i}$ 
(in \Cref{alg:mp}, we defined $D = 100$). 
For each level $\ell \in {1,\dots,\lceil \log n \rceil}$, our algorithm will compute a set $C_\ell$ of $2k$ centers, and returns $C=\cup C_\ell$.

\begin{restatable}{lemma}{MPCMP}\label{lem:MPCMP}
         Suppose there is a constant $c_\calA$ such that, for any level $\ell$, (1) the distance between two distinct centers of $C_\ell$ is greater than $3\cdot 2^{-\ell}$; and (2) for any  center $c\in C_\ell$, the value of the ball $B(c,2^{-\ell})$ is greater (up to an additive error $\theta$) than the value of any ball of $\calB_\ell$ available at scale $c_\calA$ from $C_\ell$.
        Then, $\cost(P,C) \leq O(1) \cdot \opt_{k,z} + O(k\theta)$.
\end{restatable}

Therefore, each level can be treated independently, and 
we are left with the following problem: given the set $\calB_\ell$ of all balls at level $\ell$, 
compute \textit{privately} a set of $2k$ balls that are (1) far apart and (2) have larger value than any available ball at scale $c_\calA$ from them. Doing so yields: 
\begin{theorem}\label{thm:mpc}
    Suppose each machine has memory $k^{O(1)} n^\kappa$, for $\kappa \in (0,1)$. Then, there is an $(\eps,\delta)$-private MPC algorithm, running in $O(1)$ rounds of communication that computes with probability $1-\beta$ a solution to $(k,z)$-clustering with cost at most $O(1)\opt_{k,z} + k\sqrt{d} / \eps \cdot \polylog(n, 1/\delta, 1/\beta)$. 
    Alternatively, the algorithm can compute a solution with cost $w^*(1+\alpha)\opt_{k,z} + k^{O_\alpha(1)} /\eps \cdot \polylog(n, 1/\delta, 1/\beta)$.
\end{theorem}

If all balls of a single level were to fit in a single machine, we would be done, as a simple greedy algorithm 
would achieve the guarantee of \Cref{lem:MPCMP}: until $2k$ balls have been selected, select privately (using the exponential mechanism) the largest value ball that is available at scale $D = 3$ from the previously selected ones. 
This would ensure all selected balls are at distance at least $3 \cdot 2^{-\ell}$, and that they have value greater than any available ball at scale $3$. However, the memory is limited to $k^{O(1)} n^\kappa$ for some fixed $\kappa \in (0,1)$, and there are $n^{O(d)}$ many balls.

We propose instead a "merge-and-reduce" algorithm. Each machine gets assigned a set of balls, and runs the greedy algorithm restricted to those balls, with scale parameter $D_0 = 3 \cdot 2^{1/\kappa}$. The outcome in each machine is $2k$ balls. 
Then, one can merge the results from $n^\kappa$ machines as follows: centralize the $2k n^\kappa$ selected balls on a single machine, and run the greedy only on those balls, with parameter $D_1 = D_0 / 2$. 
Therefore, in one communication round, we can reduce the number of selected balls by a factor $n^\kappa$. Repeating this process with $D_i = D_{i-1} / 2$ yields $2k$ balls after $1/\kappa$ rounds.
We show in \Cref{lem:mpcGuarantee} that they satisfy the guarantee of \Cref{lem:MPCMP} with $c_\calA = 2^{1/\kappa+1} \cdot 3$.
This concludes the proof of \Cref{thm:mpc}.

\section{Conclusion}
We show that a simple greedy algorithm unifies the problem of private clustering. We believe this unification will help further research and that our algorithm can be used as a baseline for new privacy models, as we illustrate with continual observation.

To achieve this, we revisited an old greedy algorithm from \cite{MettuP00}. Despite the complexity of its analysis, this algorithm seems to be quite flexible and powerful. Through our novel analysis, we aim to enhance the understanding of this algorithm, thereby encouraging its application in new contexts.

\section*{Acknowledgments}
\erclogowrapped{5\baselineskip}Monika Henzinger:  This project has received funding from the European Research Council (ERC) under the European Union's Horizon 2020 research and innovation programme (Grant agreement No. 101019564) and the Austrian Science Fund (FWF) grant DOI 10.55776/Z422, grant DOI 10.55776/I5982, and grant DOI 10.55776/P33775 with additional funding from the netidee SCIENCE Stiftung, 2020–2024.

This work was partially done while David Saulpic was at the Institute for Science and Technology, Austria (ISTA). David Sauplic has received funding from the European Union’s Horizon 2020 research and innovation programme under the Marie Sklodowska-Curie grant agreement No 101034413. 

\bibliographystyle{alpha}
\bibliography{biblio}

\newcommand{\etalchar}[1]{$^{#1}$}
\begin{thebibliography}{{WWD}16}

\bibitem[Abo18]{abowd2018us}
John~M Abowd.
\newblock The us census bureau adopts differential privacy.
\newblock In {\em Proceedings of the 24th ACM SIGKDD International Conference on Knowledge Discovery \& Data Mining}, pages 2867--2867, 2018.

\bibitem[App16]{apple2}
Apple.
\newblock {WWDC 2016 Keynote}, june 2016.
\newblock \url{https://www.theverge.com/2016/6/17/11957782/apple-differential-privacy-ios-10-wwdc-2016}.

\bibitem[BBC{\etalchar{+}}19]{BecchettiBC0S19}
Luca Becchetti, Marc Bury, Vincent Cohen{-}Addad, Fabrizio Grandoni, and Chris Schwiegelshohn.
\newblock Oblivious dimension reduction for \emph{k}-means: beyond subspaces and the johnson-lindenstrauss lemma.
\newblock In Moses Charikar and Edith Cohen, editors, {\em Proceedings of the 51st Annual {ACM} {SIGACT} Symposium on Theory of Computing, {STOC} 2019, Phoenix, AZ, USA, June 23-26, 2019}, pages 1039--1050. {ACM}, 2019.

\bibitem[BBGN19]{balle2019privacy}
Borja Balle, James Bell, Adri{\`a} Gasc{\'o}n, and Kobbi Nissim.
\newblock The privacy blanket of the shuffle model.
\newblock In {\em Advances in Cryptology--CRYPTO 2019: 39th Annual International Cryptology Conference, Santa Barbara, CA, USA, August 18--22, 2019, Proceedings, Part II 39}, pages 638--667. Springer, 2019.

\bibitem[BCJM21]{balcer2021connecting}
Victor Balcer, Albert Cheu, Matthew Joseph, and Jieming Mao.
\newblock Connecting robust shuffle privacy and pan-privacy.
\newblock In {\em Proceedings of the 2021 ACM-SIAM Symposium on Discrete Algorithms (SODA)}, pages 2384--2403. SIAM, 2021.

\bibitem[BEM{\etalchar{+}}17]{bittau2017prochlo}
Andrea Bittau, {\'U}lfar Erlingsson, Petros Maniatis, Ilya Mironov, Ananth Raghunathan, David Lie, Mitch Rudominer, Ushasree Kode, Julien Tinnes, and Bernhard Seefeld.
\newblock Prochlo: Strong privacy for analytics in the crowd.
\newblock In {\em Proceedings of the 26th symposium on operating systems principles}, pages 441--459, 2017.

\bibitem[BKS17]{BeameKS17}
Paul Beame, Paraschos Koutris, and Dan Suciu.
\newblock Communication steps for parallel query processing.
\newblock {\em J. {ACM}}, 64(6):40:1--40:58, 2017.

\bibitem[CEM{\etalchar{+}}22]{Cohen-AddadEMNZ22}
Vincent Cohen{-}Addad, Alessandro Epasto, Vahab Mirrokni, Shyam Narayanan, and Peilin Zhong.
\newblock Near-optimal private and scalable $k$-clustering.
\newblock In {\em NeurIPS}, 2022.

\bibitem[CGKM21]{ChangG0M21}
Alisa Chang, Badih Ghazi, Ravi Kumar, and Pasin Manurangsi.
\newblock Locally private k-means in one round.
\newblock In Marina Meila and Tong Zhang, editors, {\em Proceedings of the 38th International Conference on Machine Learning, {ICML} 2021, 18-24 July 2021, Virtual Event}, volume 139 of {\em Proceedings of Machine Learning Research}, pages 1441--1451. {PMLR}, 2021.

\bibitem[CJN22]{ChaturvediJN22}
Anamay Chaturvedi, Matthew Jones, and Huy~Le Nguyen.
\newblock Locally private k-means clustering with constant multiplicative approximation and near-optimal additive error.
\newblock In {\em Thirty-Sixth {AAAI} Conference on Artificial Intelligence, {AAAI} 2022, Thirty-Fourth Conference on Innovative Applications of Artificial Intelligence, {IAAI} 2022, The Twelveth Symposium on Educational Advances in Artificial Intelligence, {EAAI} 2022 Virtual Event, February 22 - March 1, 2022}, pages 6167--6174. {AAAI} Press, 2022.

\bibitem[CKL22]{Cohen-AddadSL22}
Vincent Cohen{-}Addad, {Karthik {C. S.}}, and Euiwoong Lee.
\newblock Johnson coverage hypothesis: Inapproximability of k-means and k-median in $\ell_p$-metrics.
\newblock In Joseph~(Seffi) Naor and Niv Buchbinder, editors, {\em Proceedings of the 2022 {ACM-SIAM} Symposium on Discrete Algorithms, {SODA} 2022, Virtual Conference / Alexandria, VA, USA, January 9 - 12, 2022}, pages 1493--1530. {SIAM}, 2022.

\bibitem[CNX21]{chaturvediCentral}
Anamay Chaturvedi, Huy~L. Nguyen, and Eric Xu.
\newblock Differentially private k-means via exponential mechanism and max cover.
\newblock In {\em Thirty-Fifth {AAAI} Conference on Artificial Intelligence, {AAAI} 2021, Thirty-Third Conference on Innovative Applications of Artificial Intelligence, {IAAI} 2021, The Eleventh Symposium on Educational Advances in Artificial Intelligence, {EAAI} 2021, Virtual Event, February 2-9, 2021}, pages 9101--9108. {AAAI} Press, 2021.

\bibitem[CSS11]{ChanSS11}
T.{-}H.~Hubert Chan, Elaine Shi, and Dawn Song.
\newblock Private and continual release of statistics.
\newblock {\em {ACM} Trans. Inf. Syst. Secur.}, 14(3):26:1--26:24, 2011.

\bibitem[Dem21]{nyt2}
David Deming.
\newblock Balancing privacy with data sharing for the public good, 2021.
\newblock \url{https://www.nytimes.com/2021/02/19/business/privacy-open-data-public.html}.

\bibitem[DMNS06]{dworkDef}
Cynthia Dwork, Frank McSherry, Kobbi Nissim, and Adam Smith.
\newblock Calibrating noise to sensitivity in private data analysis.
\newblock In Shai Halevi and Tal Rabin, editors, {\em Theory of Cryptography}, pages 265--284, Berlin, Heidelberg, 2006. Springer Berlin Heidelberg.

\bibitem[DNPR10]{dwork2010differential}
Cynthia Dwork, Moni Naor, Toniann Pitassi, and Guy~N Rothblum.
\newblock Differential privacy under continual observation.
\newblock In {\em Proceedings of the forty-second ACM symposium on Theory of computing}, pages 715--724, 2010.

\bibitem[DP20]{DesfontainesP20}
Damien Desfontaines and Bal{\'{a}}zs Pej{\'{o}}.
\newblock Sok: Differential privacies.
\newblock {\em Proc. Priv. Enhancing Technol.}, 2020(2):288--313, 2020.

\bibitem[DR14]{dwork2014algorithmic}
Cynthia Dwork and Aaron Roth.
\newblock The algorithmic foundations of differential privacy.
\newblock {\em Foundations and Trends{\textregistered} in Theoretical Computer Science}, 9(3--4):211--407, 2014.

\bibitem[EGS03]{evfimievski2003limiting}
Alexandre Evfimievski, Johannes Gehrke, and Ramakrishnan Srikant.
\newblock Limiting privacy breaches in privacy preserving data mining.
\newblock In {\em Proceedings of the twenty-second ACM SIGMOD-SIGACT-SIGART symposium on Principles of database systems}, pages 211--222, 2003.

\bibitem[EMZ23]{concurrent}
Alessandro Epasto, Tamalika Mukherjee, and Peilin Zhong.
\newblock Differentially private clustering in data streams.
\newblock 2023.

\bibitem[FHO21]{FichtenbergerHO21}
Hendrik Fichtenberger, Monika Henzinger, and Wolfgang Ost.
\newblock Differentially private algorithms for graphs under continual observation.
\newblock In Petra Mutzel, Rasmus Pagh, and Grzegorz Herman, editors, {\em 29th Annual European Symposium on Algorithms, {ESA} 2021, September 6-8, 2021, Lisbon, Portugal (Virtual Conference)}, volume 204 of {\em LIPIcs}, pages 42:1--42:16. Schloss Dagstuhl - Leibniz-Zentrum f{\"{u}}r Informatik, 2021.

\bibitem[GGK{\etalchar{+}}21]{ghazi2021power}
Badih Ghazi, Noah Golowich, Ravi Kumar, Rasmus Pagh, and Ameya Velingker.
\newblock On the power of multiple anonymous messages: Frequency estimation and selection in the shuffle model of differential privacy.
\newblock In {\em Annual International Conference on the Theory and Applications of Cryptographic Techniques}, pages 463--488. Springer, 2021.

\bibitem[GKL03]{gupta}
Anupam Gupta, Robert Krauthgamer, and James~R. Lee.
\newblock Bounded geometries, fractals, and low-distortion embeddings.
\newblock In {\em 44th Symposium on Foundations of Computer Science {(FOCS} 2003), 11-14 October 2003, Cambridge, MA, USA, Proceedings}, pages 534--543. {IEEE} Computer Society, 2003.

\bibitem[GKM20]{ghaziTight}
Badih Ghazi, Ravi Kumar, and Pasin Manurangsi.
\newblock Differentially private clustering: Tight approximation ratios.
\newblock In Hugo Larochelle, Marc'Aurelio Ranzato, Raia Hadsell, Maria{-}Florina Balcan, and Hsuan{-}Tien Lin, editors, {\em Advances in Neural Information Processing Systems}, 2020.

\bibitem[GKM{\etalchar{+}}21]{ghazi2021differentially}
Badih Ghazi, Ravi Kumar, Pasin Manurangsi, Rasmus Pagh, and Amer Sinha.
\newblock Differentially private aggregation in the shuffle model: Almost central accuracy in almost a single message.
\newblock In {\em International Conference on Machine Learning}, pages 3692--3701. PMLR, 2021.

\bibitem[GLM{\etalchar{+}}10]{GuptaLMRT10}
Anupam Gupta, Katrina Ligett, Frank McSherry, Aaron Roth, and Kunal Talwar.
\newblock Differentially private combinatorial optimization.
\newblock In Moses Charikar, editor, {\em Proceedings of the Twenty-First Annual {ACM-SIAM} Symposium on Discrete Algorithms, {SODA} 2010, Austin, Texas, USA, January 17-19, 2010}, pages 1106--1125. {SIAM}, 2010.

\bibitem[GSZ11]{GoodrichSZ11}
Michael~T. Goodrich, Nodari Sitchinava, and Qin Zhang.
\newblock Sorting, searching, and simulation in the mapreduce framework.
\newblock In Takao Asano, Shin{-}Ichi Nakano, Yoshio Okamoto, and Osamu Watanabe, editors, {\em Algorithms and Computation - 22nd International Symposium, {ISAAC} 2011, Yokohama, Japan, December 5-8, 2011. Proceedings}, volume 7074 of {\em Lecture Notes in Computer Science}, pages 374--383. Springer, 2011.

\bibitem[Gue19]{googleLib}
Miguel Guevara.
\newblock Enabling developers and organizations to use differential privacy, 2019.

\bibitem[HM06]{Har-PeledM06}
Sariel Har{-}Peled and Manor Mendel.
\newblock Fast construction of nets in low-dimensional metrics and their applications.
\newblock {\em {SIAM} J. Comput.}, 35(5):1148--1184, 2006.

\bibitem[IKI94]{InabaKI94}
Mary Inaba, Naoki Katoh, and Hiroshi Imai.
\newblock Applications of weighted voronoi diagrams and randomization to variance-based \emph{k}-clustering (extended abstract).
\newblock In {\em Proceedings of the Tenth Annual Symposium on Computational Geometry, Stony Brook, New York, USA, June 6-8, 1994}, pages 332--339, 1994.

\bibitem[JRSS21]{jain2021price}
Palak Jain, Sofya Raskhodnikova, Satchit Sivakumar, and Adam Smith.
\newblock The price of differential privacy under continual observation.
\newblock {\em arXiv preprint arXiv:2112.00828}, 2021.

\bibitem[KLN{\etalchar{+}}11]{kasiviswanathan2011can}
Shiva~Prasad Kasiviswanathan, Homin~K Lee, Kobbi Nissim, Sofya Raskhodnikova, and Adam Smith.
\newblock What can we learn privately?
\newblock {\em SIAM Journal on Computing}, 40(3):793--826, 2011.

\bibitem[KSV10]{KarloffSV10}
Howard~J. Karloff, Siddharth Suri, and Sergei Vassilvitskii.
\newblock A model of computation for mapreduce.
\newblock In Moses Charikar, editor, {\em Proceedings of the Twenty-First Annual {ACM-SIAM} Symposium on Discrete Algorithms, {SODA} 2010, Austin, Texas, USA, January 17-19, 2010}, pages 938--948. {SIAM}, 2010.

\bibitem[MMR19]{MakarychevMR19}
Konstantin Makarychev, Yury Makarychev, and Ilya~P. Razenshteyn.
\newblock Performance of johnson-lindenstrauss transform for \emph{k}-means and \emph{k}-medians clustering.
\newblock In Moses Charikar and Edith Cohen, editors, {\em Proceedings of the 51st Annual {ACM} {SIGACT} Symposium on Theory of Computing, {STOC} 2019, Phoenix, AZ, USA, June 23-26, 2019}, pages 1027--1038. {ACM}, 2019.

\bibitem[MP00]{MettuP00}
Ramgopal~R. Mettu and C.~Greg Plaxton.
\newblock The online median problem.
\newblock In {\em 41st Annual Symposium on Foundations of Computer Science, {FOCS} 2000, 12-14 November 2000, Redondo Beach, California, {USA}}, pages 339--348. {IEEE} Computer Society, 2000.

\bibitem[MT07]{expmech}
Frank McSherry and Kunal Talwar.
\newblock Mechanism design via differential privacy.
\newblock In {\em 48th Annual IEEE Symposium on Foundations of Computer Science (FOCS'07)}, pages 94--103, 2007.

\bibitem[Ngu20]{dpHD}
Huy~L. Nguyen.
\newblock A note on differentially private clustering with large additive error.
\newblock {\em CoRR}, abs/2009.13317, 2020.

\bibitem[SK18]{StemmerK18}
Uri Stemmer and Haim Kaplan.
\newblock Differentially private k-means with constant multiplicative error.
\newblock In Samy Bengio, Hanna~M. Wallach, Hugo Larochelle, Kristen Grauman, Nicol{\`{o}} Cesa{-}Bianchi, and Roman Garnett, editors, {\em Advances in Neural Information Processing Systems 31: Annual Conference on Neural Information Processing Systems 2018, NeurIPS 2018, December 3-8, 2018, Montr{\'{e}}al, Canada}, pages 5436--5446, 2018.

\bibitem[{The}20]{nyt}
{The Editorial Board of the New York Times}.
\newblock Privacy cannot be a casualty of the coronavirus, april 2020.
\newblock \url{https://www.nytimes.com/2020/04/07/opinion/digital-privacy-coronavirus.html}.

\bibitem[Uni18]{gdpr}
European Union.
\newblock General data protection regulation, 2018.
\newblock \url{https://gdpr.eu/}.

\bibitem[War65]{warner1965randomized}
Stanley~L Warner.
\newblock Randomized response: A survey technique for eliminating evasive answer bias.
\newblock {\em Journal of the American statistical association}, pages 63--69, 1965.

\bibitem[{WWD}16]{apple1}
{WWDC}.
\newblock Engineering privacy for your users, june 2016.
\newblock \url{https://developer.apple.com/videos/play/wwdc2016/709/}.

\end{thebibliography}

\onecolumn
\appendix
\section{Extended Introduction}
\subsection{Previous works and their limits}\label{sec:previous}
Most of the works in central DP rely on a process that iteratively selects $k$ "dense" balls and removes the points lying in that ball \cite{StemmerK18, ghaziTight, chaturvediCentral}. While such a process can be optimal both for the multiplicative approximation \cite{ghaziTight} and the additive error \cite{chaturvediCentral}, its iterative nature makes it hard to implement in other models. In particular, because the points of dense balls are removed, this method appears very difficult to implement in the MPC model with fewer than $k$ rounds or  in the continual observation model.

The Local-DP algorithm of \cite{ChangG0M21} requires only histogram queries and uses them to build a private \emph{coreset} -- i.e., a set such that the cost of any $k$-means solution in the coreset is roughly the same as in the original dataset. The additive error for it is inherently $2^d$ (where $d$ can be assumed to be $O(\log k)$), and it seems hard to use this algorithm to achieve tight additive error.
Furthermore, to achieve a success probability of $1-\beta$, the dimension $d$ has to be $\log(k/\beta)$: therefore, the additive error $2^d$ reduces to $\poly(k/\beta)$, and it is not clear how to boost the success probability using a single round of communication.

The MPC model differs from the others (e.g. Local DP) as only the output needs to be private, not the communication. Here, the key is to find a summary of the data that fits on a single machine and allows for computing a solution. This summary does not need to be completely private, which is different from the algorithm for Local DP (and from our algorithm for continual observation). 
In particular, the MPC algorithms of \cite{Cohen-AddadEMNZ22} compute a summary consisting of weighted input points. This is clearly not private, and cannot be applied to Local DP or other models. Furthermore, the lowest additive error they achieve is a suboptimal $(k^{2.5} + k^{1.01}d) \polylog(n)$. 

We describe in slightly more detail the MPC algorithms of \cite{Cohen-AddadEMNZ22} to explain how ours can be seen as a simplification. Both algorithms for low approximation ratio or additive error work in two steps as follows. First, embed the metric into a quadtree: at each level, the space is partitioned by a grid centered at a random location. Then, the algorithm selects, at each level, the $k$ grid cells containing most points. Their centers provide an initial bi-criteria solution, with small additive error but multiplicative $\poly(d)$. To reduce this multiplicative error, the authors use a sampling technique from the coreset literature. Starting from this $\poly(d)$ approximation, they construct $\poly(k, \log n)$ many disjoint instances with the following guarantees. First, each instance has small size, such that it fits on a single machine, allowing the use of any private non-MPC algorithm. Second, as the instances are disjoint, taking the union of the solutions for each of them results in a good bi-criteria solution: it has too many centers, but a good cost. It is then standard to reduce the number of centers, as we describe in greater detail later.

Our algorithm is very similar to the first step, with the key difference being that at each level we allow overlap between the grid cells, instead of a partition. We show that this is directly an $O(1)$-approximation, which avoid the need of building the subinstances.

\textbf{Parallel work on Continual Observation.} Epasto, Mukherjee, and Zhong~\cite{concurrent} recently released similar work on private $k$-means under dynamic updates. 
Their results are incomparable with ours: their main focus was on designing an algorithm using low memory, which they manage to do in \emph{insertion-only} streams. 
For those streams, they can achieve a smaller additive error at the cost of a larger constant multiplicative error, see their Corollary 14.

\subsection{Technical lemma and Definitions}\label{app:def}
To work with powered distances, we will often use a modification of the triangle inequality:
\begin{lemma}[Triangle Inequality for Powers, see e.g. \cite{MakarychevMR19}]
\label{lem:weaktri}
Let $a,b,c$ be three arbitrary points in a metric space with distance function $\dist$ and let $z$ be a positive integer. Then for any $\varepsilon>0$
$$\dist(a,b)^z \leq (1+\varepsilon)^{z-1} \dist(a,c)^z + \left(\frac{1+\varepsilon}{\varepsilon}\right)^{z-1} \dist(b,c)^z $$ 
\end{lemma}

\paragraph{Net Decomposition:}\label{app:net}
A $\eta$-net of $X$ is a subset $Z \subset X$ satisfying the following two properties:
\emph{Packing:} for all distinct $z,z'\in Z$ we have $\dist(z,z')>\eta$, and  \emph{Covering:} for all $x\in X$ we have $\dist(x,Z)\leq \eta$.
It is well known that the unit ball in $\R^d$ has a $\eta$-net 
for all $\eta > 0$ \cite{gupta}.
Standard net constructions allow to efficiently \emph{decode}: given any point $p \in B_d(0,1)$, one can find in time $\frac{2^{O(d)}}{r}$ all net points at distance at most $r$ of $p$. By the packing property of a $\eta$-net, there can be at most $\frac{2^{O(d)}}{r}$ such points, and therefore decoding is done efficiently. We refer to \cite{Har-PeledM06} for more details, and to the notion of \emph{Efficiently List-Decodable Covers} of \cite{ghaziTight}.

\subsection{About the Elbow Method}\label{ap:elbow}
The elbow method is a heuristic used to determine the ``right" number of clusters. The idea is that the curve of the $(k,z)$-clustering cost as a function of $k$ is decreasing, and typically admits an ``elbow" point, where the gain of adding one new cluster significantly decreases. The elbow method chooses this value of $k$ as the true number of clusters. 

To plot this curve, one therefore needs to compute the $(k,z)$-clustering cost for all value of $k$: \Cref{thm:mpWithError} precisely allows that, as for all $k'$ the set of centers $C_{k'}$ computed with \Cref{alg:mp1} is a good approximation of optimal $k'$-means solution. 
Provided a general summation algorithm, our analysis shows how to run privately the elbow method as a post-processing. Indeed,
we can evaluate the cost of all those solutions without leaking any privacy, using \Cref{lem:costVariance} shows that we can evaluate the cost of this solution using the structure of clusters and the general summation algorithm (as for boosting the probabilities). 
Lifting the result up in $\Rd$ can also be done as post-processing, using the result of the general summation algorithm.

\section{About Differential Privacy}

\subsection{Formal Definitions}\label{app:defPriv}
In this section, we define formally the different models of privacy considered. 

All of those use multisets, that are standardly defined as follows.
A \emph{multiset} $P$ of points in a \emph{universe} $X$ is a function $P:X \rightarrow \mathbb{N}$. We will use set notation and vocabulary for multisets. In particular for any function $f$ we will denote by $\sum_{p\in P} f(p):= \sum_{x\in X} P(x) \cdot f(x)$, and for any subset $Y$ of $X$, $P\cap Y$ is the multiset $x \mapsto P(x) \cdot \mathds{1}_{x \in Y}$.

We let $\eps > 0$ and $\delta \in [0, 1]$ be the privacy parameters.
\paragraph{Local Differential Privacy}: 
This model was formally defined in \cite{kasiviswanathan2011can}, although having old roots in the study of randomized response (see e.g. \cite{warner1965randomized, evfimievski2003limiting}). We follow here the presentation of \cite{dwork2014algorithmic}.
In the local model of privacy, there are $n$ users. User $i$ holds a datapoint $p_i$, and the server is interacting with users via \emph{local randomizers}: which are simply algorithms  taking as input a database of size 1. 

The server runs an algorithm that interacts with users via those local randomizers: this algorithm is $(\eps, \delta)$-DP in the local model if, for any user $i$, the mechanism outputting the results of all local-randomizer run by user $i$ is $(\eps, \delta)$-DP.

In this paper, we consider non-interactive setting, where the server interacts only once with each user.

\paragraph{Shuffle Differential Privacy}: 
This model is an attempt to provide the best of the two worlds between centralized and local DP: achieving the utility of centralized while preserving a stronger privacy as in local. 
This stronger privacy is obtained as the intermediate shuffling step can be implemented via secure cryptographic protocols, therefore strongly protecting individual's privacy.

In the shuffle model of privacy, there are $n$ users and one \emph{shuffler}. Each user holds a datapoint $p_i$, and the interaction goes in rounds. At each round, users run local-randomizers, and communicate the results to the shuffler. The shuffler releases the results, shuffled uniformly at randomized. The server then receives those messages and goes on with the computation.
The procedure is $(\eps, \delta)$-DP in the shuffle model when the set of output of the shuffler is $(\eps, \delta)$-DP.

This model was introduced by \cite{bittau2017prochlo}, and studied e.g. for aggregation \cite{balle2019privacy, ghazi2021differentially}, histograms and heavy hitters \cite{ghazi2021power}, or counting distinct elements \cite{balcer2021connecting}. Recently clustering \cite{ChangG0M21}.

\paragraph{Continual Observation Model}: 
Let $\sigma$ be a stream of updates to a multiset (the dataset): $\sigma(t)$ corresponds to either the addition of a new item to the dataset, the removal of one, or a ``no operation" time (where nothing happens to the dataset).  
Two streams are \textit{neighboring} if they differ by the addition of a single item at a single time step (and possibly its subsequent removal). 
This is the so-called \emph{event-level} privacy \cite{dwork2010differential}.
Let $\calM$ be a randomized mechanism that associates a stream $\sigma$ to an output $\calM(\sigma) \in \calX$.
For an $\eps \in \R^+$, $\calM$ is $(\eps, \delta)$-DP under continual observation if $\calM$ is $(\eps, \delta)$-DP. Formally, for any pair of neighboring streams $\sigma, \sigma'$, and any possible set of outcomes 
$\calS \subseteq \calX$,  it holds that 
$$\Pr\left[ \calM(\sigma) \in \calS\right]  \leq \exp(\eps) \Pr\left[ \calM(\sigma') \in \calS\right] + \delta.$$

\paragraph{The Massively Parellel Computation Model}
This model is not a privacy model per se, but a restricted model of computation (regardless of privacy). The input data is spread on many machines, each having sublinear memory (typically $n^\kappa$, for $\kappa \in (0, 1)$). The computation happens in round. 
At each round, the machines can send and receive messages one message from each machine, but the message length cannot exceed the memory $n^\kappa$. In the setting we consider, at the end of the computation each machine knows the solution and can output it. This model was formalized in \cite{KarloffSV10, GoodrichSZ11, BeameKS17} and is now a standard for analyzing parallel algorithms designed for large-scale environments.
The privacy constraint is that outputing this solution is $(\eps, \delta)$-DP. Note that there is no privacy constraints on how exactly the communication happens: we assume the exchanges between machines are trusted. This is different from the Local privacy model, where both the output and the exchanges have to be private.

\subsection{The Exponential Mechanism}

The \emph{exponential mechanism} is a standard tool in the literature, that allows the following: given a set of $n$ items with private value, select privately an item with value almost maximum. Formally, we have:

\begin{lemma}[Exponential Mechanism \cite{expmech}]\label{lem:exp_mechanism}
 Let $\mathcal{R}$ be some arbitrary range set, and let $f:\mathbb{N}^X \times \mathcal{R} \rightarrow \mathbb{R}$. Let $\Delta_f:= max_{r\in \mathcal{R}} max_{\|P-P'\|_1 = 1} |f(P,r)-f(P',r)|$ be the \emph{$\ell_1$-sensitivity} of $f$. The \emph{exponential mechanism} $\mathcal{M}_E(P,f,\varepsilon)$ outputs an element $r\in \mathcal{R}$ with probability proportional to $exp(\frac{\varepsilon\cdot f(P,r)}{2\cdot \Delta_f})$.
\begin{itemize}
    \item The exponential mechanism is $\varepsilon$-differentially private.
    \item For all $t>0$, with probability at least $ 1-e^{-t}$ we have: $$\max_{r\in \mathcal{R}} f(P,r) - f\lpar P,\mathcal{M}_E(P,f,\varepsilon)\rpar \leq \frac{2\cdot\Delta_f}{\varepsilon}\cdot \ln \lpar|\mathcal{R}|+t\rpar   $$
\end{itemize}
\end{lemma}

\subsection{Private generalized bucketized vector summation algorithm}\label{app:GenHist}

Recall first the definition of general summation.
\bucketSum*

Note that we use a somewhat unusual error definition: The error is the maximum over all sets $S_i$ of the $\ell_2$-norm of the error within each set.

We now show how to use a private histogram algorithm to create a private algorithm for generalized summation. Private histogram is a well-studied problem, which is the general summation problem with $b=1$ and $\forall i, p, f_i(p) = 1$.
We will show first how to turn a generalized summation problem with $b>1$ into a generalized summation problem with $b=1$, and will show afterwards  how to solve the latter in different settings.

\subsubsection[From b=1 to b >1]{From $b=1$ to $b \geq 1$}
We start by presenting a simple lemma, that uses group privacy to reduce the general summation to the case where each item is in a single set. 

\begin{lemma}\label{lem:buildGenHist}
Let $\delta, \beta > 0$, and $1 \geq \eps > 0$.
    Assume we are given an $(\varepsilon, \delta)$-differentially private algorithm for generalized summation with $b=1$ that has, with probability at least $1-\beta$,  additive error $\eta(m, |P|, \eps, \delta, \beta)$ and runs in time $T(m, |P|, \eps, \delta)$, where $m$ is the total number of sets and $|P|$ is the size of the input.
Then there exists  an $(\varepsilon, \delta)$-differentially private algorithm for generalized summation   with $b>1$,  that has, with probability at least $1-\beta$, additive error $\eta(m, b|P|, \frac{\eps}{b}, \frac{\delta}{3b}, \beta)$ and runs in time $T(m, b|P|, \frac{\eps}{b}, \frac{\delta}{3b}  , \beta)$.
    \end{lemma}
\begin{proof}
    We transform the input to the generalized summation problem with  $b>1$ to a generalized summation problem with $b=1$ by creating a new input as follows: 
    for each input data $x$  in the original input, $x$ is duplicated $|\{i:x \in S_i\}|$ times, and each copy is added to the new input as part of a different single set $S_i$, followed by $b - |\{i:x \in S_i\}|$ all-zero inputs. Thus each data item in the original input is replaced by exactly $b$ data items in the new input.
    
    As a consequence the size of the new input is exactly $b|P|$, and each item belongs to a single set.
    Running the generalized summation algorithm  for $b=1$ with parameters $\eps', \delta', \beta$, it is $(\eps',\delta')$-differentially private and has, with probability at least
    $1-\beta$, an additive error of $\eta(m, b|P|, \eps', \delta', \beta)$ and runs in time $T(m, b|P|, \eps', \delta', \beta)$.

    This is, however, not directly private with respect to the original input, as two neighboring inputs for the original problem (which, by definition differ by at most one item) might differ in up to  $b$ input items in the new problem. By group privacy, the algorithm is, thus, $(b\eps', b \exp((b-1)\eps') \delta')$-DP. Choosing $\eps' = \eps/b$ and $\delta' = \delta/(3b)$ yields an algorithm that is $(\eps, \delta)$-DP, as desired (where we used $\eps \leq 1$ to bound $\exp(\eps((b-1)/b)) \leq 3$).
\end{proof}

\subsubsection{Generalized  summation under continual observation}
We now show how to build a general summation algorithm that is private under continual observation. Recall that there are $T$ time steps and a stream $\sigma$ of updates is given as input one by one such that $\sigma(t)$ corresponds to either the addition, resp.,~removal of an item of $X$ to resp.~ from the dataset $P$ or to a ``no operation'', where the dataset remains unchanged. Two neighboring input streams differ in the operation in at most one time step.

\begin{restatable}{lemma}{dpCountNet}
\label{lem:histoContinual}
There exists an algorithm for generalized bucketized vector summation that  is $(\eps, \delta)$-DP under continual observation, and has  with  probability $1-\beta$ additive error (simultaneously over all time steps) of $O\lpar b\sqrt{D} \cdot \eps^{-1} \log(DT)  \sqrt{\log (mDT/\beta) \log(b\log(T)/\delta)}\rpar $. 

When $\delta = 0$, the additive error (simultaneously over all time steps) is, with probability $1-\beta$,
$O(bD \eps^{-1} \log^2 T \log (m/\beta))$.
\end{restatable}

\begin{proof}
In our setting, at each time step $t$ the algorithm needs to output an $m$-dimensional vector $v$ whose $i$-th entry 
$v_i$ is a $D$-dimensional vector that approximates $\sum_{p \in P \cap S_i} f_i(p)$,
where $P$ is the current dataset.
We first describe and analyze the algorithm for one set $S_i$, i.e., the case $m=1$, and then extend it to $m>1$.

{\bf Case 1: $m=1$.}
This is the case where only one set with a corresponding function exists, which we call $S_1$, resp.~$f_1$ for simplicity of notation.
For $\delta > 0$, we will use the Gaussian mechanism, and, therefore, we need to bound the $\ell_2$-sensitivity of $\sum_{p \in P \cap S_1} f_1(p)$.
As two neighboring input streams differ in the operation in at most one time step,  the datasets of two neighboring input streams differ in at most two points.
As the function $f_1$ returns a value from the unit $\ell_2$-ball $B_D(0,1)$, the $\ell_2$-sensitivity of the function $\sum_{p \in P \cap S_1} f_1(p)$ is $L_2 = 2$.

For $\delta = 0$, we will use the Laplace mechanism and need to bound the $\ell_1$-sensitivity  of the function $\sum_{p \in P \cap S_1} f_1(p)$: The Cauchy-Schwarz inequality shows that it is $2\sqrt{D}$.

{\bf Case 1.1: $D=1$.}
In the special case $D=1$ (i.e., $f_1$ returnsa scalar) and $m=1$ (i.e., there is a single set), the problem is called \textit{continual counting}, which can be solved by the so-called binary-tree mechanism, studied by \cite{dwork2010differential,ChanSS11,FichtenbergerHO21}.\footnote{Note that the proof given in these references only shows the
claim for numbers from $\{-1, +1\}$, but the same proof goes through verbatim for real numbers from $[-1,  1]$.}  It works as follows: there is a data structure storing, for each multiple of a power of $2$ (namely $j \cdot 2^\ell$), the value of the sum between time steps $j \cdot 2^\ell$ and $(j+1) 2^\ell - 1$, plus an additive noise drawn from the Laplace distribution with parameter $2 \log(T)/\eps$. 
Let $U_\ell(j)$  be this noisy sum.
Then, to compute the sum of all values before time $t$, the algorithm computes the binary decomposition of $t$, say $t = \sum_j 2^{\ell_j}$ with $\ell_j$ increasing, and outputs $\sum_j U_{\ell_j}\lpar \sum_{j' > j} 2^{\ell_{j'} - \ell_j}\rpar$. One can verify that each time $t' \in [1, t]$ appears exactly once in this sum: more visually, the algorithm maintains a binary tree, where the leafs are the input sequence and each node contains a noisy sum of all leaves below it, and decomposes the segment $(0, t]$ as a disjoint union of nodes.

The privacy of this mechanism stems from the fact that each input value appears in $\log T$ many $U_\ell(j)$ -- one for each $\ell$ with $2^\ell \leq T$.  The $\ell_1$-sensitivity of continual counting is at most $2$, as each input value is in $[-1, 1]$: therefore, standard properties of the Laplace mechanism show that the whole process is $\eps$-DP. 
More generally, 
if $f_1$ has  $\ell _1$-sensitivity $L_1$, the mechanism adds to this sum  Laplace noise with parameter $\Theta(L_1 \log(T)/\varepsilon)$ and remains
$\eps$-DP.

For $(\eps, \delta)$-DP we will replace the Laplace mechanism by the Gaussian. As a consequence we need to show privacy and accuracy for the sum of up to $\log T$ Gaussian variables.
Theorem A.1 in~\cite{dwork2014algorithmic} shows that adding Gaussian noise ${\cal N}(0,\sigma^2)$ to the value of a function $f$  gives an $(\eps,\delta)$-DP mechanism if $\sigma \ge \sqrt{2 \ln(1.25/\delta)}  L_2/\eps$,  where $L_2$ is the $\ell_2$-sensitivity of the function computed. Theorem B.1 in~\cite{dwork2014algorithmic} shows that the composition of $k \ge 2$
$(\eps,\delta)$-DP mechanisms is $(k\eps,k\delta)$-DP.
Note that the privacy loss of outputting  $\sum_j U_{\ell_j}\lpar \sum_{j' > j} 2^{\ell_{j'} - \ell_j}\rpar$ is no larger than outputting each $U_{\ell_j}\lpar \sum_{j' > j} 2^{\ell_{j'} - \ell_j} \rpar$ value individually, which corresponds to the composition of up to $\log T$ many Gaussian mechanisms.
If each Gaussian mechanism is $((\eps/\log T), (\delta/\log T))$-DP, then the resulting mechanism is $(\eps,\delta)$-DP.
Thus to guarantee $(\eps,\delta)$-DP, we choose the noise for each Gaussian variable to be
${\cal N}(0,\sigma'^2)$ with $\sigma' = \sqrt{2 \ln ((1.25 \log T) /\delta)} L_2 (\log T)/\eps = \Theta(L_2 \sqrt{\log(\log (T)/\delta)} \cdot \log(T) /\eps)$.

For the utility, \cite{jain2021price} observed that for $\eps$-DP, with probability $1-\beta'$, the bound on the additive error for any {\em individual} time step is 
$O(L_1 \eps^{-1} \log(T) \sqrt{\log(1/\beta')} ( \sqrt{\log (T)} + \sqrt{\log(1/\beta')}))$.
This is because to output the sum over time steps between $1$ and $t$, the noise computed is $O(\log T)$ Laplace noise with parameter $\Theta(L_1 \log(T) \eps^{-1})$. 
The concentration bound of random Laplacian variable from~\cite{ChanSS11} bounds this sum of noise as desired.
Rescaling $\beta'$ by $T$, a union-bound over all time steps ensures that with probability $1-\beta'$, over all time steps simultaneously, the maximum error is 
$O\lpar L_1 \eps^{-1} \log(T)^2 \log(1/\beta')\rpar$.
For $(\eps, \delta)$-DP, we can get an improved error bound (see \cite{jain2021price}): with probability $1-\beta'$, simultaneously over all time steps the error is at most 
$O\lpar L_2  \eps^{-1} \log(T)  \sqrt{\log (T/\beta') \log(\log(T)/\delta)}\rpar$.

{\bf Case 1.2: $D>1$.}
We extend this result to the case $D > 1$ as follows: instead of scalar values, the inputs and outputs are ($D$-dimensional) vectors in $B_D(0,1)$ and, in the same way as for $D=1$, the noisy sum of a suitable subset of these vectors is stored at the nodes of the binary tree mechanism. We call this a {\em $D$-dimensional continual counting algorithm}.

\textbf{Case $\delta = 0$:}
Recall that the $\ell_1$-sensitivity $L_1$ of $\sum_{p \in P \cap S_i} f_i(p)$ is $O(\sqrt{D})$ in this case.
Thus, to ensure $\eps$-DP, we use at each node of the binary tree $D$-dimensional Laplace noise with parameter $\Theta(\sqrt{D} \log(T) \eps^{-1})$ for each coordinate and the standard privacy proof applies. 
This gives, with probability at least $1-\beta'$ simultaneously for all time steps, an additive $\ell_{\infty}$-error over all $D$ dimensions of 
$O\lpar \sqrt{D} \cdot \eps^{-1} \log(T)^2 \log(1/\beta')\rpar$. This then implies, with probability at least $1-\beta'$ simultaneously for all time steps,  an additive $\ell_2$-error over the $D$ dimensions  of
$O\lpar D \cdot \eps^{-1} \log(T)^2 \log(1/\beta')\rpar$. 

\textbf{Case $\delta > 0$}: To ensure $(\eps, \delta)$-DP, we use  a $D$-dimensional Gaussian noise vector with standard deviation $O\lpar \sqrt{\log(\log(T)/\delta)} \log(T)/\eps\rpar$: since the $\ell_2$-sensitivity of $\sum_{p \in P \cap S_i} f_i(p)$  is $O(1)$, this noise ensures privacy. Furthermore,  with probability at least $1-\beta'$ simultaneously for all time steps, it
results in an additive $\ell_{\infty}$-error over all $D$ dimensions for {\em all} time steps simultaneously of $O\lpar \eps^{-1} \log(DT)  \sqrt{\log (DT/\beta') \log(\log(T)/\delta)}\rpar$.
This in turn implies, with probability at least $1-\beta'$ simultaneously for all time steps,  an additive $\ell_{2}$-error over all $D$ dimensions  of
$O\lpar \sqrt{D} \cdot \eps^{-1} \log(DT)  \sqrt{\log (DT/\beta') \log(\log(T)/\delta)}\rpar$.

{\bf Case 2: $m >1$.}
Now we extend this result to the histogram case: for this, it is enough to run $m$ $D$-dimensional continual counting algorithm in parallel, one for each set. Since each element is part of at most $b$ sets, the difference between the input of two neighboring streams is only on $b$ different executions of continual counting mechanism: therefore, this composition of algorithms is $(b \eps, b\delta)$-DP. 
For completeness, we provide a detailed proof in \Cref{lem:parallelCounting}.
 We therefore rescale $\eps$ and $\delta$ by $b$, to get an $\eps$-DP algorithm, with error at any individual time step and for any fixed set  of
$O\lpar b D \cdot \eps^{-1} \log(T)^2 \log(1/\beta')\rpar$ with probability $1-\beta'$. 
We set $\beta' = \beta/m$ to do a union-bound for all sets, and conclude the bound of the lemma. 
The same settings of parameters conclude the bound for $(\eps, \delta)$-DP.

\end{proof}

\begin{claim}\label{lem:parallelCounting}
    Consider the following procedure for generalized summation under continual observation. Let $\calA_1, ..., \calA_{m}$ be $D$-dimensional continual counting algorithms  that use independent randomness. The procedure creates $m$ streams of input, one for each set $S_i$, and feeds it to $\calA_i$. If each $\calA_i$ is $\eps$-DP and each input element is part of $b$ sets, then the procedure is $b \eps$-DP. When $\delta > 0$, the procedure is $(b\eps, b\delta)$-DP.
\end{claim}
\begin{proof}
    On an input stream $\calD$, the algorithm creates one stream per set, say $\calD_{i}$ for set $S_i$, and feeds it to a continual counting algorithm $\calA_{i}$. Let $\calA$ be the whole mechanism that outputs $(\calA_{1}(\calD_{1}), ..., \calA_{m}(\calD_{m}))$.
For a dataset $\calD$ and any potential solution $s = (s_{1}, ..., s_{m})$ we have
$\Pr[\calA(\calD) = s] = \prod_{i=1}^m \prod_{j=1}^D\Pr[\calA_{i}(\calD_{i}) = s_{i}]$,
because the random choices of the $\calA_{i}$ are independent.

Consider two neighboring streams $\calD, \calD'$ : by assumption, only $b$ of the $\calD_i$ differ from $\calD'_i$, and those at a single time step. 

\textbf{Case $\delta = 0$.}
We analyze the ratio of probability 
$$\frac{\Pr[\calA(\calD) = s]}{\Pr[\calA(\calD') = s]}= \prod_{i=1}^m  \frac{\Pr[\calA_{i}(\calD_{i}) = s_{i}]}{\Pr[\calA_{i}(\calD'_{i}) = s_{i} ]}.$$

The ratio inside the product is equal to 1 for at least $(m-b)$ many $\calD_{i}$ (the ones where $\calD_{i} = \calD'_{i}$), and since the $\calA_{i}$ are $\eps$-DP, the ratio is at most $\exp(\eps)$ for the remaining $b$ of them.
Therefore,
$$\frac{\Pr[\calA(\calD) = s]}{\Pr[\calA(\calD') = s]}\leq \exp(b \eps),$$
which concludes the proof when $\delta = 0$.

\textbf{Case $\delta > 0$.} 
For simplicity, we reorder the sets such that only the streams $\calD_1, ..., \calD_b$ differ. 
The standard proof of composition (see e.g. Corollary B.2 in \cite{dwork2014algorithmic}) gives that 
$$\Pr[(\calA_1(\calD_1), ..., \calA_b(\calD_b)) = (s_1, ..., s_b)] \leq \exp(b\eps) \cdot \Pr[(\calA_1(\calD'_1), ..., \calA_b(\calD'_b)) = (s_1, ..., s_b)] + b\delta.$$
Indeed, for $b=2$ we have (using that the $\calA_i$ are independent, which simplifies compared to \cite{dwork2014algorithmic}):
\begin{align*}
    \Pr[(\calA_1(\calD_1), \calA_2(\calD_2)) = (s_1, s_2)]  &= \Pr[\calA_1(\calD_1) = s_1] \cdot \Pr[\calA_2(\calD_2) = s_2]\\
    &\leq \Pr[\calA_1(\calD_1) = s_1] \cdot  \min\big( 1, \exp(\eps)\Pr[\calA_2(\calD'_2) = s_2] + \delta \big)\\
    &\leq \delta + \Pr[\calA_1(\calD_1) = s_1] \cdot  \min\big( 1, \exp(\eps)\Pr[\calA_2(\calD'_2) = s_2]\big)\\
    &\leq \delta + \big( \exp(\eps) \Pr[\calA_1(\calD'_1) = s_1]  + \delta \big) \cdot  \min \big( 1, \exp(\eps)\Pr[\calA_2(\calD'_2) = s_2]\big)\\
    &\leq 2\delta + \exp(2\eps)\Pr[\calA_1(\calD'_1) = s_1]\Pr[\calA_2(\calD'_2) = s_2]\\
    &\leq 2\delta + \exp(2\eps)\Pr[(\calA_1(\calD'_1), \calA_2(\calD'_2)) = (s_1, s_2)].
\end{align*}
Generalizing to larger $b$ is a simple induction.

For any  $i \geq b+1$, we have $\Pr[\calA_i(\calD_i) = s_i] = \Pr[\calA_i(\calD'_i) = s_i]$, and furthermore all $\calA_i$ use independent randomness. Therefore,
$\Pr[\calA(\calD) = s] \leq \exp(b \eps) \Pr[\calA(\calD') = s] + b\delta$. 

Using the same argument as for $\delta=0$ to consider streams $\calD_i$ that do not differ concludes the proof.
\end{proof}

\section[Missing Proof and Statement for Section 2]{Missing Proof and Statement for \Cref{sec:coating}}
\subsection{From Bi-criteria to Proper Solutions}
\label{sec:bicriteria}
One of the idea used repeatedly in literature is to relax the constraint on $k$: instead of looking for a set of exactly $k$ centers with good cost, one can first look for $k'$ centers (with $k' \geq k$). Given $k'$ centers, one can compute a private version of the dataset, where each center is weighted by the number of points in its cluster (with appropriate noise). This is private, and therefore one can use any clustering algorithm to find $k$ centers. Furthermore, triangle inequality ensures that the quality of the solution computed will be good, provided that the $k'$ centers are good as well.

Formally, we have: 
\begin{lemma}\label{lem:bicriteria}
    Let $S$ be a set of $k'$ centers, with $\cost(P, S) \leq M \opt_{k,z} + A$. 
    For each center $s_i \in S$, let $S_i$ be its cluster, and $n_i$ be such that $|n_i - |S_i|| \leq e$.
    Let $C$ be a set of $k$ centers that is a $M'$ approximation on the dataset containing $n_i$ copies of each $s_i$. 
    Then, for any $ \alpha \geq 0$ $C$ is a solution with multiplicative approximation $\lpar 2(1+1/\alpha)^{z-1} M + (1+\alpha)^{z-1}\rpar M'$ and additive error $(1+1/\alpha)^{z-1} A + O(k' e)$.
\end{lemma}
To help parse the bounds of this lemma, we note here the two applications we will make: when $M = O(1)$, we will use $\alpha = 1$ to get a multiplicative approximation $O(1)$ and additive error $O(A + k' e)$; and when is much more tiny (e.g. $M = {\alpha'}^z$ for some $\alpha' < 1$) we will use $\alpha = \alpha' / z$ to get multiplicative approximation $1 + O(\alpha')$ and additive error $O(A/{\alpha'}^{z-1} + k'e)$.
\begin{proof}
    Let $\opt$ be the optimal solution for $(k,z)$-clustering on $P$.
    Fix a point $p \in P$, in cluster $S_i$. We have, by the generalized triangle inequality (\Cref{lem:weaktri}):
    \begin{align*}
        \cost(s_i, \opt) \leq (1+1/\alpha)^{z-1} \cost(p, \opt) + (1+\alpha)^{z-1} \cost(p, s_i).
    \end{align*}
    Since $C$ is a $M'$ approximation, we have
        \begin{align}
        \notag
        \sum_{i=1}^{k'} n_i \cost(s_i, \opt) &\leq \sum_{i=1}^{k'} |S_i| \cost(s_i, \opt) + k' e \\
        \notag
        &= \sum_{i=1}^{k'} \sum_{p\in S_i} \cost(s_i, \opt) + k' e \\
        \notag
        &\leq \sum_{i=1}^{k'} \sum_{p\in S_i} (1+1/\alpha)^{z-1}\cost(p, s_i) + (1+\alpha)^{z-1}\cost(p, \opt) + k' e \\
        \notag
        &= (1+1/\alpha)^{z-1}\cost(P, S) + (1+\alpha)^{z-1}\opt_{k,z} + k' e \\
        \label{eq:bicritToOpt}
        &\leq (1+1/\alpha)^{z-1}\cost(P, S) + (1+\alpha)^{z-1}\opt_{k,z} + k' e 
    \end{align}
Then, we can bound the cost of the full dataset $P$ in solution $C$:
\begin{align*}
    \cost(P, C) &= \sum_{p \in P} \cost(p, C)\\
    &\leq \sum_{i=1}^{k'} \sum_{p\in S_i} (1+1/\alpha)^{z-1} \cost(p, s_i) + (1+\alpha)^{z-1} \cost(s_i, C)\\
    &= (1+1/\alpha)^{z-1}\cost(P, S) + (1+\alpha)^{z-1} \sum_{i=1}^{k'} |S_i| \cost(s_i, C)\\
    &\leq (1+1/\alpha)^{z-1}\cost(P, S) + (1+\alpha)^{z-1} \sum_{i=1}^{k'} n_i \cost(s_i, C) + (1+\alpha)^{z-1} k' e \\
    &\leq (1+1/\alpha)^{z-1}\cost(P, S) + (1+\alpha)^{z-1} M' \sum_{i=1}^{k'} n_i \cost(s_i, \opt) + (1+\alpha)^{z-1} k' e 
\end{align*}
Using $\cost(P, S) \leq M\opt_{k,z}+A$ and \Cref{eq:bicritToOpt}, we conclude
\begin{equation*}
    \cost(P,C) \leq 
    \lpar 2(1+1/\alpha)^{z-1} M + (1+\alpha)^{z-1}\rpar M' \cdot \opt_{k,z} + (1+1/\alpha)^{z-1} A + O(k' e) .
\end{equation*}
\end{proof}

This implies in particular \Cref{lem:boostApprox} since \cite{dpHD} showed that, for $k' = k \alpha^{-O(d)} \log(n/\alpha)$, the optimal cost for $(k',z)$-clustering is at most $\alpha \opt_k$, namely an $\alpha$-fraction of the optimal cost with $k$ centers. Thus, we can chose $S$ in \Cref{lem:bicriteria} to be an $M$-approximation to $(k', z)$-clustering, which yields \Cref{lem:boostApprox}.

\subsection{Reducing the dimension}\label{ap:dimred}

In the next theorem, a clustering  is viewed as a partition $P_1, ..., P_k$, and the center serving part $P_i$ is the optimal center $\mu_z(P_i)$ for $(1,z)$-clustering on $P_i$.  Therefore, the clustering cost of a partition is $\sum_{i} \sum_{p \in P_i} \dist(p, \mu_z(P_i))^z$. The next theorem states that one can randomly project into $O(\log k)$ dimension while preserving the clustering cost of any partition:

\begin{lemma}[see Theorem 1.3 in~\cite{MakarychevMR19}, also Becchetti et al.~\cite{BecchettiBC0S19}]\label{lem:dimremFull}
    Fix some $1/4 \geq \alpha > 0$ and $1 > \beta > 0$. There exists a family of random projection $\pi : \Rd \rightarrow \hRd$ for some  $\hat d = O(z^4 \cdot \log (k/\beta) \alpha^{-2})$ such that, for any multiset $P \in \Rd$, it holds with probability $1-\beta$ that for any partition of $P$ into $k$ parts $P_1, ..., P_k$,
    \[\sum_{j=1}^k \cost\big(\pi(P_j), \mu_z(\pi(P_j))\big) \in \lpar 1 \pm \alpha\rpar \cdot (\frac{\hat d}{d})^{z/2} \cdot \sum_{j=1}^k \cost(P_j, \mu_z(P_j)).\]
\end{lemma}

\liftingViaHist*
\begin{proof}

We start with the $k$-means case. With probability $1-\beta/2$, \Cref{lem:dimremFull} ensures that the cost of any partition is preserved up to a factor $(1\pm \alpha)$, then any $(a,b)$-approximation for $\pi(P)$ is an $((1+ \alpha)a, b)$-approximation for $P$. Moreover with probability at least $1-\beta/2$,  all the points of projected and rescaled dataset $\pi(P)$ lie within the ball $B_{\hd}(0, \sqrt{2\log(n/\beta)})$.
We condition on both of those events: by our assumptions, clustering  each $P_i$ to its optimal center is then a $(M, \sqrt{2\log(n/\beta)}^z \cdot A(k,\hd))$-approximation.

Note that the optimal $1$-mean center for a cluster $P_i$ is its average $\mu(P_i) := \frac{\sum_{p\in P_i} p}{|P_i|}$. We now bound the additive error incurred by using $n_i$ and $\sumEst_i$ instead $|P_i|$ and $\sum_{p\in P_i} p$. By assumption, we have $n_i = |P_i| \pm e$ and $\lnor \sumEst_i - \sum_{p\in P_i} p\rnor \leq e$.

We let $\bc_i = \frac{\sumEst_i}{n_i}$ be our estimate on $\mu(P_i)$. We first prove that $\bc_i$  is close to $\mu(P_i)$ and then bound the desired cost difference: For this, we use that both the numerator and denominator are well approximated. 

First, if a cluster contains fewer than $e$ points, then regardless of the position of $\bc_i$ its $1$-means cost is upper-bounded by $e$, and this is accounted for by the additive error.
    Otherwise, we can show that $\bc_i$ is very close to $\mu(P_i)$ as follows:

    \begin{eqnarray*}
        \lnor \frac{\sumEst_i}{n_i} - \mu(P_i)\rnor_2 
        &=& \lnor  \lpar \frac{\sumEst_i}{n_i} - \frac{\sumEst_i}{|P_i|}\rpar + \lpar \frac{\sumEst_i}{|P_i|} - \frac{|P_i|\mu(P_i)}{|P_i|}\rpar \rnor_2\\
        &\leq& \lnor \sumEst_i\rnor_2 \lpar \frac{1}{|P_i| - e} - \frac{1}{|P_i|}\rpar + \frac{\lnor \sumEst_i - |P_i|\mu(P_i)\rnor_2}{|P_i|}\\
    \end{eqnarray*}
    Using the bound on $\sumEst_i$, and $(1-x)^{-1} \leq 1 + 2x$ for $x \leq 1/2$, this is upper-bounded by:
    \begin{eqnarray*}
         \lnor \frac{\sumEst_i}{n_i} - \mu(P_i)\rnor_2 &\leq& \lnor \sumEst_i\rnor_2 \lpar \frac {1} {|P_i|}   + \frac{2e}{|P_i|^2} - \frac{1}{|P_i|}\rpar + \frac{e }{|P_i|}
    \end{eqnarray*}

Now, note that $\lnor \sumEst_i\rnor_2 \leq \lnor \sum_{p \in P_i} p \rnor_2 + e \leq |P_i| + e$. Therefore, we get the upper bound  
    \begin{align}
        \nonumber
         \lnor \frac{\sumEst_i}{n_i} - \mu(P_i)\rnor_2 &\leq  \frac{2e (|P_i| + e)}{|P_i|^2} + \frac{e}{|P_i|}\\
        \label{eq:errMu}
         &\leq  \frac{5e}{|P_i|},
    \end{align}
    where the last line follows from $|P_i| \geq e$.

    Finally, we get, using  the folklore property that for any set $P$ and point $x$, $\cost(P, x) = \cost(P, \mu(P)) + |P| \lnor \mu(P)-x\rnor_2 ^2$ (see e.g. \cite{InabaKI94} proof of Theorem 2):
    \begin{align*}
        \cost(P_i, \bc_i) &= \cost(P_i, \mu(P_i)) + |P_i| \cdot \lnor \bc_i - \mu(P_i)\rnor_2^2\\
        &\leq \cost(P_i, \mu(P_i)) +\frac{O(e^2) }{|P_i|} \\
        &\leq \cost(P_i, \mu(P_i)) + O(e),
    \end{align*}
    where the last line uses again $|P_i| \geq e$.
    Summing over all clusters concludes the proof: the additive error in addition to $\polylog(n) \cdot A(\hd)$ is $k\cdot e$, and the multiplicative approximation is $(1+\alpha)M$ (from the dimension reduction guarantee).

    For $(k,z)$-clustering, we use the fact that $\mu(P_i)$ is a $2^z$-approximation for $(1, z)$-clustering on $P_i$, as shown in \Cref{lem:meanMed} below. Therefore, $\cost(P_i, \mu(P_i)) \leq 2^z\cost(P_i, m_i)$, where $m_i$ is the optimal $(1,z)$-center for $P_i$: using \Cref{eq:errMu} and Generalized triangle inequality \Cref{lem:weaktri} therefore yield
        \begin{align*}
        \cost(P_i, \bc_i) &\leq (1+\alpha)\cost(P_i, \mu(P_i)) + (4z/\alpha)^{z-1}|P_i| \cdot \lnor \bc_i - \mu(P_i)\rnor_2^z\\
        &\leq 2^z (1+\alpha)\cost(P_i, m_i) + O(e) .
    \end{align*}
    Summing over all clusters, this shows that the $\bc_i$ yields an multiplicative approximation $2^z(1+\alpha)M$ and additive error $\polylog(n) \cdot A(\hd) +k\cdot e$.
\end{proof}

\meanMed*
\begin{proof}
Our goal is to bound the distance between $\mu_z$ and $\mu$. 
    Consider the line $\ell$ going through $\mu_z$ and $\mu$. 
    Let $P_\ell$ be the multiset $P$ projected onto $\ell$, and $p_\ell$ the projection of any point $p$.
    Since the mean is linear, $\mu$ is also the mean of $P_\ell$, i.e., $\mu = \frac{\sum_{p_\ell \in P_\ell} p_\ell}{|P|}$.
    Therefore, we have $\mu - \mu_z = \frac{\sum_{p_\ell \in P_\ell} p_\ell-\mu_z}{|P|}$, and $|P| \|\mu - \mu_z\| \leq \sum_{p_\ell \in P_\ell} \|p_\ell-\mu_z\|$: since projection only decrease the norm, $\|p_\ell - \mu_z\| \leq \|p-\mu_z\|$, and the sum is at most $ \sum_{p \in P} \|p-\mu_z\|$. 

   This means that $\|\mu - \mu_z\| \leq \frac{\sum_{p \in P} \|p-\mu_z\|}{|P|}$
   and Jensen's inequality yields $\|\mu - \mu_z\|^z \leq \frac{\sum_{p \in P} \|p-\mu_z\|^z}{|P|}$.
   Therefore,
    $\sum_{p\in P} \|p-\mu\|^z \leq 2^{z-1}\sum_{p\in P} \lpar \|p-\mu_z\|^z +  \|\mu - \mu_z\|^z\rpar \leq 2^z\sum_{p\in P} \|p-\mu_z\|^z$. This concludes the lemma.
\end{proof}

\subsection{Boosting the probabilities}
\costVariance*
\begin{proof}
    Using properties of squared distances, we have the following:
    \begin{align*}
        \cost(E, \mu(E)) &= \sum_{p\in E} \|p - \mu(E)\|_2^2 
        = \sum_{p\in E} \sum_{i = 1}^d |p_i - \mu(E)_i|^2\\
        &= \sum_{p\in E} \sum_{i = 1}^d p_i^2 + \mu(E)_i^2 - 2 p_i \mu(E)_i\\
        &= \sum_{p\in E} \|p\|_2^2 + \|\mu(E)\|_2^2  - 2 \sum_{i=1}^d p_i \mu(E)_i.
    \end{align*}
    Here, we note that $\mu(E) = \frac{1}{|E|}\sum_{y \in E} y$: therefore, $ \sum_{p\in E} \|\mu(E)\|_2^2 = |E| \cdot \|\mu(E)\|_2^2 
 =  \frac{\lnor \sum_{y \in E} y \rnor_2 ^2}{|E|}$. Furthermore, $\mu(E)_i = \frac{1}{|E|}\sum_{y \in E} y_i$, and so:
    \begin{align*}
        \sum_{p\in E} \sum_{i=1}^d p_i \mu(E)_i &= \sum_{i=1}^d \lpar \frac{1}{|E|}\sum_{y \in E} y_i\rpar \lpar \sum_{p\in E} p_i \rpar \\
        &= \frac{1}{|E|}\sum_{i=1}^d \lpar \sum_{p \in E} p_i\rpar^2\\
        &=  \frac{\lnor \sum_{p \in E} p\rnor_2^2}{|E|}.
    \end{align*}

    Combining those equations concludes the lemma:
        \begin{align*}
        \cost(E, \mu(E)) 
        &= \sum_{p\in E} \|p\|_2^2 + \|\mu(E)\|_2^2  - 2 \sum_{i=1}^d p_i \mu(E)_i\\
        &= \sum_{p\in E} \|p\|_2^2 - \frac{\lnor \sum_{p \in E} p\rnor_2^2}{|E|}.\qedhere
    \end{align*}
\end{proof}

We now turn to the proof of \Cref{cor:boostProba}.
\boostProba*
\begin{proof}
    To boost the probability from $2/3$ to $1-\beta$, one merely needs to run $\log(1/\beta)$ independent copies of the algorithm, estimate the cost and output the solution with cheapest cost. If the estimate of the cost for each cluster is within an additive $O(e )$ of the true cost, Chernoff bounds ensure the desired guarantees. Therefore, we only need to show that one can estimate the cost of a clustering, provided the estimate values of the lemma.

By assumption, we have $n_i = |P_i| \pm e, \|\sumEst_i\|_2 = \lnor \sum_{p \in P_i} p\rnor_2 \pm e $ and $\sumNorm_i = \sum_{p \in P_i} \lnor p\rnor_2 \pm e $.

    First, in the case where $n_i \leq 2e$, then we know $|P_i|\leq e$ and the cost of the cluster is at most $e $: therefore, the estimation $e $ is fine enough.

    Otherwise, we have from \Cref{lem:costVariance} that the cost of the cluster is 
    $\sum_{p\in P_i} \|p\|_2^2 - \frac{\lnor \sum_{p \in P_i} p\rnor_2^2}{|P_i|}$. The first term is estimated by $\sumNorm_i$, up to an additive $e $. For the second one, we have (dropping for simplicity the subscript $p\in P_i$):
    \begin{align*}
        \left| \frac{\lnor \sum p\rnor_2^2}{|P_i|} - \frac{\|\sumEst_i\|_2^2}{n_i}\right| &\leq \frac{\left|\lnor \sum p\rnor_2^2 - \|\sumEst_i\|_2^2\right|}{|P_i|} +  \|\sumEst_i\|_2^2 \left| \frac{1} {|P_i|} - \frac{1}{n_i}\right|.
    \end{align*}
    We bound the first term using the guarantees on $\|\sumEst_i\|_2$: $\|\sumEst_i\|_2^2 = \lpar \lnor \sum p\rnor_2 \pm e \rpar^2 = \lnor \sum p\rnor_2^2 \pm (2|P_i| e  + e^2)$, where we used in the last inequality that $\lnor \sum p\rnor_2 \leq |P_i|$.
    Therefore, using $|P_i|\geq e$, the first term is $\frac{\lnor \sum_{p \in P_i} p\rnor_2^2}{|P_i|} \pm O(e )$.
    
    For the second term, we first bound $\left| \frac{1} {|P_i|} - \frac{1}{n_i}\right|$: using standard approximation of $(1+x)^{-1}$, this is at most $\frac{2e}{|P_i|^2}$. Now, the term $\|\sumEst_i\|_2^2$ can be bounded as follows:
    $ \|\sumEst_i\|_2^2 \leq 2\lnor \sum p\rnor_2 + 2e^2 \leq 2(|P_i|^2+e^2) $. 
    Using again $|P_i|\geq e$, we conclude 
    \begin{align*}
         \|\sumEst_i\|_2^2 \left| \frac{1} {|P_i|} - \frac{1}{n_i}\right| &\leq 2(|P_i|^2+e^2)  \cdot \frac{2e}{|P_i|^2}\\
         &= O(e ).
    \end{align*}

    Combining all those guarantees, we get: 
    \begin{align*}
        \sumNorm_i + \frac{\|\sumEst_i\|_2}{n_i} &= \sum_{p\in P_i} \|p\|_2^2 - \frac{\lnor \sum_{p \in P_i} p\rnor_2^2}{|P_i|} \pm O(e).
    \end{align*}
    This estimation of the cost concludes the proof.
\end{proof}

\section[Missing Proof of Section 3]{Missing Proof of \Cref{sec:MPalgo}}\label{sec:appMP}
\subsection[The Original Algorithm of Mettu and Plaxton]{The Original Algorithm of \cite{MettuP00}}
To emphasize that our algorithmic modifications are light, we state the original algorithm of \cite{MettuP00} in \Cref{alg:mp1}, in order to allow comparison with our \Cref{alg:mp}.

\begin{algorithm}[H]
\caption{$\mettuP(P)$}
\label{alg:mp1}
\begin{algorithmic}[1]
\STATE{Let $C_0 = \emptyset$}
\FOR{$i$ from $0$ to $n-1$}
\STATE{Let $\sigma_i$ denote the singleton sequence $(B)$ where $B$ is a maximum value ball in $\{\isolated(x,C_i)| x\in P\setminus C_i\}$}
\WHILE{The last element of $\sigma_i$ has more than one child}
\STATE{Select a maximum value child of the last element of $\sigma_i$, and append it to $\sigma_i$.}
\ENDWHILE
\STATE{$c_{i+1}$ is the center of the last ball of $\sigma_i$}
\STATE{$C_{i+1} = C_i \cup \{c_{i+1}\}$}
\ENDFOR
\end{algorithmic}
\end{algorithm}

\subsection[The Proof of Theorem 3.4]{The Proof of \Cref{thm:mpWithError}}
This section is devoted to the proof of Theorem~\ref{thm:mpWithError}. All the proof follow the original proof of \cite{MettuP00}, adapted to our setting. We restate for convenience both the algorithm and theorem we seek to proof:

\MetPla*

\mp*

In what follows, we consider a fixed solution $\Gamma$ with $k$ centers. For the purpose of this analysis, we assume that the algorithm stops after selecting $k$ centers and outputs $C_k$. Our objective is to compare the cost of the solution $C_k$ to the cost of $\Gamma$, and show that $\cost(P, C_k) \leq O(1) \cost(P, \Gamma) + O(k\theta)$. Fixing $\Gamma$ to be the optimal $(k,z)$-clustering solution will conclude.

To reuse the first part of this section in a different context later (\Cref{app:MPC}), we introduce a new parameter $c_\calA \geq 100$. We denote $\calA(c_\calA,C_k)$ as the set of balls of the form $B(x,2^{-i}) \in \calB_i$ such that for all centers $c \in C_k$, $\dist(x, c) > c_\calA \cdot 2^{-i}$. When $c_\calA = 100$, this corresponds exactly to the set of available balls at the end of the algorithm. We will state the next few definitions and lemmas with $c_\calA$, however in order to prove \Cref{thm:mpWithError}, we will only use the case $c_\calA = 100$.

For a center $\gamma \in \Gamma$, we let $P_\gamma$ denote $\gamma$'s cluster, which consists of all points in $P$ assigned to $\gamma$ in the solution $\Gamma$. We analyze the cost of each cluster independently as follows.

We split $\Gamma$ into two parts: $\Gamma_0$ is the set of $\gamma\in \Gamma$ such that there exists a center of $C_k$ at distance less than $(c_\calA + 1) \cdot 2^{-\lceil \log n\rceil}$ from $\gamma$, and let $\Gamma_1 = \Gamma - \Gamma_0$. Centers in $\Gamma_0$ are very close to centers in $C_k$, and the cost of their cluster is therefore almost the same in $C_k$.  The bulk of the work is to show that clusters in $\Gamma_1$ are also well approximated.

To analyze the cost of those clusters, we will consider the largest ball centered at $\gamma$ that is still available at the end of the algorithm. The next lemma shows that this is well defined; we show later that points outside of this ball have roughly the same cost in $C_k$ and $\Gamma$, as they are far from every center in both case; and most of the work is spent on showing points inside of this ball (called the inner cluster) have also a cheap cost: we will relate their cost to the value of the ball. 
The key observation for this is that, since the ball is still available, the algorithm selected balls with larger value. As the value is a measure of how expensive a region is, we can show that the ball selected as larger cost in $\Gamma$ than the inner cluster in $C_k$. A careful analysis concludes from this that the cost of $C_k$ is cheap relative to the one of $\Gamma$. 

\begin{lemma}\label{lem:defzgamma}
    For any $\gamma \in \Gamma_1$, the following set is not empty:
$$ \{ B = B(x,r) \in \calB : \dist(x,\gamma) \leq r/2 \text{ and } B \in \calA(c_\calA,C_k) \} $$
\end{lemma}
\begin{proof}
    Let $x$ be the closest point to $\gamma$ in $\calN_{\lceil \log n\rceil}$. By the covering property of nets, we know that $\dist(x,\gamma) \leq 2^{-\lceil \log n\rceil}/2$. We show that the ball $B(x, 2^{-\lceil \log n\rceil})$ is in $\calA(c_\calA,C_k)$. For this, assume by contradiction that it is not the case. There exists a center $c \in C_k$ such that $\dist(x, c) \leq c_\calA  \cdot 2^{-\lceil \log n\rceil}$.

    Using the triangle inequality, we get $\dist(\gamma,c) \leq \dist(\gamma,x) + \dist(x,c) \leq 2^{-\lceil \log n\rceil}/2 + c_\calA \cdot 2^{-\lceil \log n\rceil} \leq (c_\calA + 1) \cdot 2^{-\lceil \log n\rceil}$.
    Therefore $\gamma$ is in the set $\Gamma_0$, contradicting the assumption that it is in $\Gamma_1$. This contradiction completes the proof.
\end{proof}

For any $\gamma \in \Gamma_1$, let $B_\gamma = B(x_\gamma, 2^{-l_\gamma})\in \calB$ be a ball of maximum radius in the set defined in Lemma~\ref{lem:defzgamma}. We split $P_\gamma$ into two parts:  $In(P_\gamma):= P_\gamma \cap B_\gamma$ and $Out(P_\gamma) = P_\gamma - In(P_\gamma)$. 

Intuitively, points in $Out(P_\gamma)$ are about the same distance to $\gamma$ than to a center in $C_k$, which allows to easily bound their cost. On the other hand, $In(P_\gamma)$ consists of points from the cluster $P_\gamma$ that are much closer to $\gamma$ than to any selected center. We can relate the cost of $In(P_\gamma)$ for the solution $C_k$ to the value of $B_\gamma$, as done in the following lemma.

\begin{lemma}\label{lem:inAndout}
For all $\gamma \in \Gamma_1$, we have:
    \begin{align}
        &\cost(In(P_\gamma),C_k)  \leq (16 c_\calA +24 )^z\cdot (\cost(In(P_\gamma), \Gamma) + \val(B_\gamma))\label{lem:In}\\
        &\cost(Out(P_\gamma), C_k) \leq ( 4c_\calA +3 )^z \cdot  \cost(Out(P_\gamma), \Gamma). \label{lem:Out}
    \end{align}
And for all $\gamma \in \Gamma_0$, we have:
\begin{equation}
    \cost(P_\gamma, C_k) \leq 2^z\cdot (\cost(P_\gamma, \Gamma)+ |P_\gamma| \cdot (( c_\calA +1 )/n)^z). \label{lem:notDefined} 
\end{equation}
\end{lemma}
\begin{proof}
    
We begin by examining the (more interesting) case where $\gamma \in \Gamma_1$. The first step is to establish the existence of a center in $C_k$ at a distance of $O(2^{-l_\gamma})$ from $\gamma$. 

Essentially, since the ball $B_\gamma$ has maximal radius among available balls close to $\gamma$, there is one center not too far from that ball. We formalize now this idea.
When the first center $c_1$ is selected by the algorithm, the entire universe $B(0,1)$ is included in $B(c_1, c_\calA \cdot 2^{-1})$, and therefore none of the balls of level $1$ are in $\calA(c_\calA,C_k)$. 
By definition, $B_\gamma$ is in $\calA(c_\calA,C_k)$, so $l_\gamma \geq 2$. According to the covering property of nets, there exists therefore $x \in \calN_{l_\gamma-1}$ such that $\dist(\gamma, x) \leq 2^{-(l_\gamma-1)}/2 = 2^{-l_\gamma}$. Furthermore, the maximality of the radius of $B_\gamma$ ensures that the ball $B(x, 2^{-(l_\gamma-1)}) \in \mathcal{B}$ is not in $\calA(c_\calA,C_k)$. Thus, $\dist(x, C_k) \leq c_\calA \cdot 2^{-(l_\gamma-1)} = 2 c_\calA \cdot 2^{-l_\gamma}$. Combining these two inequalities, we get:
\begin{equation}
    \label{eq:gammaCk}
     \dist( \gamma, C_k) \leq \dist(\gamma, x) + \dist(x, C_k) \leq 2^{-l_\gamma} + 2 c_\calA \cdot 2^{-l_\gamma} = (2 c_\calA +1 ) \cdot 2^{-l_\gamma}.
\end{equation}

\paragraph{Proof of \Cref{lem:In}:} 

Let $p \in In(P_\gamma)$, we want to bound the cost of $p$ for the solution $C_k$. We have, with triangle inequality:
    $\cost(p,C_k) = \dist(p,C_k)^z 
    \leq (\dist(p, \gamma) + \dist(\gamma,C_k))^z.$

We already have a bound on $ \dist( \gamma, C_k)$, so we turn to $\dist(p, \gamma)$. Given $p \in B_\gamma = B(x_\gamma, 2^{-l_\gamma})$, we have $\dist(p,x_\gamma) \leq 2^{-l_\gamma}$ and by the definition of $B_\gamma$, $\dist(x_\gamma, \gamma) \leq 2^{-l_\gamma}/2$. Combining these with the triangle inequality, we obtain:
$$
\dist(p,\gamma) \leq \dist(p, x_\gamma) + \dist(x_\gamma, \gamma) \leq 2^{-l_\gamma} + 2^{-l_\gamma}/2 \leq 2 \cdot 2^{-l_\gamma}.
$$

This yields 
\begin{equation}
    \label{eq:pCk}
    \cost(p,C_k) \le (2 \cdot 2^{-l_\gamma} + (2 c_\calA +1 ) \cdot 2^{-l_\gamma})^z = (2 c_\calA +3 )^z \cdot 2^{-z \cdot l_\gamma}.
\end{equation}
We now bound $2^{-z \cdot l_\gamma}$ as follows.
Either $\dist(p,\gamma) \leq 2^{-l_\gamma}/4$: then, we get by the triangle inequality 
\begin{alignat*}{2}
    & \quad &\dist(p,x_\gamma) &\leq \dist(p,\gamma) + \dist(\gamma, x_\gamma)\\
   & \quad & &\leq 2^{-l_\gamma}/4 + 2^{-l_\gamma}/2 = 3\cdot 2^{-l_\gamma}/4 \\
   &\Rightarrow \quad & 2^{-l_\gamma} &\leq 4\cdot(2^{-l_\gamma} - \dist(p,x_\gamma)).
\end{alignat*}

Or, $2^{-l_\gamma}/4 \leq \dist(p,\gamma)$. Then, we have $2^{-l_\gamma} \leq 4 \cdot \dist(p,\gamma)$. 
Therefore, in both cases it holds that 
$2^{-l_\gamma} \leq \max(4 \cdot \dist(p,\gamma), 4\cdot(2^{-l_\gamma} - \dist(p,x_\gamma))) \leq 4 \cdot (\dist(p,\gamma) + 2^{-l_\gamma} - \dist(p,x_\gamma))$.
Raising both sides to the power of $z$ yields:
\begin{align*}
    2^{-z \cdot l_\gamma} &\leq 4^z \cdot (\dist(p,\gamma) + 2^{-l_\gamma} - \dist(p,x_\gamma))^z\\ 
    &\leq 8^z \cdot (\dist(p,\gamma)^z + (2^{-l_\gamma} - \dist(p,x_\gamma))^z).
\end{align*}

Plugging this inequality in the bound on $\cost(p, C_k)$ given by \Cref{eq:pCk}, we obtain
$$\cost(p,C_k) \leq (2 c_\calA +3 )^z \cdot 8^z \cdot (\dist(p,\gamma)^z + (2^{-l_\gamma} - \dist(p,x_\gamma))^z) = (16 c_\calA +24 )^z \cdot (\dist(p,\gamma)^z + (2^{-l_\gamma} - \dist(p,x_\gamma))^z).$$

Summing this inequality over all $p \in In(P_\gamma)$, we get 
\begin{align*}
\cost(In(P_\gamma), C_k) &\leq (16 c_\calA +24 )^z \cdot (\sum_{p \in In(P_\gamma)} \dist(p,\gamma)^z + \sum_{In(P_\gamma)}(2^{-l_\gamma} - \dist(p,x_\gamma))^z) \\
&\leq (16 c_\calA +24 )^z \cdot (\cost(In(P_\gamma), \gamma) + \sum_{p \in B_\gamma \cap P}(2^{-l_\gamma} - \dist(p,x_\gamma))^z) \\
& = (16 c_\calA +24 )^z\cdot (\cost(In(P_\gamma), \Gamma) + \val(B_\gamma)).
\end{align*}

\paragraph{Proving \Cref{lem:Out}:} We turn to $Out(P_\gamma)$, and let $p\in Out(P_\gamma)$. 
As previously, we have from \Cref{eq:gammaCk}:
$ \cost(p,C_k) \leq (\dist(p, \gamma) + (2 c_\calA +1 ) \cdot 2^{-l_\gamma})^z$. We provide in this case as well a bound on $2^{-l_\gamma}$.

The point $p$ is outside the ball $B_\gamma = B(x_\gamma, 2^{-l_\gamma})$ and therefore $\dist(p,x_\gamma) \geq 2^{-l_\gamma}$. 
Moreover, by definition of $B_\gamma$, $\dist(\gamma,x_\gamma) \leq 2^{-l_\gamma}/2$. Thus, we get:
\begin{align*}
    \dist(p,\gamma) &\geq \dist(p,x_\gamma) - \dist(\gamma,x_\gamma)\\
    &\geq 2^{-l_\gamma}- 2^{-l_\gamma}/2 = 2^{-l_\gamma}/2\\
    \Rightarrow \quad 2^{-l_\gamma} &\leq 2\cdot \dist(p,\gamma).
\end{align*}
Hence
\begin{align*}
    \cost(p,C_k) 
    &\leq (\dist(p, \gamma) + (2 c_\calA +1 ) \cdot 2^{-l_\gamma})^z\\
    &\leq (\dist(p, \gamma) +  (2 c_\calA +1 ) \cdot 2 \cdot  \dist(p, \gamma))^z\\
    &\leq ( 4c_\calA +3 )^z \cdot \dist(p, \gamma)^z.
\end{align*}
Summing this inequality over all $p \in Out(P_\gamma)$ concludes:
$$ \cost(Out(P_\gamma),C_k) \leq ( 4c_\calA +3 )^z \cdot \cost(Out(P_\gamma),\Gamma).$$
\paragraph{Proving \Cref{lem:notDefined}:} Finally, we consider the case $\gamma \in \Gamma_0$: by definition, $\dist(\gamma,C_k) \leq ( c_\calA +1 ) \cdot 2^{-\lceil \log n\rceil}$. Therefore for any $p \in P_\gamma$
\begin{align*}
    \cost(p,C_k) = \dist(p,C_k)^z &\leq  (\dist(p,\gamma) + \dist(\gamma,C_k))^z\\
    &\leq 2^z\cdot \lpar\dist(p,\gamma)^z + \dist(\gamma,C_k)^z\rpar\\
    &\leq 2^z\cdot\lpar\dist(p,\gamma)^z + ( c_\calA +1 )^z \cdot  2^{-z\lceil \log n\rceil}\rpar\\
    &\leq 2^z\cdot\lpar \dist(p,\gamma)^z + (( c_\calA +1 )/n)^z \rpar.
\end{align*}
Summing this inequality over all $p \in P_\gamma$
\begin{align*}
    \cost(P_\gamma, C_k) &\leq 2^z\cdot \lpar\cost(P_\gamma, \Gamma)+ |P_\gamma| \cdot (( c_\calA +1 )/n)^z\rpar. \qedhere
\end{align*}
\end{proof}

\Cref{lem:inAndout} allows us to derive a first bound on $\cost(P,C_k)$. Indeed, summing \Cref{lem:In} for all $\gamma \in \Gamma_1$, we get
\begin{align*}
    \cost(\bigcup_{\gamma \in \Gamma_1}In(P_\gamma), C_k) &= \sum_{\gamma \in \Gamma_1}\cost(In(P_\gamma), C_k) \leq  (16 c_\calA +24 )^z\cdot (\sum_{\gamma \in \Gamma_1} \cost(In(P_\gamma), \Gamma) + \sum_{\gamma \in \Gamma_1} \val(B_\gamma))\\
    &\leq (16 c_\calA +24 )^z\cdot \cost(\bigcup_{\gamma \in \Gamma_1}In(P_\gamma),\Gamma) + (16 c_\calA +24 )^z\cdot\sum_{\gamma \in \Gamma_1} \val(B_\gamma).
\end{align*}

Summing \Cref{lem:Out} for all $\gamma \in \Gamma_1$, we get
\begin{align*}
    \cost(\bigcup_{\gamma \in \Gamma_1}Out(P_\gamma), C_k) = \sum_{\gamma \in \Gamma_1} \cost(Out(P_\gamma), C_k) &\leq (4 c_\calA +3 )^z \cdot   \sum_{\gamma \in \Gamma_1} \cost(Out(P_\gamma), \Gamma)\\
      &\leq  (4 c_\calA +3 )^z \cdot \cost(\bigcup_{\gamma \in \Gamma_1}Out(P_\gamma), \Gamma).
\end{align*}
Summing \Cref{lem:notDefined} for all $\gamma \in \Gamma_0$, we get
\begin{align*}
\cost(\bigcup_{\gamma \in \Gamma_0}P_\gamma, C_k) &\leq 2^z\cdot (\sum_{\gamma \in \Gamma_0}\cost(P_\gamma, \Gamma)+ \sum_{\gamma \in \Gamma_0}|P_\gamma| \cdot ((1 c_\calA +1 )/n)^z)\\
    &\leq 2^z \cdot \cost(\bigcup_{\gamma \in \Gamma_0}P_\gamma,\Gamma) + (2 c_\calA +2)^z \cdot n/n^z.
\end{align*}
And finally summing the three parts, we obtain
\begin{equation}\label{lem:redvalue}
    \cost(P,C_k) \leq (16 c_\calA +24)^z\cdot \cost(P,\Gamma) + (2 c_\calA +2)^z \cdot n^{1-z} + (16 c_\calA +24)^z\cdot\sum_{\gamma \in \Gamma_1} \val(B_\gamma).
\end{equation}
From now on, we will fix the constant $c_\calA = 100$ for the rest of the proof. \Cref{lem:redvalue} allows us to prove \Cref{thm:mpWithError} in the easy case where all the available balls at the end of the algorithm have values less than $\theta$. Indeed all the balls $B_\gamma$ are available at the end of the algorithm, in that case $\sum_{\gamma \in \Gamma_1} \val(B_\gamma) \leq k\cdot \theta$ and the theorem follows immediately from \Cref{lem:redvalue} (Recall that $\theta \geq 1$). 

We note $\mval$ the maximum value of an available ball at the end of the algorithm. In the rest of the proof, we show how to bound $\sum_{\gamma \in \Gamma_1} \val(B_\gamma)$ in the remaining case, when $\mval \geq \theta$.

\subsection{Bounding the Values}\label{sec:boundValues}
To do so, we start by showing a first lemma lower bounding the cost of the balls that do not intersect $\Gamma$. We say that a ball $B\in \calB$ is \emph{covered} by $\Gamma$ if $B\cap \Gamma \neq \emptyset$ (and recall that in our analysis $\Gamma$ plays the role of the optimal solution).

\begin{lemma}\label{lem:cover}
    If a ball $B\in \calB$ is not covered by $\Gamma$, then $\cost(B \cap P,\Gamma)\geq \val(B)$.
\end{lemma}
\begin{proof}
    Consider $B = B(x, r) \in \calB$, a ball not covered by $\Gamma$. Here, $\dist(x, \Gamma) \ge r$. For any $p \in B \cap P$, the triangle inequality implies $\dist(p, \Gamma) \ge \dist(x,\Gamma) - \dist(x,p) \geq r - \dist(x, p)$. Raising both sides to the power of $z$ and summing for all $p\in B\cap P$, we get $\cost(B\cap P, \Gamma) = \sum_{p\in B\cap P} \dist(p,\Gamma)^z \geq \sum_{p\in B\cap P} (r - \dist(x, p))^z = \val(B)$.
\end{proof}

The strategy for bounding the sum of values $\sum_{\gamma \in \Gamma_1} \text{val}(B_\gamma)$ relies on the preceding lemma. Our objective is to match each $B_\gamma$ (for $\gamma \in \Gamma_1$) with an uncovered ball of at least the same value (approximately), ensuring that all the uncovered balls are disjoint. Consequently, we can then upper bound the sum of values by the cost for $\Gamma$ of those uncovered balls, as established in \Cref{lem:cover}. 

In order to define the matching, we recall the notations from \Cref{alg:mp}. During the $i$-th loop, the algorithm defines a sequence of balls $\sigma_i = (\sigma_i^1, \sigma_i^2, ...)$, that are smaller and smaller, such that the initial ball has maximum value (up to $\theta$) among the available balls, and for $j \geq 2$ the ball $\sigma_i^j$ has maximum value (up to $\theta$) among the children of $\sigma_i^{j-1}$.

In the following lemma, we show that we can \emph{prune} all the sequences $\sigma_i$ to establish some separation property. This pruning removes some balls in each sequence, ensuring that the value of the first remaining ball in each sequence is at least $\mval- \theta$, while guaranteeing that none of the remaining balls intersect. 
We denote $x_i^l$ the center of the ball $\sigma_i^l$.

\begin{lemma}\label{lem:pruning}
For all $i \in \{1,\dots,k\}$, there exists an index $l_i$ such that:
    \begin{itemize}
        \item $\val(\sigma_i^{l_i}) \geq \mval - \theta$.
        \item For all $i,i' \in  \{1,\dots,k\}$, and for all $l \geq l_i$, $l'\geq l_{i'}$, $\sigma_i^l \cap \sigma_{i'}^{l'} = \emptyset$
    \end{itemize}
\end{lemma}

\subsubsection{Proof of the Lemma \ref{lem:pruning}}

To compute each $l_i$, the pruning procedure works as follows. Start with $l_i = 1$ for all $i$. The first condition is clearly satisfied: when $\sigma_i^1$ is selected, it maximizes the value among the available balls up to $\theta$.

The procedure enforces the second condition as follows: as long as there exist sequences $\sigma_i, \sigma_{i'}$ with $i < i'$ and two levels $\ell \geq l_i$, $\ell' \geq l_{i'}$ such that $\sigma_i^{\ell} \cap \sigma_{i'}^{\ell'} \neq \emptyset$, set $l_i = \ell+1$ (i.e., \emph{prune} the sequence $\sigma_i$ to remove its entries before $\ell+1$)

We will demonstrate that this procedure is well-defined and inductively preserves the first property. Consequently, when it terminates, both conditions are satisfied.

Before proving \Cref{lem:pruning}, we need some preliminary results. Our first remark on the algorithm is as follows: at the time when a ball $\sigma_i^j$ is selected by the algorithm, it is available. This is clearly true when the first ball of the sequence is picked in line $4$. To prove that it is also true for the balls picked in line $6$, it suffices to prove the following lemma.

\begin{lemma}\label{lem:available}
    At any moment of \Cref{alg:mp1}, if a ball is available, its children are also available.
\end{lemma}
\begin{proof}
    Let $B_1 = B(x_1, r)\in \calB$ be a ball and let $B_2 = B(x_2,r/2) \in \calB$ be a child of $B_1$.  By definition, we have $\dist(x_1,x_2) \leq 10 \cdot r$. Suppose $B_2$ is forbidden by some center $c$: then  $\dist(c,x_2) \leq 100 \cdot r/2$. Using the triangle inequality yields
    $$\dist(c,x_1) \leq \dist(c,x_2) + \dist(x_2,x_1) \leq 10 \cdot r + 50 \cdot r < 100 \cdot r.$$
    Therefore $B_1$ is also forbidden by $c$.
\end{proof}

We now establish a simple lemma that sets a limit on the distance between the center of two balls appearing in a same sequence $\sigma_i$.

\begin{lemma}\label{lem:20diam}
    Let $u_0= B(x_0,r_0),\dots,u_l = B(x_l,r_l)$ be a sequence of balls of $\calB$ such that for all $i$, $u_{i+1}$ is a child of $u_i$. Then $\dist(x_1,x_l) \leq 20 \cdot r_0$.
\end{lemma}
\begin{proof}
    By definition, the distance between the center of a ball $u = B(x,r)$ and the center of any child is at most $10 \cdot r$. By induction, we get:

\begin{align*}
\dist(u_0,u_l) &\leq \sum_{j=0}^{l-1} \dist(u_j,u_{j+1})\\
&\leq \sum_{j=0}^{l-1} 10 \cdot r_j =  \sum_{j=0}^{l-1}\frac{10 \cdot r_0}{2^j} \\
&\leq 20\cdot r_0.\qedhere
\end{align*}
\end{proof}

The next Lemma is the key to prove that the procedure to compute the $l_i$'s terminates and verifies the conditions of \Cref{lem:pruning}.

\begin{lemma}\label{lem:pruningOneStep}
For every $i,i',l,l'$ such that $i<i'$ and $\sigma_i^l \cap \sigma_i^{l'} \neq \emptyset$,
\begin{itemize}
    \item $\level(x_{i'}^1) \geq \level(x_i^l)+2 $
    \item $\val(\sigma_i^{l+1}) \geq \val(\sigma_{i'}^1)-\theta$.
\end{itemize}
\end{lemma}
\begin{proof}
    Let $i,i',l,l'$ such that $i<i'$ and $\sigma_i^l \cap \sigma_{i'}^{l'} \neq \emptyset$. We start by proving the first point by contradiction: suppose that $\level(x_{i'}^1) \leq \level(x_i^l)+1$. We extend the sequence starting from $x_{i'}^{l'}$ to prove the existence of a "descendant" of $x_{i'}^1$ of level $\level(x_i^l)+1$ close to $x_i^l$. We will then prove that this descendant became unavailable when $c_i$ was selected, contradicting Lemma \ref{lem:available}.  
    
More precisely, we pick recursively a sequence $y_0,\dots, y_{j_{max}}$ such that $y_0 = x_{i'}^{l'}$ and $B(y_{j+1},2^{-\level(y_{j+1})})$ is an arbitrary child of $B(y_{j},2^{-\level(y_{j})})$ such that in $\sigma_i^l \cap B(y_{j+1},2^{-\level(y_{j+1})}) \neq \emptyset $ and we stop when $\level(y_{j_{max}}) \geq \level(x_{i}^{l})+1$. 

We prove by induction that this sequence can be defined: assume that $B(y_{j},2^{-\level(y_{j})}) \cap \sigma_i^l \neq\emptyset$, and let $x$ be a point lying in the intersection. By the covering property of nets, there exists a net point $y_{j+1}$ in $\calN_{\level(y_{j+1})}$ such that $\dist(x,y_{j+1}) \leq 2^{-\level(y_{j+1})}/2$. We have $\dist(y_j,y_{j+1}) \leq \dist(y_j,x) + \dist(x,y_{j+1}) \leq 2^{-\level(y_{j})} + 2^{-\level(y_{j+1})}/2 \leq 10 \cdot 2^{-\level(y_{j})}$. Therefore $B(y_{j+1},2^{-\level(y_{j+1})})$ is a child of $B(y_{j},2^{-\level(y_{j})})$ and by construction $x\in \sigma_i^l \cap B(y_{j+1},2^{-\level(y_{j+1})})$. 

There exists a net point $y$ in the sequence $x_{i'}^{1},\dots,x_{i'}^{l'},y_1,\dots,y_{j_{max}}$ of level $\level(x_{i}^{l})+1$. This point $y$ is either $y_{j_{max}}$ if $\level{x_{i'}^{l'}} < \level(x_{i}^{l})+1$, or some point in the sequence $x_{i'}^{1},\dots,x_{i'}^{l'} = y_0 = y_{j_{max}}$ otherwise. 

By construction, the two balls $\sigma_i^l$ and $B(y_{j_{max}},2^{-\level(y_{j_{max}})})$ are intersecting, and therefore $$\dist(x_i^l, y_{j_{max}}) \leq 2^{-\level(x_i^l)} + 2^{-\level(y_{j_{max}})} \leq 2 \cdot 2^{-\level(x_i^l)} = 4\cdot 2^{-\level(y)}.$$
Applying \Cref{lem:20diam} to the sequence $B(y,2^{-\level(y)}),\dots,B(y_{j_{max}},2^{-\level(y_{j_{max}})})$, we get $\dist(y_{j_{max}},y) \leq 20\cdot 2^{-\level(y)}$. Using the triangle inequality, we finally obtain 
$$\dist(x_i^l,y) \leq \dist(x_i^l, y_{j_{max}}) +\dist(y_{j_{max}},y) \leq 24 \cdot  2^{-\level(y)}. $$
On the other hand, applying \Cref{lem:20diam} again, we get $\dist(c_i,x_i^l) \leq 20 \cdot 2^{-\level(x_i^l)} = 40 \cdot 2^{-\level(y)}$ and therefore $\dist(c_i,y) \leq \dist(c_i,x_i^l) + \dist(x_i^l,y) \leq 64 \cdot 2^{-\level(y)} \leq 100 \cdot 2^{-\level(y)}.$ Hence, the ball $B(y,2^{-\level(y)})$ became unavailable when $c_i$ was selected and is not available when $\sigma_{i'}^1$ is picked because $i'>i$. But Lemma \ref{lem:available} guarantees the availability of all the balls within the sequence $\sigma_{i'}^{1},\dots,\sigma_{i'}^{l'},B(y_1,2^{-\level(y_{1})}),\dots,B(y_{j_{max}},2^{-\level(y_{j_{max}})})$, including $B(y,2^{-\level(y)})$, when the algorithm selects $\sigma_{i'}^1$. We have a contradiction and this concludes the proof of the first point.

We turn to the second point. The idea is to prove the existence of a net point $x$ of level $\level(x_i^l)+1$ such that the ball $B(x,2^{-\level(x)})$ is a child of $\sigma_i^l$, and such that $\val(B(x,2^{-\level(x)}))$ is greater than $\val(\sigma_{i'}^1)$. Given that the algorithm selects a child with the maximum value up to $\theta$, this will allow us to conclude.

More precisely, we start by deriving a bound on $\dist(x_i^l,x_{i'}^1)$. Applying \Cref{lem:20diam} we get $\dist(x_{i'}^{l'},x_{i'}^1) \leq 20 \cdot 2^{-\level(x_{i'}^1)}$. By the fist point of the Lemma, we know that $\level(x_{i'}^1) \geq \level(x_i^l)+2$ and therefore $\dist(x_{i'}^{l'},x_{i'}^1) \leq 5 \cdot 2^{-\level(x_{i}^l)}$. On the other hand we have $\sigma_i^l \cap \sigma_{i'}^{l'} \neq \emptyset$ and so $\dist(x_i^l,x_{i'}^{l'} )\leq 2 \cdot 2^{-\level(x_{i}^l)}$. Using the triangle inequality, we finally get $\dist(x_i^l,x_{i'}^1) \leq \dist(x_i^l,x_{i'}^{l'}) + \dist(x_{i'}^{l'},x_{i'}^1) \leq 7 \cdot  2^{-\level(x_{i}^l)}$.

Now let $x$ be a net point of level $\level(x_i^l)+1$ such that $x_{i'}^1 \in B(x, 2^{-\level(x)}/2)$, such a point exist by the covering property of nets. We have $\dist(x_i^l,x) \leq \dist(x_i^l,x_{i'}^1) + \dist(x_{i'}^1,x) \leq 7 \cdot  2^{-\level(x_{i}^l)} + 2^{-\level(x)}/2 \leq 10 \cdot 2^{-\level(x_{i}^l)}$. Therefore the ball $B(x,2^{-\level(x)})$ is a child of $\sigma_i^l$ and could have been chosen by the algorithm instead of $\sigma_i^{l+1}$. The algorithm selects a child of $\sigma_i^l$ with the maximum value up to $\theta$. Hence $\val(\sigma_i^{l+1}) \geq \val(B(x,2^{-\level(x)}))-\theta$.

To conclude the proof, it just remains to show that $ \val(B(x,2^{-\level(x)}))\geq \val(\sigma_{i'}^1)$. We will show that for any $p \in P\cap\sigma_{i'}^1$, the contribution of $p$ to the value of  $\sigma_{i'}^1$ is lower than its contribution to the value of $B(x,2^{-\level(x)})$. Let $p$ be a point of $P\cap\sigma_{i'}^1$, the contribution of $p$ to the value of $\sigma_{i'}^1$ is $(2^{-\level(x_{i'}^1)} - \dist(p,x_{i'}^1))^z$. We have established that $\level(x_{i'}^1) \geq \level(x_i^l)+2 = \level(x)+1$. Therefore we can bound $2^{-\level(x_{i'}^1)} \leq 2^{-\level(x)}/2$. Moreover, $x_{i'}^1$ is by construction in $B(x, 2^{-\level(x)}/2)$, hence using the triangle inequality we get $\dist(p,x_{i'}^1) \geq \dist(p,x) - \dist(x,x_{i'}^1) \geq \dist(p,x) - 2^{-\level(x)}/2$. We can now bound the contribution of $p$ to the value of $\sigma_{i'}^1$
\begin{align*}
    (2^{-\level(x_{i'}^1)} - \dist(p,x_{i'}^1))^z &\leq (2^{-\level(x)}/2 - (\dist(p,x) - 2^{-\level(x)}/2))^z\\
    &\leq (2^{-\level(x)}-\dist(p,x))^z.
\end{align*}
This is precisely the contribution of $p$ to the value of $B(x,2^{-\level(x)})$. Summing this inequality for all $p \in P \cap \sigma_{i'}^1$, we obtain $\val(B(x,2^{-\level(x)})) \geq \val(\sigma_{i'}^1)$, concluding the proof.
\end{proof}

\begin{proof}[Proof \Cref{lem:pruning}]
    We are now ready to prove \Cref{lem:pruning}. We recall for convenience the procedure to compute the $l_i$'s, described at the beginning of \Cref{sec:boundValues}. Start with $l_i = 1$ for all $i$, and note that with this choice, the first condition is verified by the design of the algorithm: when $\sigma_i^1$ is picked, it maximizes the value up to $\theta$ among the available balls. In particular, $\sigma_i^1$ has a value larger up to $\theta$ than any ball still available at the end of the algorithm.

To enforce the second condition, proceed as follows: as long as we can find $i < i'$ and $l \geq l_i$, $l' \geq l_{i'}$ such that $\sigma_i^{l} \cap \sigma_{i'}^{l'} \neq \emptyset$, update $l_i = l + 1$. The first item of \Cref{lem:pruningOneStep} guarantees that this procedure is well-defined, namely that $\sigma_i^{l}$ is not the last ball of the sequence, and that $\sigma_i^{l+1}$ does indeed exist. Furthermore, the second item of that Lemma ensures that $\val(\sigma_i^{l_i+1}) \geq \val(\sigma_{i'}^1) - \theta \geq \mval-\theta$. Therefore, the first condition remains satisfied after each update.

At each step, one of the $l_i$ gets incremented, so this procedure must terminate because the maximum level is $\lceil \log(n) \rceil$: when it ends, both conditions are satisfied, concluding the proof.
\end{proof}

\subsubsection[Back to the proof of Theorem 3.4]{Back to the proof of \Cref{thm:mpWithError}}
Our goal is to identify $k$ disjoint balls with values greater than those of $B_\gamma$ using \Cref{lem:pruning} and then apply \Cref{lem:cover} to finalize the proof.

\begin{proof}[Proof of \Cref{thm:mpWithError}]
To prove the theorem, we define a function $\phi$ that maps the center of $\Gamma_1$ to balls of $\calB$ such that for all $\gamma \in \Gamma_1$: 
\begin{enumerate}
    \item for all $\gamma' \in \Gamma_1$ with $\gamma \neq \gamma'$, $\phi(\gamma) \cap \phi(\gamma') = \emptyset$,
    \item $\phi(\gamma)$ is not covered by $\Gamma$,
    \item The value of $B_\gamma$ is less than $\val(\phi(\gamma)) + \theta$.
\end{enumerate}
Given such a matching $\phi$, we can conclude as follows. Summing the inequality of the third property of $\phi$ gives 
\begin{align*}
    \sum_{\gamma \in \Gamma_1}\val(B_\gamma) &\leq k\cdot \theta + \sum_{\gamma \in \Gamma_1}\val(\phi(\gamma)).
\end{align*}
Each $\phi(\gamma)$ is not covered by the second property of $\phi$, so we can apply \Cref{lem:cover}. Moreover those balls are disjoint because of the first property of $\phi$ and we obtain
\begin{align*}
    \sum_{\gamma \in \Gamma_1}\val(\phi(\gamma)) &\leq \sum_{\gamma \in \Gamma_1}\cost(\phi(\gamma),\Gamma)\\
    &\leq \cost(\bigcup_{\gamma \in \Gamma_1}\phi(\gamma),\Gamma)\\
    &\leq \cost(P,\Gamma).
\end{align*}
Putting everything together we finally get
$$ \sum_{\gamma \in \Gamma_1}\val(B_\gamma) \leq k\cdot \theta + \cost(P,\Gamma).$$
Combining this with \Cref{lem:redvalue}, this conclude the proof of \Cref{thm:mpWithError}. 

The rest of the proof is dedicated to the construction of the matching $\phi$ with the three desired properties. We construct a more general function $\phi$ that can also map the centers of $\Gamma_0$. The restriction of $\phi$ to $\Gamma_1$ will verify the desired properties. Let $l_i$ be the indices provided by \Cref{lem:pruning}. We have $3$ steps in the construction of $\phi$: 
\begin{itemize}
\item First, for any $i$ such that the last ball of the sequence $\sigma_i$ is covered, we let $B=(c_i,2^{-\lceil \log n \rceil})$ be this last ball (with center at $c_i \in C_k$) and chose arbitrarily $\gamma_i \in \Gamma$ covering $B$. We define $\phi(\gamma_i) = B$. We note that, in this case, $\gamma_i \in \Gamma_0$ -- since $\dist(c_i,\gamma_i) \leq 2^{-\lceil \log n \rceil}$ .
    \item Second, for any $i$ such that at least one of balls of the pruned sequence $(\sigma_i^{l})_{l\geq l_i}$ is covered but not the last one. We define $\lambda_i \geq l_i$ to be the smallest index such that for all $l\geq \lambda_i$ $\sigma_i^{l}$ is not covered. Let $\gamma_i$ be an arbitrary element of $\Gamma$ that covers $\sigma_i^{\lambda_i-1}$. We define $\phi(\gamma_i) = \sigma_i^{\lambda_i}$. 
    \item Last, for any element in $\Gamma_1$ that is still unmatched, we extend $\phi$ to form an arbitrary one-to-one matching between the remaining elements $\gamma \in \Gamma_1$ and all $\sigma_i^{l_i}$, where $i$ is such that none of the balls in the pruned sequence $(\sigma_i^{l})_{l\geq l_i}$ are covered.
\end{itemize}

Note that the second item of Lemma \ref{lem:pruning} guarantees that if $\gamma\in \Gamma$ covers a ball of a pruned sequence, it cannot cover a ball of another pruned sequence (as those balls are disjoints): this ensures that our definition of $\phi$ is consistent and that $\phi$ is one-to-one. We can now verify it satisfies the three desired properties.

\textbf{Property 1.} For any $\gamma,\gamma' \in \Gamma_1$, let $i$ such that $\phi(\gamma)$ is a ball of the pruned sequence $(\sigma_i^{l})_{l\geq l_i}$, and let $i'$ such that $\phi(\gamma')$ is a ball of the pruned sequence $(\sigma_{i'}^{l})_{l\geq l_{i'}}$. By construction of $\phi$ we have $i \neq i'$ and therefore \Cref{lem:pruning} ensures that $\phi(\gamma) \cap \phi(\gamma') = \emptyset$.   

\textbf{Property 2.} For any $\gamma\in \Gamma_1$, $\phi(\gamma)$ is not covered by $\Gamma$ by construction.

 \textbf{Property 3.} Fix a $\gamma \in \Gamma_1$. We distinguish two cases, based on whether $\phi(\gamma)$ was defined at the second or last step of the procedure (it cannot be defined at the first, as $\gamma$ involved there are in $\Gamma_0$). We aim at showing $\val(B_\gamma) \leq \val(\phi(\gamma)) + \theta$.
    
 \begin{itemize}
     \item If $\phi(\gamma)$ is defined in the last step, then it is of the form $\sigma_i^{l_i}$, and \Cref{lem:pruning} ensures that $\val(\sigma_i^{l_i}) \geq \mval - \theta$. 
     Combined with the fact that $B_\gamma$ is available at the end of the algorithm (and therefore by definition of $\mval$, $\val(B_\gamma) \leq \mval$), we obtain directly $\val(B_\gamma) \leq \val(\phi(\gamma)) + \theta$.
     \item 
Otherwise, $\phi(\gamma)$ is defined in the second step, and $\phi(\gamma) = \sigma_i^{\lambda_i}$. The proof follows a structure similar to the one of \Cref{lem:pruning}. 
We will show that there exists a ball $B(x,2^{-\level(x_i^{\lambda_i})})$ in $\calB$, that contains $B_\gamma$ and is a child of $\sigma_i^{\lambda_i-1}$.  In that case, since the algorithm opted for $\sigma_i^{\lambda_i}$ over $B(x,2^{-\level(x_i^{\lambda_i})})$, we get $\val(\sigma_i^{\lambda_i}) \geq \val(B(x,2^{-\level(x_i^{\lambda_i})})) - \theta$. Since this ball contains $B_\gamma$, it has greater value, which concludes the proof of Property 3. Therefore, we only need to show the existence of such a ball. Note that we don't actually need that $B_\gamma$ is contained in it, only that its value is greater than $\val(B_\gamma)$: we will only show this simpler property.

Let $x_\gamma$ be the center of $B_\gamma$. We start by proving that $\level(x_\gamma) \geq \level(x_i^{\lambda_i})+1$. Assume by contradiction that $\level(x_\gamma) \leq \level(x_i^{\lambda_i}).$ We are going to prove that in that case, $x_\gamma$ is too close to the center $c_i$ and therefore is not available at the end of the algorithm, contradicting the definition of $B_\gamma$.

First, by definition of $B_\gamma$ we know that $\dist(x_\gamma, \gamma) \leq 2^{-\level(x_\gamma)}/2$. Second, since $\gamma$ is an element covering $\sigma_i^{\lambda_i-1}$, it holds that $\dist(\gamma,x_i^{\lambda_i-1}) \leq 2^{-\level(x_i^{\lambda_i-1})} = 2 \cdot 2^{-\level(x_i^{\lambda_i})} \leq 2 \cdot 2^{-\level(x_\gamma)}.$ Third applying \Cref{lem:20diam} we get $\dist(x_i^{\lambda_i-1},c_i) \leq 20 \cdot 2^{-\level(x_i^{\lambda_i-1})} \leq 40 \cdot 2^{-\level(x_\gamma)}$. Combining these three inequalities using the triangle inequality we obtain
\begin{align*}
    \dist(x_\gamma,c_i) &\leq \dist(x_\gamma, \gamma) + \dist(\gamma,x_i^{\lambda_i-1}) + \dist(x_i^{\lambda_i-1},c_i)\\
    &\leq (0.5 + 2 + 40) \cdot 2^{-\level(x_\gamma)}\\
    &\leq 100 \cdot 2^{-\level(x_\gamma)}.
\end{align*}
Therefore $B_\gamma$ is not available at the end of the algorithm, contradicting the definition of $B_\gamma$. This finalize the proof of the inequality $\level(x_\gamma) \geq \level(x_i^{\lambda_i})+1$.

We now identify the desired child of $\sigma_i^{\lambda_i-1}$. Let $x$ be a net point of level $\level(x_i^{\lambda_i})$ such that $\dist(x,x_\gamma) \leq 2^{-\level(x_i^{\lambda_i})}/2$ (such a point exist by the covering property of nets).
We show that $B(x,2^{-\level(x)})$ is a child of $\sigma_i^{\lambda_i-1}$.
Since $\gamma$ is covering $\sigma_i^{\lambda_i-1}$, it holds that $\dist(x_i^{\lambda_i-1},\gamma) \leq 2 \cdot 2^{-\level(x_i^{\lambda_i})}$. Additionally, by definition of $B_\gamma$ we know that $\dist(\gamma, x_\gamma) \leq 2^{-\level(x_\gamma)}/2 \leq 2^{-\level(x_i^{\lambda_i})}/4$. Combining these three inequalities, and using the triangle inequality, we obtain

\begin{align*}
    \dist(x_i^{\lambda_i-1},x) &\leq \dist(x_i^{\lambda_i-1},\gamma) + \dist(\gamma, x_\gamma) + \dist(x_\gamma,x)\\
    &\leq (2 + 0.25 + 0.5)\cdot 2^{-\level(x_i^{\lambda_i})}\\
    &\leq 20 \cdot 2^{-\level(x_i^{\lambda_i})}\\
    &\leq 10 \cdot 2^{-\level(x_i^{\lambda_i-1})}.
\end{align*}

Therefore the ball $B(x,2^{-\level(x)})$ is a child of $\sigma_i^{\lambda_i-1}$ and could have been chosen by the algorithm instead of $\sigma_i^{\lambda_i}$. The algorithm selects a child of $\sigma_i^l$ with the maximum value up to $\theta$. Hence $\val(\sigma_i^{\lambda_i}) \geq \val(B(x,2^{-\level(x)}))-\theta$.

To conclude the proof, it just remains to show that $ \val(B(x,2^{-\level(x)}))\geq \val(B_\gamma)$. We will show that for any $p \in P\cap B_\gamma$, the contribution of $p$ to the value of  $B_\gamma$ is lower than its contribution to the value of $B(x,2^{-\level(x)})$. Let $p$ be a point of $P\cap B_\gamma$, the contribution of $p$ to the value of $B_\gamma$ is $(2^{-\level(x_\gamma)} - \dist(p,x_\gamma))^z$. We have established that $\level(x_\gamma) \geq \level(x_i^{\lambda_i-1})+2 = \level(x)+1$. Therefore we can bound $2^{-\level(x_\gamma)} \leq 2^{-\level(x)}/2$. Moreover, $x_\gamma$ is by construction in $B(x, 2^{-\level(x)}/2)$, hence using the triangle inequality we get $\dist(p,x_\gamma) \geq \dist(p,x) - \dist(x,x_\gamma) \geq \dist(p,x) - 2^{-\level(x)}/2$. We can now bound the contribution of $p$ to the value of $B_\gamma$
\begin{align*}
    (2^{-\level(x_\gamma)} - \dist(p,x_\gamma))^z &\leq (2^{-\level(x)}/2 - (\dist(p,x) - 2^{-\level(x)}/2))^z\\
    &\leq (2^{-\level(x)}-\dist(p,x))^z.
\end{align*}
This is precisely the contribution of $p$ to the value of $B(x,2^{-\level(x)})$. Summing this inequality for all $p \in P \cap B_\gamma$, we obtain $\val(B(x,2^{-\level(x)})) \geq \val(B_\gamma)$, concluding the proof. 
 \end{itemize}

\end{proof}

\section[Missing Proofs of Section 4]{Missing Proofs of \Cref{sec:approx}} \label{app:approx}

\subsection{Direct application to centralized DP}
As an illustration, we sketch an application of previous results to the centralized DP setting. A more formal and general argument will be given later in \Cref{thm:mainApprox}.
In centralized DP, when each element is part of at most $b$ sets, standard result show the existence of generalized histograms with additive  error $O\lpar \sqrt{bD} \cdot \log(m/(\delta \beta))/\eps\rpar$. To compute the values of each cell, $D=1$, there are $m = n^{O(d)}$ many cells, and $b = 2^{O(d)} \log n$. Therefore, those results and \Cref{thm:mpWithError} yield an $(\eps, \delta)$-DP algorithm $\calA$ for $(k,z)$-clustering with multiplicative approximation $O(1)$ and additive error $O\lpar k2^{O(d)} \cdot \log(n/(\beta\delta)) \log^{1/2}(n)/\eps\rpar $, with probability $1-\beta$. 

This can be combined with the results of \Cref{sec:coating} to get multiplicative approximation $w^*(1+\alpha)$ and additive error $k^{O(1)} \log^{3/2}(n)/\eps + O\lpar k\sqrt{d} \log(k)/\eps\rpar $ as follows: first reduce the dimension to $O\lpar\log(k)/\alpha^2\rpar$, then use algorithm $\calA$ and the Laplace mechanism to apply \Cref{lem:boostApprox}. To lift the result back up to the original space, estimate the mean of each cluster as follow:  simply estimate the size of each cluster ($D=1, b=1, m=k$) and $\sum_{p\in P_i} p$ with a Laplace mechanism (the $\ell_1$ sensitivity is $1$ in the first case, $\sqrt d$ in the second): \Cref{assumption:coating} is satisfied with $e = \sqrt{d}/\eps$. \Cref{lem:liftingViaHist} shows that this gives a solution with the desired cost. 
Finally, repeat $\log(1/\beta)$ times to boost the probability, which increases the additive error by $\log(1/\beta)$.  

For $(k,z)$-clustering, instead of lifting the clustering using the average of each cluster, we can solve the $(1,z)$-clustering problem on each, which can be done efficiently with additive error $\sqrt{d} \polylog(n/\delta) $ (see \cite{ghaziTight}). This therefore reduces the multiplicative approximation to $w^*(1+\alpha)$, as in the $k$-means case.

\subsection{Structured Clusters}
We start by providing a formal statement for \Cref{lem:structClusters}:
\begin{lemma}\label{lem:structClustersFormal}
    There exists a family of set $\calG = \lbra G_1,..., G_m\rbra$ efficiently computable, with $m = n^{O(d)}$, and with the following properties. First, any point from $B_d(0,1)$ is part of at most $O(\log n)$ sets $G_i$. 
Second, each cluster can be transformed to consist of the union of few $G_i$. 
Formally, fix any $1>\alpha > 0$, and let $\calC = \lbra c_1, ..., c_k\rbra$ be any set of centers. 
Then, it is possible to compute a partition $\calA$ of $B_d(0, 1)$ together with an assignment $a : \calA \rightarrow \calC$ of parts to centers, such that (1) each part of $\calA$ is a set from $\calG$, (2) $|\calA| \leq k\cdot \alpha^{-O(d)} \log(n)$, and (3) for any multiset $P$,
\[\left|\cost(P, \calC) - \sum_{A \in \calA} \sum_{p \in P \cap A} \dist(p, a(A)) \right| \leq \alpha\cost(P, \calC) + \frac{9}{\alpha}.\]
\end{lemma}
\begin{proof}
To define the sets $G_1, ... G_m$, we make use of a 
\emph{hierarchical decomposition} $\calG = \lbra \calG_1,..., \calG_{\lceil \log (n)\rceil}\rbra$ with the following properties: 
\begin{enumerate}
    \item for all  $i = 0, ..., \lceil \log (n)\rceil$, $\calG_i$ is a partition of $B_d(0,1)$. Each part $A$ of $\calG_i$ is called a cell of level $i$, denoted $\level(A) = i$, and has diameter at most $2^{-i}$.
    \item The partition $\calG_{i+1}$ is a refinement of the partition $\calG_i$, namely every part of $\calG_{i+1}$ is strictly contained in one part of $\calG_i$. $\calG_0$ is defined as $\lbra B_d(0, 1)\rbra$.
    \item for any point $p \in B_d(0, 1)$, any level $i$ and any $r \geq 1$, the ball $B_d(p, r 2^{-i})$ intersects at most $r^{O(d)}$ many cells of level $i$.
\end{enumerate}

Such a recursive decomposition can be computed using e.g. the net tree of \cite{Har-PeledM06}. This construction ensures $|\calG_i| = 2^{-O(id)}$. We let $\calG = \cup \calG_i$: this has size $n^{O(d)}$, as desired. For the first property, since the depth of the decomposition is $ \lceil \log (n)\rceil +1$ and each level is a partition, each point from $B_d(0,1)$ appears in exactly $ \lceil \log (n)\rceil +1$ many $G_i$. 

We now show the second property. Fix a set of centers $\calC = \lbra c_1, ..., c_k\rbra$. We show how to construct a partition $\calA$ of $B_d(0,1)$ and assign each part to a center as in the lemma.

For a cell $A$ of level $i+1$, we denote by $\parent(A)$  the unique cell of level $i$ containing $A$. 
    Let $\ell = \lceil 10/\alpha \rceil$. For any cell $A$ of the decomposition, fix an arbitrary point $v_A \in A$. We call the $\ell$-neighborhood of $A$ all cells at the same level as $A$ that intersect with the ball $B\lpar v_A, \ell 2^{-\level(A)}\rpar$. Note that this includes $A$.

    Let $\calA$ be the set of cells $A$ such that their $\ell$-neighborhood does not contain a center of $\calC$, but the $\ell$-neighborhood of $\parent(A)$  does contain one. Add furthermore to $\calA$ the cells at level $\lceil \log (n) \rceil$ that contain a center in their $\ell$-neighborhood.
We verify now that $\calA$ has the desired properties: first, that its size is bounded; second, it is a partition; and finally, that each cell can be fully assigned to a center while preserving the cost.

    \textbf{Size of $\calA$.} We first bound the size of the set $\calA$ obtained. For this, we will count for each level how many cells have a center in their $\ell$-neighborhood, and how many children each cell has. The product of those two quantities is an upper bound on $|\calA|$.
    First, for a fixed level $i$ and center $c^*$, the decomposition ensures that there are at most $\ell^{O(d)}$ many cells intersecting $B_d(c^*,  2\ell 2^{-i})$. Note that any cell $A$ that contains $c^*$ in its $\ell$-neighborhood must intersect $B_d(c^*,  2\ell 2^{-i})$: indeed, if $B$ is the cell containing $c^*$ this means there is a point in $B$ at distance at most $\ell 2^{-i}$ of $v_A$. Since the diameter of $B$ is at most $2^{-i}$, $c^*$ is at distance at most $\ell 2^{-i} + 2^{-i}$ of $v_A$, and therefore $A$ intersects with the ball $B_d(c^*,  2\ell 2^{-i})$. 
    
    Now, using the third property of decomposition, any cell $A$ of level $i$ is the parent of at most $2^{O(d)}$ many cells. Indeed, all those cells are included in $A$: this means they are contained in a ball of radius $2^{-i}$, and since they are at level $i+1$, third property ensures that there are at most $2^{O(d)}$ many of them.
    Therefore, at any level, at most $k\cdot \ell^{O(d)} = k \cdot \alpha^{-O(d)}$  cells $A$ are added, as $N^\ell(\parent(A))$ must contain one center. Hence, $|\calA| \leq k \cdot \alpha^{-O(d)} \log(n)$.

    \textbf{$\calA$ is a partition of $B_d(0, 1)$.} Now, we show that $\calA$ is a partition of $B_d(0, 1)$. For this, we observe that if a cell $A$ contains a center in its $\ell$-neighborhood, then $\parent(A)$ also contains one.
    First, we argue that the cells in $\calA$ cover the whole unit ball. For a point $p$ in the ball, consider the sequence of parts containing it. There are two cases: Either the smallest part at level $\lceil \log (n) \rceil$ contains a center in its $\ell$-neighborhood, and, thus, it is added to $\calA$; or it does not, and then the largest part of the sequence that does not contain a center in its $\ell$-neighborhood is added to $\calA$. Note that such a largest part must exist, as the cell of level $0$ contains all centers.

    Next we argue that all  cells 
    in $\calA$ are disjoint. Consider two  intersecting cells $A$ and $B$ in $\calA$. 
    By property 2 of the decomposition, it must be that one is included in the other: assume wlog that $A \subset B$ (therefore $\level(A) > \level(B)$). 
    Note that if $A$ is on level $\log (n)$ and has a center in its $\ell$-neighborhood, then trivially $N^\ell(\parent(A))$ contains a center.
    Thus, for all $A \in \cal A$ it holds that $N^\ell(\parent(A))$ contains a center.
    But $\parent(A) \subseteq B$: therefore, $N^\ell(B)$ contains a center, which implies that $B \notin \calA$. 
    Thus, $A$ and $B$ cannot be both in $\calA$, which concludes the proof that $\calA$ is a partition of $B_d(0, 1)$. 

    \textbf{Assigning cells to center.} We can now define $c$ as follows: for each cell $A \in \calA$, we define $a(A) = \argmin_{i} \dist(v_A, c_i)$ to be the closest point from $\calC$ to $v_A$.

    We first define an assignment for a cell  $A \in \calA$ at level $\log (n)$. In this case, triangle inequality ensures that paying the diameter of $A$ for each point allows to serve all points in $A$ with the same center. Formally, let $a(A)$ be the closest center to $v_A$. For any $p\in A$, we have using \Cref{lem:weaktri}:
    \begin{align*}
        \dist(p, a(A))^2 &\leq (\dist(p, v_A) + \dist(v_A, a(A)))^2 \leq (2\dist(p, v_A) + \dist(p, \calC))^2\\
        &\leq (1+\alpha/2)\dist(p, \calC)^2 + (1+\frac{2}{\alpha})\cdot 4\dist(p, v_A)^2 \\
        &\leq (1+\alpha/2)\dist(p, \calC)^2 +\frac{9}{\alpha}\cdot \frac{1}{n}.
    \end{align*}
    Therefore, summing all points in cells at level $\log (n)$ gives an additive error at most $\frac{9}{\alpha}$.

    Now, fix a cell $A \in \calA$ at level $> \log (n)$, and a point $p \in A$. Since, by construction of $\calA$, the $\ell$-neighborhood of $A$ contains at most one center from $\calC$, it holds that:
    \begin{itemize}
        \item either $p$ is already assigned to $a(A)$, and then $\dist(p, a(A)) = \dist(p, \calC)$,
        \item or the center serving $p$ in $\calC$ is outside of the $\ell$-neighborhood of $A$: in particular, $\dist(p, \calC) \geq \ell \cdot \dist(p, v_A)$. In that case, we use the modified triangle inequality from \Cref{lem:weaktri}. This yields similarly as above:
    \begin{align*}
        \dist(p, a(A))^2 &\leq (\dist(p, v_A) + \dist(v_A, a(A)))^2 \leq (2\dist(p, v_A) + \dist(p, \calC))^2\\
        &\leq (1+\alpha/2)\dist(p, \calC)^2 + (1+\frac{2}{\alpha})\cdot 4\dist(p, v_A)^2 \\
        &\leq (1+\alpha/2)\dist(p, \calC)^2 + (1+\frac{2}{\alpha})\cdot 4\cdot \frac{1}{\ell^2} \dist(p, \calC)^2\\
        &= (1+\alpha/2)\dist(p, \calC)^2+ (4 + \frac{8}{\alpha}) \cdot \frac{\alpha^2}{100} \cdot \dist(p, \calC)^2 \\
        &\leq (1+\alpha)\dist(p, \calC)^2.
    \end{align*}
    \end{itemize}
    Summing over all $p\in P$ and combining with the cells at level $\log(n)$, we conclude that:
    \begin{align*}
        \sum_{A \in \calA} \sum_{p \in P \cap A} \dist(p, a(A))^2 \leq (1+\alpha) \cost(P, \calC) +\frac{9}{\alpha}.
    \end{align*}
    The other direction is straightforward: since $a(A) \in \calC$, we have by definition for any $p\in A$ $\dist(p, \calC) \leq \dist(p, a(A))$, and therefore $\cost(P, \calC) \leq \sum_{A \in \calA} \sum_{p \in P \cap A} \dist(p, a(A))^2$.
\end{proof}

\subsection{Main Approximation Result}

\begin{lemma}\label{lem:mainApproxFormal}
        Fix a privacy model where there exist private generalized bucketized vector summation such that with probability $2/3$ the additive error is $A(m,b, D)$, where $m$ is the number of sets, $b$ the maximal number of set containing any given element, and $D$ the dimension of the image of the $f_i$. 
    Assume that the error $A$ is non-decreasing.

    Then, there is a private algorithm for $k$-means  of points in $\Rd$ that, with  probability $1-\beta$, computes a solution with multiplicative approximation $w^*(1+\alpha)$ and additive error $\lpar A(m,  k^{O_\alpha(1)} b, 1) + A(m, b, d)\rpar \cdot k^{O_\alpha(1)} \polylog(n) \log(1/\beta)$, with $m = n^{O_\alpha(\log(k))}$ and $b = \log(n)$. 

    There is also a private algorithm for $(k,z)$-clustering of points in $\Rd$ that,  with  probability $2/3$, computes a solution with multiplicative approximation $w^*(2^z+\alpha)$ and additive error $\lpar A(m,  k^{O_\alpha(1)} b, 1) + A(m, b, d)\rpar \cdot k^{O_\alpha(1)} \polylog(n)$. 
\end{lemma}
\begin{proof}
    This theorem combines all previous results. 
    We first describe and analyze our private algorithm designed for low-dimensions $\hd = O_\alpha(\log k)$.
    To use \Cref{alg:mp}, one needs to estimate the value of each balls. There are at most $m_1 = n^{O(\hd)}$ of them, and each input point is in $b_1 = 2^{O(\hd)} \log n$ many balls. This can be done with a summation algorithm, with additive error $A(m_1, b_1, 1)$ on the values.
    Therefore, \Cref{thm:mp1} shows that this implementation of \Cref{alg:mp} computes a solution with multiplicative approximation $O(1)$ and additive error $k A(m_1, b_1, 1) $.

    Then, one can apply \cref{lem:structClusters} to get structured cluster, with same multiplicative approximation, and additive error increased by $O(1/\alpha)$. 
    We have the following guarantee: each cluster is described with $k \alpha^{-\hd}\log(n)$ many sets of $\calG$, each element appears in at most $\log n$ sets of $\calG$, and $|\calG|= n^{O_\alpha(\log k)}$.  
    Therefore, using the summation algorithm, one can compute the size of each cluster 
    up to an additive error $e_1 = A(|\calG|, \log n, 1) \cdot \alpha^{-\hd} \log(n)$.
    Similarly, one can compute $\sumEst_i$ for each cluster $i$ up to an additive error $A(|\calG|, \log n, d)$. Therefore, \Cref{assumption:coating} is satisfied with an error $e= A(|\calG|, \log n, d)$.
    
    This is enough to boost the approximation ratio using \Cref{lem:boostApprox}, and get a multiplicative approximation $w^*$ in the space $\hRd$. The additive error increases to $k' A(m_1, b_1, 1)  + ke_1 $, for $k' = k^{O(1)} \log(n)$.

    Now, we show how to use this algorithm as a subroutine, to solve the problem in large dimension $d$.
    The algorithm starts by applying the dimension reduction of \Cref{lem:liftingViaHist}, to reduce the dimension to $\hd = O_\alpha(\log k)$. Then, it computes a clustering using the previous argument, and turns the result into a structured partition using \Cref{lem:structClustersFormal}. Using  generalized summation algorithms, the algorithm computes $n_i, \sumEst_i$ and $\sumNorm_i$ that satisfy \Cref{assumption:coating}.
    The additive error in the estimation is $e  := A(|\calG|, \log n, 1)  \cdot \alpha^{-\hd} \polylog(n)$ (note crucially that the exponent of $\alpha$ is still $\hd$, as the clusters are structured in $\hRd$).

    We conclude the theorem for $k$-means by plugging those estimations in  \Cref{lem:liftingViaHist} for dimension reduction and \Cref{cor:boostProba} to boost the probabilities: with probability $1-\beta$, the total additive error is 
    \begin{align*}
    &\underbrace{k' A(m_1, b_1, 1) }_{\text{Compute initial solution}} + \underbrace{ A(|\calG|, \log n, 1) \cdot k \cdot \alpha^{-\hd} \log(n) }_{\text{Boost approx.}} + \underbrace{A(|\calG|, \log n, d)  \cdot \alpha^{-\hd} \polylog(n)}_{\text{Lift up the results and boost proba.}}\\
        &\leq \lpar A(n^{O_\alpha(\log k)}, k^{O_\alpha(1)} \log(n), 1)  + A(n^{O_\alpha(\log k)}, \log n, 1) + A(n^{O_\alpha(\log k)}, \log n, d) \rpar \cdot k^{O_\alpha(1)} \polylog(n) \\
        &\leq \lpar A(n^{O_\alpha(\log k)},  k^{O_\alpha(1)} \log n, 1) + A(n^{O_\alpha(\log k)}, \log n, d)\rpar \cdot k^{O_\alpha(1)} \polylog(n) .
    \end{align*}

    For $(k,z)$-clustering, the same holds, although we cannot boost the success probability, and the multiplicative approximation is $(2^z+\alpha)w^*$. 
\end{proof}

\mainResApprox*
\begin{proof}
    The guarantee for the local model stem from Lemma 28 in \cite{ChangG0M21} combined with \Cref{lem:buildGenHist}, which gives a generalized bucketized summation algorithm with desired guarantee (even for $\delta = 0$): the additive error is with probability $2/3$ $A(m,b,D) = \frac{b\sqrt{nD\log(m)}}{\eps}$.
    For the $(\eps, \delta)$-shuffle model, this is Theorem 33 of \cite{ChangG0M21}: the additive error is with probability $2/3$ $A(m,b,D) = \frac{b\sqrt{D}}{\eps} \cdot   \polylog(\frac{mD}{\delta})$ (with an unspecified $\polylog$).

    Under continual observation, we give in \Cref{lem:histoContinual} a generalized counting mechanism that has with probability $2/3$, at all time step simultaneously, an additive error $O\lpar b\sqrt{D} \cdot \eps^{-1} \log (DT)\sqrt{\log (mDT/\beta) \log(b \log(T)/\delta)}\rpar$. Plugging in the values of $b, m$ and using \Cref{thm:mainApprox} concludes.
\end{proof}

\section[Missing Proofs of Section 5]{Missing Proofs of \Cref{sec:error}}

\subsection{Centralized DP}
\begin{restatable}{lemma}{centralizedError}\label{lem:centralizedError}
    \Cref{alg:mp} implemented with the exponential mechanism computes a solution to $(k,z)$-clustering with multiplicative approximation $O(1)$ and additive error $\frac{k\sqrt{d}\polylog(n/\delta)}{\eps'}$. 
    Furthermore, the running time is $k^{O(1)} \cdot n \log^2 n$.
\end{restatable}
\begin{proof}
    For the initial choice of the sequence (line 4 of \Cref{alg:mp}), the exponential mechanism takes a decision out of $n^{O(d)}$ many balls: therefore, the value of the chosen ball deviates from the optimum by an additive error $O\lpar \frac{d \log(n/\beta) }{\eps'}\rpar$, with probability $1-\beta$.
    Each choice made in the while loop line 5 takes a decision out of the children of the current ball, and there are at most $2^{O(d)}$ of them: therefore, the additive error is $O\lpar \frac{d \log(1/\beta) }{\eps'}\rpar$, with probability $1-\beta$.

    In total, $k \log n$ choices are made: therefore, taking $\beta = 1/(3k \log n)$ ensures that all errors are bounded by $O\lpar \frac{d \log(n/\beta) }{\eps'}\rpar$ with probability $2/3$. Combined with \cref{thm:mpWithError}, this shows that the additive error of \Cref{alg:mp} is $O\lpar \frac{kd \log(n/\beta) }{\eps'}\rpar$. 

    Dimension reduction reduces $d$ to $\hd = O(\log(k))$ while preserving a multiplicative approximation $O(1)$, and the additive error induced by lifting the clustering up in the original space with \Cref{lem:liftingViaHist} is $O\lpar \frac{\polylog(n) k \sqrt{d}}{\eps'}\rpar$. 
    Indeed, one can estimate the size of each cluster and $\sum_{p\in P_i} p$ with the Gaussian mechanism: in the both cases, the $\ell_2$ sensitivity is $O(1)$. The Gaussian mechanism ensures that, for each of those estimates, the additive error is bounded by $\sqrt{d}/\eps'  \log(1/\beta)$ with probability $1-\beta$. Therefore, a union-bound over all $2k$ estimates shows that \Cref{assumption:coating} is satisfied with $e = \sqrt{d} \log(k)/ \eps'$. 
    This concludes the approximation guarantee.

    The running time of a naive implementation of this algorithm is $n^{\hd}$, to run the exponential mechanism. However, note that only $2^{\hd} n \log n$ balls are non-empty (since at each level, each point is in at most $2^{\hd}$ many balls). All empty balls have same value, equal to $0$. Therefore, the exponential mechanism can be implemented in time $2^{\hd} n \log n = k^{O(1)} n \log n$. As the algorithm makes $O(k \log n)$ calls to the exponential mechanism, this concludes the complexity analysis.
\end{proof}

To show that the procedure is $(\eps, \delta)$-DP, we follow the analysis of \cite{GuptaLMRT10} for weighted set cover (adapted by \cite{chaturvediCentral} for $k$-means): in our language, they analyzed the exponential mechanism for a process that selects iteratively balls with largest value, and then remove all points of the selected ball. Our process is more intricate, as it combines selecting the heads with constructing the sequence. 
We manage nonetheless to adapt their proof in the following lemma:
\begin{lemma}\label{lem:central-exp-priv}
\Cref{alg:mp}  when implemented using the exponential mechanism with parameter $\eps' = \frac{\eps}{4\log(n/\delta)}$, is $(\eps, \delta)$-DP.
\end{lemma}
\begin{proof}
We follow the proof of \cite{GuptaLMRT10} for set cover. Our two step process incurs some complication: analyzing the selection of the first ball of each sequence $\sigma_i$ is similar to \cite{GuptaLMRT10}, but in our case it is interlaced with the recursive greedy choices, which makes thing more technical.

We consider the outcome of the algorithm to be the $k$ sequences of balls $\sigma_1, ..., \sigma_k$ (instead of merely the $k$ centers), and we will show that this is $(\eps, \delta)$-DP. For this, we will fix a given set of sequences $C = (\sigma_1, ..., \sigma_k)$ and will compare the probability that the outcome of $\calA$ is $C$ on two neighboring inputs $X$ and $X'$. 
Note that fixing a sequence $C$ also fixes the center chosen by the algorithm: the $i$-th center is the (geometric) center of the last ball in $\sigma_i$. 

We start with some notations. 
We write $\calA(X)_i^j$ to be the $j$-th choice made by the algorithm in the $i$-th sequence, on input $X$.
We denote $\{\sigma, < i, < j\}$ the event $\calA(X)_{i'}^{j'} = \sigma_i^j$ for all $i', j'$ such that either $i' < i$ or $i' = i$ and $j'<j$. We write $\{\sigma, \leq i\}$ for $\{\sigma, < i, < \infty\}$.
Furthermore, for a ball $B$, we write $\children(B)$ the children of $B$ (see \Cref{sec:MPalgo} for the definition of children, forbidden and available)

For any ball $B \in \calB$ and dataset $X$, we write $\val_i(B, X) = -\infty$ when $B$ is forbidden by some center among the first $i-1$ of $C$; and $\val_i(B, X) := \val(B, X)$ otherwise (when $B$ is still available after choosing the first $i-1$ centers from $C$).

Let $X$ and $X'$ be two neighboring datasets, with symmetric difference $X \Delta X' = \{p\}$. We aim at bounding $\frac{\Pr[\calA(X) = C]}{\Pr[\calA(X') = C]}$. 

    There are two different choices in the algorithm: the choice of the first elements of each sequence, and then the recursive construction of the sequences themselves. 
    
    \begin{align*}
        \Pr\lbrak\calA(X)_i^1 = \sigma_i^1 ~|~ \{\sigma, \leq i\}\rbrak = \frac{\exp(\eps' \val(\sigma_i^1, X))}{\sum_{B \in \calB} \exp(\eps' \val_{i}(B, X))},\\
        \Pr\lbrak\calA(X)_i^j = \sigma_i^j ~|~ \{\sigma, < i, < j\}\rbrak 
        = \frac{\exp(\eps' \val(\sigma_i^j, X))}{\sum_{c \in \children(\sigma_i^{j-1})} \exp(\eps' \val(c, X))}.
    \end{align*}
    
    Note that, if there is $i$ such that $\sigma_{i}^1$ is not available  at the $i$-th step, then $\Pr[\calA(X) = C] = \Pr[\calA(X') = C] = 0$; and similarly  if there are some $i, j$ such that $\sigma_i^j$ is not a children of $\sigma_i^{j-1}$, namely  $\sigma_i^j\notin \children(\sigma_i^{j-1})$. 
    Therefore, we only need to focus on admissible sequences, where the previous cases cannot happen. 
    This ensures that one input point $p$ can appear only in $\log n$ many $\sigma_i^j$ and $C(\sigma_i^{j})$, one per level of the ball hierarchy. 
    Indeed, if $p \in \sigma_i^j$, then all balls containing $p$ at the level of $\sigma_i^j$ and the level below are forbidden by the algorithm, and thus in any subsequent admissible sequence $p$ cannot appear in another ball at those levels;    since there are $\log n$ levels, $p$ appears in at most $\log n$ balls or children of balls in an admissible sequence of balls.
    
    We now analyse the ratio of the probabilities. We write:
    \begin{align}
         \notag 
         &\frac{\Pr[\calA(X) = C]}{\Pr[\calA(X') = C]} \\
        \label{eq:exp1}
        =& \prod_{i=1}^k \frac{\exp(\eps' \val_i(\sigma_i^1, X))}{\exp(\eps' \val(\sigma_i^1, X'))}\\ 
        \label{eq:exp2}
        & \cdot  \prod_{i=1}^k\frac{\sum_{B} \exp(\eps \val_i(B, X'))}{\sum_{B} \exp(\eps \val_i(B, X))}\\
        \label{eq:exp3}
        & \cdot \prod_{i=1}^k \prod_{j = 1}^{\len(\sigma_i)}\frac{\exp(\eps' \val(\sigma_i^j, X))}{\exp(\eps' \val(\sigma_i^j, X'))}\\
        \label{eq:exp4}
        & \cdot \prod_{i=1}^k \prod_{j = 1}^{\len(\sigma_i)}\frac{\sum_{B \in \children(\sigma_i^{j-1})} \exp(\eps' \val(B, X'))}{\sum_{B \in \children(\sigma_i^{j-1})} \exp(\eps' \val(B, X))}
    \end{align}

We start by bounding the two easier \Cref{eq:exp3} and \Cref{eq:exp4}. Those are the standard terms from the exponential mechanism: for twe neighboring dataset differing on a single point $p$, the value of any ball changes by at most $\Delta$, and does so only when $p$ is part of the ball. So, when $p$ is part of the ball $\sigma_i^j$, we have: $\frac{\exp(\eps' \val(\sigma_i^j, X))}{\exp(\eps' \val(\sigma_i^j, X'))} \leq \exp(\eps' \Delta)$, and  when $p$ is part of a ball among $\children(\sigma_i^j)$, we have:
$\frac{\sum_{B \in \children(\sigma_i^{j-1})} \exp(\eps' \val(B, X'))}{\sum_{B \in C(\sigma_i^{j-1})} \exp(\eps' \val(B, X))} \leq \exp(\eps' \Delta)$.

As mentioned before, is an admissible sequence $p$ is part of at most $\log n$ $\sigma_i^j$, and $\log n$ $C(\sigma_i^j)$: therefore, the products of \Cref{eq:exp3} and \Cref{eq:exp4} is bounded by $\exp(2\log n \cdot \eps'\Delta)$.

We can now focus on the two more intricate terms, \Cref{eq:exp1} and \Cref{eq:exp2}.
In the case where $X = X' \cup \{p\}$, then for any ball $B$ and any $i$, $\val_i(B, X') \leq \val_i(B,X)$, and therefore term in \Cref{eq:exp2} is at most one. 
On the other hand, $\val(B, X) \leq \val_i(B, X') + \Delta$ when $B$ contains the point $p$, and $\val_i(B, X) = \val(B, X')$ if $B$ doesn't contain $p$.  Since at most $\log n$ of the $\sigma_i^1$ contain $p$, the term in \Cref{eq:exp1} is at most $\exp(\log n \cdot \eps)$. Therefore, 
\[\Pr[\calA(X) = C] \leq \exp(\log n \cdot \eps) \Pr[\calA(X') = C].\] 

The second case, where $X' = X \cup \{p\}$, is more intricate. There, \Cref{eq:exp1} is at most one, and we bound \Cref{eq:exp2} as follows: first, we write for simplicity $\pi_i(B) := \Pr[\calA(X)_i^j = x~|~ \{\sigma, < i, < j\}] = \frac{\exp(\eps' \val_i(B, X))}{\sum_{B'} \exp(\eps'\val_i(B', X'))}$.

\begin{align*}
    &\frac{\sum_{B'} \exp(\eps \val_i(B', X'))}{\sum_{B'} \exp(\eps' \val_i(B', X))} 
    \\=& \sum_{B'} \pi_i(B') \exp\left(\eps' (\val_i(B', X') - \val_i(B', X))\right)\\
    \leq&  \sum_{B'} \pi_i(B') (1+2\eps' (\val_i(B', X') - \val_i(B', X)) / \Delta)\\
    \leq& 1+ 2\eps'\sum_{B'} \pi_i(B') (\val_i(B', X') - \val_i(B', X))\\
    \leq& \exp\left(2\eps'\sum_{B'} \pi_i(B') (\val_i(B', X') - \val_i(B', X))\right)\\
    \leq& \exp\left(2\eps' \Delta \sum_{B' \text{ s.t. } p \in B'} \pi_i(B')\right).
\end{align*}
The sum in the exponent corresponds to the probability to chose a ball $B'$ that contains $p$, after having made the choices $\{\sigma, < i, < j\}$. We write $\pi_i$ this probability. 

Therefore, combining with the bounds on \Cref{eq:exp3} and \Cref{eq:exp4}, we have:
\begin{align*}
    \frac{\Pr[\calA(x) = C]}{\Pr[\calA(X') = C]} 
    \leq \exp\left(2\eps' \Delta \sum_{i=1}^k \pi_i\right) \cdot \exp\left(2 \log n \cdot \eps' \Delta\right).
\end{align*}

The analysis of \cite{GuptaLMRT10} works black box from this point. They show that the above equations implies that, for any set of possible outcomes $\mathcal{C}$, it holds that
\begin{align*}
    \Pr[\calA(x) \in \mathcal{C}] &\leq \exp(2\eps' \Delta  \log(1/\delta) + 2 \log n \cdot \eps' \Delta)\Pr[\calA(X') = \mathcal{C}] + \delta,
\end{align*}
which, with our choice of $\eps'$, concludes the proof.

\end{proof}

\subsection{Parallel algorithm for MPC}\label{app:MPC}
\MPCMP*
\begin{proof}
The proof of this lemma is grounded in the analysis of \Cref{alg:mp}, which we conducted in \Cref{sec:appMP}. We  define $\calA(c_\calA,C)$,$\Gamma, \Gamma_0$, $\Gamma_1$, and the $B_\gamma$'s for $\gamma \in \Gamma_1$ as detailed in \Cref{sec:appMP}. Note that $\calA(c_\calA,C)$ is the set of balls of $\calB$ available at scale $c_\calA$ from $C$. Using \Cref{lem:inAndout} and \Cref{lem:redvalue}, we can reduce the task of proving \Cref{lem:MPCMP} to bounding $\sum_{\gamma \in \Gamma_1} \val{B_\gamma}$.
Note that the constants appearing in \Cref{lem:inAndout} depend on $c_\calA$, but we consider $c_\calA$ as a fixed constant that does not appear in the $O(\cdot)$ notation. 

In the proof of \Cref{thm:mpWithError}, to $\sum_{\gamma \in \Gamma_1} \val{B_\gamma}$, the proof defines a function $\phi$ that maps the center of $\Gamma_1$ to balls of $\calB$ such that for all $\gamma \in \Gamma_1$:
\begin{enumerate}
    \item for all $\gamma' \in \Gamma_1$ with $\gamma \neq \gamma'$, $\phi(\gamma) \cap \phi(\gamma') = \emptyset$,
    \item $\phi(\gamma)$ is not covered by $\Gamma$,
    \item The value of $B_\gamma$ is less than $\val(\phi(\gamma)) + \theta$.
\end{enumerate}
Given such a mapping, the proof of \Cref{thm:mpWithError} shows $\sum_{\gamma \in \Gamma_1} \val{B_\gamma} \leq \cost(P, \Gamma)$.
The same proof can be applied here: therefore, it is enough to construct the mapping $\phi$. For this, we use the following procedure.

At the beginning of the procedure, the set $\Phi$ consists of all balls that can be matched to a $\gamma$: this is all the balls of the form $B(c,2^{-\ell})$ with $c \in C_\ell$, for all $\ell$. 
Until all $\gamma \in \Gamma_1$ have been matched, repeat the following process. 
Let $\gamma \in \Gamma_1$ be such that the radius of $B_\gamma$ is minimal, among the centers of $\Gamma_1$ that have not yet been matched.
$\gamma$ is matched to an arbitrary ball $B(c,2^{-\ell}) \in \Phi$ (ensuring that the radius of $B_\gamma$ is also $2^{-\ell}$ and $c \in C_\ell$). All balls intersecting $B(c,2^{-\ell})$ or containing $\gamma$ are removed from $\Phi$; repeat the procedure for the next $\gamma$.

We begin by demonstrating that this procedure is well defined, namely that there is always a ball $B(c,2^{-\ell})$ in $\Phi$. 
Once a ball $B(c,2^{-\ell})$ is selected, at most $2$ balls per level $\ell' \leq \ell$ are eliminated. This is due to condition (1) of the lemma, which ensures that the distance between two distinct centers of $C_{\ell'}$ is greater than $3\cdot 2^{-\ell'}$.  This condition has two implications:
first,  the balls $B(c',2^{-\ell'})$ with $c' \in C_{\ell'}$ are disjoint, and therefore at most one contains $\gamma$. 
Second, a single ball $B(c',2^{-\ell'})$ at level $\ell'$ can intersect $B(c,2^{-\ell})$. Indeed, assuming by contradiction that two balls $B(c'_1,2^{-\ell'})$ and $B(c'_2,2^{-\ell'})$ intersect $B(c,2^{-\ell})$, we have, by the triangle inequality, $\dist(c'_1,c'_2) \leq 2\cdot 2^{-\ell'} + 2^{-\ell} \leq 3 \cdot 2^{-\ell'}$ because we assumed $\ell' \leq \ell$. 

The procedure defines $\phi$ by increasing order of the radii of the $B_\gamma$'s. Therefore, when the center $\gamma'$ with $B_{\gamma'} = B(c', 2^{-\ell'})$ is considered in the procedure, all centers $\gamma$ previously matched had $B_{\gamma} = B(c, 2^{-\ell})$ with $\ell' \leq \ell$: the previous discussion implies that each of those removed at most two balls with radius $2^{-\ell'}$ from $\Phi$.

By construction of our algorithm, each $C_\ell$ contains $2k$ centers: since there are $k$ centers in $\Gamma$, the matching $\phi$ is well defined.
We can now verify that this matching satisfies the three desired properties.

\begin{enumerate}
    \item The procedure guarantees that for all $\gamma,\gamma'\in \Gamma_1$, $\phi(\gamma) \cap \phi(\gamma') = \emptyset$ because when a ball is selected, any ball intersecting it is removed from $\Phi$. 
    \item Let $\gamma \in \Gamma$. Assume by contradiction that there exists $\gamma'\in \Gamma$ covering $\phi(\gamma)$, and let $B_\gamma = B(z,2^{-\ell})$, $\phi(\gamma) = B(c,2^{-\ell})$ (the procedure ensures that $\phi(\gamma)$ has same radius as $B_\gamma$), and $B_{\gamma'} = B(z',2^{-\ell'})$. By definition of covering, $\gamma' \in \phi(\gamma)$.
    First, in the case where $\ell' > \ell $, $\gamma'$ is processed before $\gamma$ in the procedure: all the balls containing $\gamma'$ have been removed from $\Phi$ when $\phi(\gamma')$ is defined, and in particular $\phi(\gamma)$. This yields to a contradiction.
    In the other case, when $\ell' \leq \ell $, we note the following inequality: $\dist(c,\gamma') \leq 2^{-\ell} \leq 2^{-\ell'}$. Moreover by definition of $B_{\gamma'}$, $\dist(\gamma',z') \leq 2^{-\ell'}/2$. Using the triangle inequality we obtain $\dist(c,z') \leq 1.5 \cdot 2^{-\ell'} \leq c_\calA \cdot 2^{-\ell'}$, this implies that $B_{\gamma'}$ is not available at scale $c_\calA$ from $C$, which contradicts the definition of $B_{\gamma'}$. Those two cases concludes that there is no $\gamma'$ covering $\phi(\gamma)$.
    \item By definition, $B_\gamma$ is in $\calA(c_\calA,C)$, and the condition $3.$ of the lemma states that the value of the ball $\phi(\gamma) = B(c,2^{-\ell})$ is greater (up to an additive error $\theta$) than the value of any ball of $\calB_\ell$ available at scale $c_\calA$ from $C_\ell$: this set contains $\calA(c_\calA,C)$, and therefore it contains $B_\gamma$, and the third condition holds.
\end{enumerate}
This conclude the proof of \Cref{lem:MPCMP}.
\end{proof}

We describe more formally the greedy algorithm described in the main body. We suppose that machines are organized in a tree structure, where the degree of every node is $n^{\kappa}$ -- so that the depth of the tree is $1/\kappa$. We also suppose the dimension is reduced to $\hd = O(\log(k))$. The algorithm works as in the centralized case: first solve $(k,z)$-clustering in this projected space, then lift-up the centers in $\Rd$.

Initially, each possible ball is assigned randomly to one of the leaf (this can be implemented with a simple hash-function). Using standard MPC procedure, each machine can compute the value of each non-empty ball assigned to it.

For a leaf $M$, let $B(M)$ be the set of non-empty balls assigned to $M$. 
For each round $i = 0, \ldots, 1/\kappa$, the following process happens.
First, each machine $M$ at height $i$ in the tree runs the greedy algorithm with distance parameter $D_i := 2^{-i} \cdot 3 2^{1/\kappa}$. Let $S_M$ be the $2k$ balls selected by $M$: $M$ sends $S_M$ to its parent. 
For each machine at height $i+1$, let $B(M)$ be the set of balls received by $M$. This marks the end of the round.

First, the privacy of this algorithm follows directly from composition: at each level of the tree, properties of the exponential mechanism (or, more precisely, the proof of \Cref{lem:central-exp-priv}) ensure that process is $(\eps, \delta)$-DP, and there are $1/\kappa = O(1)$ many levels.

We verify that this algorithm can indeed be implemented in MPC when machines have memory $k^{O(1)} n^{\kappa}$. Note there are at most $2^{\hd} n = k^{O(1)} n$ non-empty balls at each level: therefore it is possible to represent all non-empty balls in the memory: each of the $n^{1-\kappa}$ machine has memory $k^{O(1)} n^\kappa \polylog(n)$. Assigning uniformly at random the balls to leafs ensures that with high probability they are assigned in a balanced way, and that they fit in memory.

To implement the exponential mechanism, the algorithm do not actually have access to all balls assigned to the machine $M$ -- only the non-empty ones. Instead, the algorithm uses the exponential mechanism on the following set: all possible balls, with value $0$, and all non-empty balls assigned to $M$, with their true value. 
This can be efficiently implemented in time $k^{O(1)} n^\kappa \polylog(n)$, and ensures that with probability $1-\beta'$, the ball selected by one exponential mechanism have maximum value up to $\theta = O\lpar \frac{\hd \log(n/\beta')}{\eps} \rpar$.

The exponential mechanism is used $2k$ times on each machine: choosing $\beta' = \frac{\beta}{2k n^{1-\kappa}}$ allows to do a union bound, and shows that all ball selected have maximum value up to $\theta = O\lpar \frac{\hd \log(n/\beta)}{\eps} \rpar$.

We will condition all our analysis on the fact that the exponential mechanism selects ball with maximal value, up to an additive error $\theta$. To compute $\theta$, note that the exponential mechanism is used $2k$ times in each machine, each time to select one ball out of $n^{\hd}$ many balls. in total it is used $2k n^{1-\kappa}$ times
To analyse this greedy algorithm, we note the following crucial claims:

\begin{claim}
    \label{claim:valIncreases}
    Let $M$ be a machine and $M'$ be its parent. Then, the smallest value in $S_{M'}$ is larger (up to $\theta$) than the smallest value in $S_{M}$.
\end{claim}
\begin{proof}
    The greedy process ensures all balls in $S_{M}$ are at distance at least $2^{-h(M)} c_\calA \cdot 2^{-\ell}$ of each other, where $h(M)$ is the height of $M$ in the tree. Therefore, by triangle inequality every ball selected in $S_{M'}$ forbids at most one ball of $S_{M}$ (due to the exponential decrease of the distance parameter): as a consequence, at any choice made by the greedy on $M'$, at least a ball of $S_M$ is still available. Therefore, the value of selected balls in $S_{M'}$ are only larger than the ones in $S_{M}$ (up to an additive error $\theta$ due to the exponential mechanism).
\end{proof}

\begin{claim}\label{claim:valForbidden}
    Let $B$ be a ball in $\calB(M)$ forbidden by some $B' \in S_{M}$ at scale $\sum_{i \leq h(M)} D_{i}$. Then, $\val(B) \leq \val(B') + h(M) \theta$.
\end{claim}
\begin{proof}
    Let $M_i$ be the machine at height $i$ such that $B \in \calB(M_i)$. Define $B_i$ as follows: 
    \begin{enumerate}
        \item if $B_{i-1} \in S_{M_i}$, $B_i = B_{i-1}$
        \item if $B_{i-1}$ is available at scale $D_{i}$ from $S_{M_i}$: let $B_i$ be an arbitrary ball in $S_{M_i}$
        \item otherwise, $B_{i-1}$ is forbidden by some ball in $S_{M_i}$. Let $B_i$ be the first ball in the greedy process to forbid $B_{i-1}$.
    \end{enumerate}
    Now, we claim that $\val(B_i) \geq \val(B) - i \theta$. This is obvious in the first case. In the second case, this is directly due to the greedy procedure: for each choice of the greedy on $M_i$, the ball $B_{i-1}$ is available, and therefore the ball selected has value at least $\val(B_{i-1}) - \theta$. In the third case, at the moment $B_i$ is chosen by the greedy procedure on $M_i$, $B_{i-1}$ is still available and therefore $\val(B_i) \geq \val(B_{i-1}) - \theta$.

    Since $B' = B_{h(m)}$, this concludes the claim.
\end{proof}

we can now show this implementation of the greedy achieve the guarantees required by \Cref{lem:MPCMP}.
\begin{restatable}{lemma}{mpcGuarantee}\label{lem:mpcGuarantee}
Let $C_\ell$ be the set of balls selected by the merge-and-reduce process described above after $1/\kappa$ rounds. 
The balls of $C_\ell$ are at distance at least $3 \cdot 2^{-\ell}$ from each other, and the value of any ball in $C_\ell$ is larger (up to $\theta/\kappa$) than the value of any ball available at scale $2^{1/\kappa+1} \cdot 3$.
\end{restatable}
\begin{proof}
    The first part of the statement follows directly from the greedy merging: at the root of the tree, the ball selected are at distance $3\cdot 2^{-\ell}$ from each other. For the second, we proceed by induction.
    
    We start with few notations: for a machine $M$, we let $h(M)$ be the height of $M$ in the tree, and $c(M)$ be the set of children of $M$. We define inductively $\calB_M := \cup_{M' \in c(M)} \calB_{M'}$ to be the set of balls covered by machine $M$ (for a leaf $M$ $\calB_M$ be the balls assigned to the leaf $M$ by the random assignment).
    We define $S_M$ be the set of balls selected by the greedy procedure run on $M$.

    Our inductive claim is that the balls selected in $S_M$ have larger values, up to an additive error $\theta h(M)$, than the balls in $\calB_M$ available at scale $\sum_{i \leq h(M)} D_i \cdot 2^{-\ell}$ from $S_M$.
    When $h(M) = 0$, this property stems directly from the greedy algorithm. 

    For $h(M) \geq 1$, we proceed as follows: let $B$ be a ball available at scale $\sum_{i \leq h(M)} D_i \cdot 2^{-\ell}$ from $S_M$. 
    Let $M'$ be the children of $M$ such that $B \in \calB_{M'}$.
    First, in the case where $B \in S_{M'}$, $B$ is part of the greedy process on machine $M$: since it is available and not selected, it has smaller value than any ball selected in $S_M$ up to an additive error $\theta$.²
    Second, in the case where $B \notin S_{M'}$: either $B$ is forbidden at scale $D_{h(M')}$ by some ball $B' \in S_{M'}$, or $B$ is available at scale $D_{h(M')}$ from $S_{M'}$. 
\begin{itemize}
    \item  First, suppose $B$ is forbidden by some $B' \in S_{M'}$ at scale $\sum_{i \leq h(M')} D_{i}$.
    It cannot be that $B' \in S_M$, as otherwise $B$ would be forbidden at scale $\sum_{i \leq h(M)} D_{i}$ from $S_M$. 
    Therefore, either $B'$ is forbidden at scale $D_{h(M)}$ from $S_M$, or it is still available at the end of the greedy on $M$. 
    In the first case, triangle inequality shows that $B$ is at distance at most $\sum_{i \leq h(M)} D_{i}$ of $S_M$, which contradicts the initial condition on $B$. 
    On the second case, if $B'$ is still available at scale $D_{h(M)}$ from $S_{M}$, then the greedy ensures that its value is smaller than that of any ball selected in $S_M$, up to an additive error $\theta$. Hence, combined with \Cref{claim:valForbidden}, this shows that the value of $B$ is smaller than any ball of $S_M$, up to an additive error $h(M) \theta$.
    \item In the case $B$ is available at scale $\sum_{i \leq h(M')} D_{i}$ from $S_{M'}$, then we know by induction that its value is less than any selected ball in $S_{M'}$, up to an additive $h(M') \theta$. As shown in \Cref{claim:valIncreases}, this implies its value is less than any selected ball in $S_M$ as well, up to an additive error $h(M') \theta + \theta = h(M) \theta$, which concludes the inductive claim.
\end{itemize}

Therefore, at height $2^{1/\kappa}$ (i.e., the root of the merge-and-reduce tree), the balls selected have higher values than the one available at scale $\sum_{i \leq 2^{1/\kappa}} D_i 2^{-\ell}$. As $D_i = 2^{1/\kappa - i}$, this is at most $2^{1/\kappa+1} \cdot 2^{-\ell}$, which concludes the lemma.
\end{proof}

Combining those lemmas proves \Cref{thm:mpc}:
\begin{proof}[Proof of \Cref{thm:mpc}]
    \Cref{lem:mpcGuarantee} shows that, at each level, the set $C_\ell$ computed by the algorithm satisfies the conditions of \Cref{lem:MPCMP} (with $c_\calA = 2^{1/\kappa+1} \cdot 3$ and $\theta = O\lpar \frac{\hd \log(n/\beta)}{\eps} \rpar$): therefore, $C = \cup C_\ell$ verifies $\cost(P, C) \leq O(1) \opt_{k,z} + O\lpar k \cdot \frac{\hd \log(n/\beta)}{\eps} \rpar)$.

    Furthermore, $C$ has size $O(k\log n)$: it can be aggregated in a single machine. The techniques from \Cref{sec:bicriteria} allows then to reduce the number of centers selected to $k$, without increasing the error.

    Finally, to lift-up the solution in the original space, we use  \Cref{lem:liftingViaHist}: using noisy average, we can estimate the mean of each cluster up to an additive error $ O\lpar\sqrt{d}/\eps\rpar $, which yields a solution with cost $O(1) \opt_{k,z} + O\lpar \polylog(n) \cdot \frac{k\sqrt{d}}{\eps}\rpar$.

    For the second part of the lemma, we need to boost the approximation ratio using \Cref{sec:boostApprox}. For this, \Cref{lem:boostApprox} directly applies, if we can lift up the solution to the original space. For $k$-means, we can simply compute the noisy average of each cluster -- which gives an extra additive error $O(k\sqrt{d} \log(1/\delta)/\eps)$, as in the proof of \Cref{lem:centralizedError}.

    Recovering the optimal $(1,z)$-clustering in each cluster, using low memory, is more intricate. For this, we apply the techniques presented in Theorem E.6 of \cite{Cohen-AddadEMNZ22}, to get an additive error $(2^{\hd} + k \sqrt{d}) \polylog(n/\delta) / \eps$.
    
    In both cases, this provides a solution with cost $(1+\alpha)w^*\opt_{k,z} + (k^{O_\alpha(1)} + k\sqrt{d})\polylog(n/\delta)/\eps$.
\end{proof}

\section[A note on epsilon-Differential Privacy]{A note on $\eps$-Differential Privacy}\label{app:epsprivacy}

Our results partially extend to pure-DP, i.e. when $\delta = 0$. 
In that case, we can adapt our algorithm for the optimal multiplicative approximation : in most models, there is a general summation algorithm with additive error essentially worsen by a $\sqrt{b}$ factor, compared to $(\eps, \delta)$-DP. 
For centralized DP, the additive error is $b \sqrt{D} \log(m/\beta)/\eps$ -- this is a direct extension of histograms. 
In local DP, the additive error we presented already worked in the case $\delta=0$ -- as the summation result of \cite{ChangG0M21} that we used readily works in this case.

For $\eps$-DP  under continual observation, we can apply the second part of our \Cref{lem:histoContinual} to get that an additive error worsen by a $\sqrt{d \log (T)}$ -- leading to a $(k,z)$-clustering with additive error $d k^{O_\alpha(1)} \log^3(n) \log^2(T) \log(1/\beta) $.

\end{document}